\documentclass[letterpaper,11pt]{article}
\usepackage[colorinlistoftodos]{todonotes}
\usepackage{lipsum}
\usepackage{lineno}

\usepackage[USenglish]{babel}
\usepackage[utf8]{inputenc}
\usepackage{amssymb,amsmath,amsfonts,amsthm,mathtools}
\usepackage{microtype}
\DisableLigatures{encoding = *, family = *}
\usepackage{hyphenat}
\usepackage{dsfont} 
\usepackage{multirow}
\usepackage{booktabs}
\usepackage{graphicx}
\usepackage{subcaption}
\usepackage{rotating}
\usepackage{caption}
\newcommand\exhibitnote[1]{\captionsetup{font=footnotesize,justification=raggedright,singlelinecheck=false}\caption*{#1}}
\makeatletter
    \renewcommand\@makefntext[1]{\onehalfspacing\raggedright\hskip+1em\@makefnmark#1}
\makeatother

\usepackage{pdflscape}

\usepackage[longnamesfirst, authoryear]{natbib}

\usepackage{setspace}

\usepackage{geometry}
\usepackage[tiny]{titlesec}

\newlength{\normalparindent}
\usepackage{enumitem}

\usepackage{xpatch}
\xpatchcmd{\proof}
  {\itshape}
  {\bfseries}
  {}
  {}

\usepackage{accents}

\newtheorem{theorem}{Theorem}
\newtheorem{example}{Example}
\newtheorem{assumption}{Assumption} 
\newtheorem{lemma}{Lemma} 
\newtheorem{corollary}{Corollary} 

\newcommand{\eps}{\varepsilon} 
\newcommand{\m}[1]{\mathcal{#1}}
\newcommand{\ha}[1]{\widehat{#1}}
\newcommand{\ti}[1]{\widetilde{#1}}

\newcommand{\mme}{\mathbb{E}} 
\newcommand{\mmi}{\mathbb{I}} 

\newcommand{\mmp}{\mathbb{P}} 
\newcommand{\mmr}{\mathbb{R}} 

\newif\ifid
\idtrue 

\begin{document}

\onehalfspacing

\title{\uppercase{Better Bunching, Nicer Notching}
\ifid 
\fi }

\ifid
\author{
Marinho Bertanha\thanks{Corresponding author. Department of Economics, University of Notre Dame, 3060 Jenkins Nanovic Halls, Notre Dame IN 46556. Email: mbertanha@nd.edu. Website: www.nd.edu/$\sim$mbertanh.
} 
\and
Andrew H. McCallum\thanks{Board of Governors of the Federal Reserve System. Email: andrew.h.mccallum@frb.gov.
Website: www.andrewhmccallum.com.}
\and
Nathan Seegert
\thanks{Eccles School of Business, University of Utah. Email: nathan.seegert@eccles.utah.edu.
Website: www.nathanseegert.com.} 
}
\fi


\vspace{1cm}

\date{First draft: August 5, 2017
\\
This draft: June 12, 2023}
\maketitle
\thispagestyle{empty}

\vspace{1cm}

\begin{abstract} 
This paper studies the bunching identification strategy for an elasticity parameter that summarizes agents' responses to changes in slope (kink) or intercept (notch) of a
schedule of incentives.
We show that current bunching methods may be very sensitive to implicit assumptions in the literature about unobserved individual heterogeneity.
We overcome this sensitivity concern with new non- and semi-parametric estimators. 
Our estimators allow researchers to show how bunching elasticities depend on different identifying assumptions and when elasticities are robust to them. 
We follow the literature and derive our methods in the context of the iso-elastic utility model and an income tax schedule that creates a piece-wise linear budget constraint.
We demonstrate bunching behavior provides robust estimates for self-employed and not-married taxpayers in the context of the U.S. Earned Income Tax Credit. 
In contrast, estimates for self-employed and married taxpayers depend on specific identifying assumptions, which highlight the value of our approach.
We provide the Stata package \texttt{bunching} to implement our procedures.
\end{abstract}

\vspace{1cm}


\clearpage
\onehalfspacing
\section{Introduction}
\label{sec:introduction}
\indent 

Estimating agents' responses to incentives is a central objective in economics and many other social sciences.  
Piecewise-linear schedules provide identifying variation in incentives to estimate these responses. 
A continuous distribution of agents that face a piecewise-linear schedule of incentives results in a distribution of responses with mass points located where the slope or intercept of the schedule changes. 
For example, a progressive schedule of marginal income tax rates induces a mass of heterogeneous individuals to report the same income at the level where marginal rates increase.
Many studies in economics use mass points in the response distribution to recover primitive parameters that govern agents' responses to incentives.

Pioneering work by \cite{Saez2010}, \cite{chetty2011}, and \cite{KlevenWaseem2013} develop bunching estimators to use mass points in response distributions to recover primitive parameters.
These estimators are widely applied in economics and rely on the idea that a mass point is larger, the more responsive agents are to incentives. 
The size of the mass point, however, also depends on the unobserved distribution of agents' heterogeneity. 
Current methods are only able to map the size of mass points to primitive parameters because they make specific assumptions about the unobserved distribution.

This paper places bunching estimators on a statistical foundation and makes three contributions on the identification of a primitive parameter that summarizes agents' responses to incentives.
First, we clarify how the mapping of observed variables to an elasticity parameter depends on assumptions about the unobserved distribution of heterogeneity.
The elasticity parameter captures the log percentage change of a response to a log percentage change in an incentive.
A change in the intercept of the incentive schedule admits nonparametric point identification of the elasticity but a change in slope does not. 
Second, we examine the assumptions made by current bunching methods and propose weaker assumptions for partial and point identification of the elasticity.
Third, we revisit the original empirical application of the bunching estimator, which is in the 
literature that examines the largest means-tested cash transfer program in the United States \textemdash the Earned Income Tax Credit (EITC). 
Our weaker assumptions about the unobserved distribution of heterogeneity result in meaningful changes in estimates of individual responses to taxes. 

Our first contribution is to clarify the importance of assumptions about unobserved heterogeneity for the identification of the elasticity.
Many existing estimates are based on an agent optimization problem with a piece-wise linear constraint that has one change in slope or intercept. Slope changes in the constraint are often referred to as ``kinks'' while intercept changes are often called ``notches.''\footnote{We generalize the constraint of the agent's problem to a schedule with multiple changes in intercepts and slopes because agents typically encounter a combination of both kinks and notches. The general problem and solution are in Sections \ref{sec:app:general_prob} and \ref{sec:app:general_prob_sol} of the supplement. For ease of exposition, we keep the problem with one kink or notch in the main text (Section \ref{sec:model}).} 
The literature began with an iso-elastic utility model and an income tax schedule that creates a piece-wise linear constraint. 
We demonstrate our methods in this context and note that our results extend to other contexts with piece-wise linear constraints. 

We highlight two insights about identification with kinks and notches assuming a nonparametric family of distributions for unobserved heterogeneity that have continuous probability density functions (PDFs).
First, if the constraint has at least one notch, it is possible to point identify the elasticity.
Identification comes from using the empty interval in the support of the observed distribution that is created by agents' responses to a notch.
Second, point identification is impossible if the incentive schedule only contains kinks. 
Identification is impossible because there always exists an unobserved distribution that reconciles any elasticity with the observed distribution of responses. 

Our second contribution is to propose three novel identification strategies for the elasticity if the incentive schedule has kinks but no notches. 
Each of these strategies relies on weaker assumptions than those implicit in current implementations of the bunching estimator.
Our first strategy identifies upper and lower bounds  on the elasticity \textemdash partially identifies the elasticity \textemdash by making a mild shape restriction
on the nonparametric family of heterogeneity distributions.
The other two strategies point identify the elasticity using covariates and semi-parametric restrictions on the distribution of heterogeneity.

The first strategy partially identifies the elasticity by assuming a bound on the slope magnitude of the heterogeneity PDF, that is, Lipschitz continuity.
Intuition for identification of the elasticity in this setting is as follows. 
We observe the mass of agents who bunch, which equals the area under the heterogeneity PDF inside an interval.
The length of this bunching interval depends on the unknown elasticity.
The maximum slope magnitude of the PDF implies upper and lower bounds for all possible PDF values inside the bunching interval that are consistent with the observed bunching mass.
This translates into lower and upper bounds, respectively, on the size of the bunching interval, which corresponds to lower and upper bounds on the elasticity. 
These bounds allow researchers to examine the magnitude of the impossibility result in their empirical context.
Depending on the data, it might take an unreasonably high slope magnitude on the heterogeneity PDF to produce bounds that include all possible elasticity values. 
In other settings, the difference between upper and lower bounds may be economically large even for small slope magnitudes.



The next two strategies rely on the fact that bunching can be rewritten as a censored regression model with a middle censoring point. 
We stress that while these strategies necessarily add structure to point identify the elasticity, they do not require fully parametric assumptions, such as normality, on the unconditional distribution of heterogeneity.

The second strategy identifies the elasticity by estimating a maximum likelihood mid-censored model, using data truncated to a window local to the kink. 
The likelihood function assumes that the unobserved distribution conditional on covariates is parametric, but we demonstrate that correct specification of the conditional distribution is not necessary for consistency, as long as the unconditional distribution is correctly specified.
For example, conditional normality yields a mid-censored Tobit model, which has a globally concave likelihood and is easy to implement.
Nevertheless, consistency only requires that the unobserved distribution is a semi-parametric mixture of normals, and that the unconditional distribution implied by the Tobit model matches that; conditional normality is not necessary.
Truncating the sample around the kink point improves the fit of the model and further weakens these distribution assumptions.

The third strategy restricts a quantile of the unobserved distribution, conditional on covariates, and point identification follows existing theory for censored quantile regressions \citep{powell1986,chernozhukov2002,chernozhukov2015}.


Both of the two semi-parametric methods are censored regression models that incorporate covariates.
These approaches extend bunching estimators to control for observable heterogeneity for the first time. 
Observable individual characteristics generally account for substantial variation across agents and leave less heterogeneity unobserved.   
This fact suggests that identification strategies that utilize covariates should be preferred over identifying assumptions that only restrict the shape of the unobserved distribution without covariates. 
In addition, covariates generally allow for more precise estimates.

Our third contribution is to illustrate the empirical relevance of our methods by revisiting \cite{Saez2010}'s original, influential application of bunching in the distribution of U.S. income caused by kinks in the EITC schedule. 
That approach implicitly assumes the unobserved PDF of agents that bunch is linear and uses a trapezoidal approximation to compute the bunching mass. 
This assumption fits poorly when the true density is non-linear or the interval of agents that bunch is large.
We compare elasticity estimates based on our identification assumptions with estimates based on the trapezoidal approximation using annual samples of U.S. federal tax returns from the Internal Revenue Service (IRS).
 
Our partial identification method indicates that households adjust their reported income in response to marginal tax rates by a considerable amount.
Placing a conservative limit on the slope magnitude, the lower bound for the elasticity among self-employed married individuals is 0.48 \textemdash that is, a one percent increase in the marginal tax rate results in a reduction in reported income of at least 0.48 percent. 
This estimate contrasts with the estimate of 0.77 using the trapezoidal approximation.
The difference in these estimates matters. 
For example, \cite{saez2001} shows that the optimal top marginal tax rate for an economy with an elasticity of 0.48 is 48\%, while the optimal tax rate for an elasticity of 0.77 is 37\%\textemdash a difference of 11 percentage points.


The truncated Tobit model with covariates fits well the observed distribution of income making our semi-parametric consistency result operative.
Elasticity estimates from this model differ substantially from estimates based on the trapezoidal approximation for some categories of U.S. taxpayers.
For example, we estimate an elasticity of 0.55 versus a trapezoidal estimate of 0.77 for self-employed and married individuals. 
This large difference highlights the sensitivity of estimates to functional form assumptions, as well as the need for methods that rely on weaker assumptions.

Our three new methods provide a suite of ways to recover elasticities from bunching behavior. 
Each method differs in the assumptions they make about the unobserved distribution to achieve identification.
There is no way to determine which assumption is correct because the unobserved distribution is not fully identified.
Nevertheless, estimates that are stable across many methods indicate that different identifying assumptions do not play a major role in the construction of those estimates.   
On the contrary, estimates that are sensitive to different assumptions are dependent on the validity of those assumptions. 
Therefore, we recommend that researchers examine the sensitivity of elasticity estimates across all available methods as a matter of routine. 

Bunching estimators are widely applied in settings including fuel economy regulations \citep{Sallee2012}, electricity demand \citep{ito2014}, real estate taxes \citep{kopczuk2015}, labor regulations \citep{Garicano2016, goff2022}, prescription drug insurance \citep{EinavFinkelsteinSchrimpf2017}, marathon finishing times \citep{Allen2017}, attribute-based regulations \citep{Ito2018}, education \citep{Dee2019, caetano2020should}, minimum wage \citep{Jales2018,Cengiz2019}, charitable giving \citep{hungerman2021}, and air-pollution data manipulation \citep{ghanem2019}, among others. 
\cite{Kleven2016} and \cite{bertanha2023} provide a recent review of the many applications and branches of the bunching literature; \cite{JalesYu2017} relates bunching to regression discontinuity design (RDD).\footnote{
Variation in the size of the mass point across groups of individuals has also been used as a first stage in a two stage approach to control for endogeneity \citep{chetty2013,Caetano2015,Grossman2019}
An additional complication in many applications arises when the bunching mass is spread over a range instead of being a mass point. 
\cite{blomquist2019} provide a discussion about the potential sources for this complication and \cite{cattaneo2018} propose a filtering method to resolve it.}


In the context of kinks, \cite{blomquist2017a} were the first to prove the impossibility of point identification and the possibility of partial identification in the iso-elastic quasi-linear utility model \textemdash and an earlier paper provides intuition for the impossibility result \citep{blomquist2015}. 
We derive partial identification bounds by assuming the PDF has a bounded slope, whereas \cite{blomquist2017a} assume the PDF of heterogeneity is monotone. 
We developed our partial identification result independently of theirs.
Our partial identification approach has three valuable features that make it novel: closed-form solutions, observed bunching always implies a positive elasticity, and nesting of the original bunching estimator. 
\cite{blomquist2018a} explain that a notch can identify the elasticity and a formal proof of identification appears contemporaneously in an earlier version of our paper, \cite{BMS2018}.
To the best of our knowledge, ours is the first paper to demonstrate point identification using censored regression models, covariates, and semi-parametric assumptions on the distribution of heterogeneity.
More generally, the theory demonstrating that a kink fails to point identify the elasticity relates to the literature on impossible inference reviewed by \cite{bertanha2019}.

The paper proceeds with an utility maximization model subject to a piecewise-linear budget constraint in Section \ref{sec:model}.
Section \ref{sec:identification} investigates the identification of the elasticity in the case of kinks and notches. 
We propose the three identification strategies for the elasticity in Section \ref{sec:solutions} and illustrate these methods empirically in the context of the EITC in Section \ref{sec:application}.
Section \ref{sec:conclusion} concludes. 
Appendix \ref{app} contains all proofs, and supplemental Appendix \ref{app:supp}
collects auxiliary results and examples.
Finally, we developed the Stata command \texttt{bunching} that implements our procedures. 
The Stata package is presented by \cite{bertanha2022stata} and available for download from the website of the authors or the Statistical Software Components (SSC) online repository.\footnote{
    Type \texttt{ssc install bunching} in Stata to install the package.    
    } 

\section{Utility Maximization Subject to Piecewise-Linear Constraints}
\label{sec:model}

\indent 

Firms' and individuals' optimization problems often face piecewise-linear constraints.  
The nature of constraints is dictated by differential tax rates, insurance reimbursement rates, or contract bonuses.
A budget set is fully characterized by a sequence of intercepts and slopes that change at known points.
A change in the intercept is referred to as a notch, and a change in the slope is referred to as a kink. 

\subsection{Model Setup}
\indent

We start with the labor supply characterization employed by the vast majority of the literature, 
which follows the seminal work of \cite{Saez2010} and \cite{KlevenWaseem2013}.
Agents maximize an iso-elastic quasi-linear utility function and choose consumption and labor subject to a piecewise-linear budget set. 
For ease of exposition, we focus on budget sets with one kink or one notch in the main text 
and generalize to multiple kinks and notches in the supplemental appendix, Sections \ref{sec:app:general_prob} and \ref{sec:app:general_prob_sol}

Consider a population of agents  that are heterogeneous with respect to a scalar variable $N^{\ast}$,
referred to as ability.
Ability is distributed according to a continuous probability density function (PDF) $f_{N^{\ast}}$, with support $(0,\infty)$, and a cumulative distribution function (CDF) $F_{N^{\ast}}$. 
Agents know their $N^{\ast}$,  but the econometrician does not observe the distribution of $N^{\ast}$. 

Agents maximize utility by jointly choosing a composite consumption good $C$ and labor supply $L$. Utility is increasing in $C$ and decreasing in $L$. These variables are constrained by a budget set, where the agent may consume all of its labor income net of taxes plus an exogenous endowment $I_0$. For simplicity, we assume the price of labor and consumption are equal to one, such that taxable labor income $Y$ is  equal to $L$.

In the budget constraint with a kink, the tax rate increases from $t_0$ to $t_1$ as income increases above the kink value $K$.
The budget constraint has a notch when the agent is charged a lump-sum tax of $\Delta>0$ as income crosses $K$. 
Agent type $N^\ast$ maximizes utility $U(C,Y;N^\ast)$ as follows,
\begin{eqnarray}
\max_{C,Y} & &  C - \frac{N^*}{1 + 1/\varepsilon} \left(\frac{Y}{N^*} \right)^{1 + \frac{1}{\varepsilon}}   
\label{eq:util_saez}
\\
 s.t. & &
\nonumber
\\
 &  & C=  \mmi\{ Y \leq  K \}[ I_0 + (1 - t_0) Y] 
 + \mmi\{ Y > K \} \left[ I_1 + \left( 1- t_{1} \right) (Y - K) \right],
\label{eq:bf}
\end{eqnarray}
where
$\mmi\{ \cdot \}$ is the indicator function;
the budget line has intercept $I_0$ and slope $1-t_0$ if $Y\leq K$, 
but intercept $I_1 = I_0+K(1-t_0) - \Delta $ with slope $1-t_1$ if $Y > K$;
and $\eps$ is the elasticity of income $Y$ with respect to one minus the tax rate when the solution is interior.
In the case of a kink, $\Delta =0 $, 
and the budget frontier is continuous; 
otherwise, in the case of a notch, it has a jump discontinuity of size $\Delta$ at $Y=K$.
The solution is always on the budget frontier in Equation \ref{eq:bf}.

\subsection{Model Solution}
\label{sec:model:sol}

\indent 

The solution for $Y$ in Problem \ref{eq:util_saez} is well known in the literature, when $K$ is a kink \citep{Saez2010} 
and when $K$ is a notch \citep{KlevenWaseem2013}:
\begin{equation}
Y = 
\left\{
\begin{array}{ccl}
N^{\ast} (1-t_{0})^{\eps} & \text{, if} & 0 < N^{\ast} < \underline{N}
\\
K                    & \text{, if} &  \underline{N} \leq  N^{\ast} \leq \overline{N}
\\
N^{\ast} (1-t_{1})^{\eps} & \text{, if} & \overline{N}    < N^{\ast} ,
\end{array}
\right.
\label{eq:levelsol-onekink}
\end{equation}
where the expressions for the thresholds 
$\underline{N}$
and 
$\overline{N}$
are given below. 
Section \ref{sec:app:general_prob_sol} in the supplement presents the proof of this result and the more general solution in the case of multiple kinks and notches.
We provide a graphical analysis of these solutions in Section \ref{sec:app:figures} in the supplement. 

In the case of a kink, $\underline{N} = K(1-t_0)^{-\eps}$,
and $\overline{N} = K(1-t_1)^{-\eps}$.
The budget frontier is continuous, but its slope suddenly decreases at $Y=K$.
For values of $N^*$ inside the bunching interval $[\underline{N},\overline{N}]$,
the agent's indifference curve is never tangent to the budget frontier, and we have the non-interior solution $Y=K$.
For values of $N^*$ outside of the bunching interval, the indifference curve is always tangent to 
some point on the budget frontier. 

In the case of a notch, the solution is interior for $N^* < \underline{N} = K(1-t_0)^{-\eps}$,
but there are no tangent indifference curves for $N^* \in [K(1-t_0)^{-\eps}, K(1-t_1)^{-\eps}]$, just as in the case of a kink.
Although tangency occurs for $N^* > K(1-t_1)^{-\eps}$, some of the resulting utility levels are lower than the utility at the notch point.
The budget frontier with a jump-down discontinuity at $Y=K$ has an interval of income values $(K,Y^I]$ that no agent ever chooses.
The value $Y^I>K$ corresponds to the interior solution of the agent with $N^*=N^I$; that is, the smallest $N^*$ such that the agent's utility is equal to the utility of the agent choosing $Y=K$.
Thus $\overline{N} = N^I$, and the solution is at $Y=K$ for $N^* \in [ \underline{N}, \overline{N}] $.
As the ability $N^*$ increases above $N^I$, the utility  gets larger than the utility at $K$, and again there is an interior solution.  
Section \ref{sec:app:general_prob_sol} in the supplemental appendix has a formal definition of $N^I$ in Equation \ref{eq:indiff}.

To make the solution more tractable, we take the natural logarithm of all variables.
Define $y= \log(Y)$, $n^*=\log(N^*)$, $\underline{n}=\log(\underline{N})$, $\overline{n}=\log(\overline{N})$, $k=\log(K)$, $s_0 = \log(1-t_0)$, and $s_1=\log(1-t_1)$.
\begin{equation}
y = 
\left\{
\begin{array}{ccl}
n^{\ast} + {\eps} s_0 & \text{, if} & n^{\ast} < \underline{n}
\\
k                  & \text{, if} &\underline{n} \leq  n^{\ast} \leq \overline{n}
\\
n^{\ast} + {\eps} s_1  & \text{, if} & \overline{n}    < n^{\ast} .
\end{array}
\right.
\label{eq:logsol-onekink}
\end{equation}

As ability $n^*$ increases, the optimal choice of $y$ increases, except when $n^*$ falls inside the bunching interval $[\underline{n},\overline{n}]$, in which $y$ remains constant and equal to $k$.


\subsection{Bunching and the Counterfactual Distribution of Income}
\label{sec:model:counterfactual}

\indent 

The solution in the previous section expresses income as a function of the model parameters and $n^*$.
For given values of $(t_0,t_1,k,\eps)$, the continuously distributed $n^*$ maps into a mixed continuous-discrete distribution for $y$.
The model predicts bunching in the distribution of $y$ at a kink or notch point (i.e. $\mmp(y=k)>0$),
but a continuous distribution of $y$ otherwise.
The amount of bunching depends on the elasticity $\eps$ and the unobserved distribution $n^{\ast}$, 
\begin{equation}
B \equiv \mmp \left(y = k \right) 
= \mmp \left(\underline{n} \le n^{\ast}  \le \overline{n} \right)
= \int^{\overline{n}}_{\underline{n}} f_{n^{\ast}} \left( u \right) ~ du
= F_{n^{\ast}}\left(\overline{n} \right)-F_{n^{\ast}}\left(\underline{n} \right),
\label{eq:bunch_j}
\end{equation}
where the length of the interval $[\underline{n},\overline{n}]$ varies with $\eps$. 

The literature typically defines $B$ in terms of the counterfactual distribution of income in the scenario without any kinks or notches. 
Let counterfactual income be $y_0$ in such case. 
The solution to Problem \ref{eq:util_saez} is simply $y_0=n^* + \eps s_0$ for every value of $n^*$.
The variable $y_0$ has continuous PDF $f_{y_0}$ and CDF $F_{y_0}$.
The bunching mass is derived as
\begin{gather}
B=\int_{k}^{k + \Delta y } f_{y_0} \left( u \right) ~ du
=F_{y_0}\left(k + \Delta y \right) - F_{y_0} \left( k \right),
\label{eq:bunch_j_saez}
\end{gather}
where $\Delta y = \eps (s_0 - s_1)$.
Figure \ref{fig:imposs}, Panels a and b,
illustrate the distributions of $y$ and $y_0$, 
and how they relate to each other, to $B$, and to $f_{n^*}$.

\cite{Saez2010}'s insight is that the mass of agents bunching $B$ is increasing in the elasticity $\eps$ for a given distribution of  $y_0$. 
In other words, the more agents shift income to the kink-point $k$, 
the more sensitive they are to changes in tax rates.
All current bunching and notching estimators use this insight to identify the elasticity.
First, the researcher obtains an estimate of the counterfactual distribution of ${y_0}$ 
and the bunching mass $B$.
Plugging these into Equation \ref{eq:bunch_j_saez} allows us to solve for an estimate of the elasticity.

In reality, instead of $y$, researchers typically observe the distribution of $\ti{y}=y+e$, where $e$ is a random variable accounting for optimization and friction errors.
We focus on the identification problem associated with identifying the counterfactual distribution of $y_0$ in the absence of errors $e$. 
In work in progress, \cite{cattaneo2018} show how to solve this problem with a deconvolution method, which is necessary to recover the distribution in the absence of these errors.
Common strategies such as the polynomial strategy, first proposed by \cite{chetty2011}, fails to solve this issue. 
We provide a simple counterexample in the supplemental appendix Section \ref{sec:app:friction_error} where the polynomial strategy fails to recover the true distribution of $y.$
A practical solution that works under limited assumptions is provided by \cite{bertanha2022stata} and implemented in Section \ref{sec:application}.

\section{Identification}
\label{sec:identification}

\indent

This  
section investigates identification with one notch or one kink.
We show that identification is possible with one notch without any restriction
on the distribution of $n^*$.
On the other hand, identification in case of a kink is impossible, unless the researcher imposes restrictions on the distribution of $n^*$.
The general solution to Problem \ref{eq:util_saez} with multiple kinks and notches 
is found in Section \ref{sec:app:general_prob_sol} of the supplement.
That section discusses interesting insights to the identification of the elasticity that arise in that context.

\subsection{Identification from Gaps in the Distribution}

\indent

We show that identification in the case of a notch is possible using the additional information from the gap in the distribution that is not used in previous studies.
Specifically, the gap in the distribution of $Y$ is $(K, Y^I],$ where $Y^I = N^I (1-t_{1})^{\eps}$,
and  $N^I$ is defined above.
Once $Y^I$ is identified from the support of the distribution of $Y$, we numerically solve for $\eps$ that satisfies the indifference condition in Equation \ref{eq:notch-poss} below.



\begin{theorem}
\label{theo:notch-poss}
Suppose the support of $N^*$ is equal to $(0,\infty)$, that $K$ is a notch, and that the upper limit of the empty interval in the support of $Y$ to the right of $K$ is equal to $Y^I$.
Then the indifference condition that defines $Y^I$ is equivalent to
\begin{gather}
Y^I   + \eps K  \left( \frac{K }{Y^I} \right)^{\frac{1}{ \eps}}    
= \left( {1+\eps } \right) 
 \left(
 \frac{C + I_{1} + K(1 - t_{1})}
 {1 - t_{1}}
 \right),
 \label{eq:notch-poss}
\end{gather}
where $C$ is the consumption value on the budget frontier at the notch point.
Moreover, there exists an unique $\eps$ that solves Equation \ref{eq:notch-poss} as a function of $Y^I$, $K$, $C$, $I_{1}$, $t_{1}$.
Therefore the elasticity is identified.
\end{theorem}

This and all other proofs are given in Appendix \ref{app}.\footnote{
Similar arguments were given by \cite{blomquist2018a} around the same time this result appeared in an earlier version of our paper, \cite{BMS2018}.}
Theorem \ref{theo:notch-poss} and its proof consider the case of a notch from a lump-sum tax,  $ \Delta > 0 $ in Equation \ref{eq:bf}.
A minor change to that proof shows the elasticity is also nonparametrically identified in the case of a notch from a lump-sum subsidy, $ \Delta < 0 $.
Another case where the elasticity is nonparametrically identified is that of a concave kink, that is, a kink arising from a decrease in tax rates, $t_0>t_1$. 
This is formally demonstrated in Section \ref{sec:app:convex:kinks} of the appendix and is the only part of the paper where we refer to a kink caused by $t_0>t_1$. 
When the rest of the paper refers to a kink, we mean a kink generated by $t_0<t_1$. 

The identification for these three cases (notch with $\Delta >0$, notch with $\Delta<0$, and kink with $t_0>t_1$) does not depend on the bunching mass, but instead on the size of the region with missing mass (size of the gaps in the distribution). 
In all the cases we consider, the distribution of income cannot have optimization frictions, which we discuss in detail in Section \ref{sec:app:friction_error} of the supplement.



\subsection{Lack of Identification With One Kink}
\label{sec:identification:imposs}
\indent

Although bunching is increasing in the elasticity for a fixed distribution of $y_0$ or $n^*$, it is also true that, for a fixed elasticity, bunching increases as 
$f_{n^*}$ becomes more concentrated between $\underline{n}$ and $\overline{n}$. 
If all we know about $f_{n^*}$ is that it is continuous with full support and that its integral over $[\underline{n},\overline{n}]$ equals $B$, then there is no way to identify both the elasticity and $f_{n^*}$ using only Equation $\ref{eq:bunch_j}$;
equivalently, there is no way to identify both the elasticity and the distribution of $y_0$ using only Equation $\ref{eq:bunch_j_saez}$.
Intuitively, identification using only $\eqref{eq:bunch_j}$ or $\eqref{eq:bunch_j_saez}$ is impossible because each uses one equation to solve for two unknowns. 
This was first shown in Theorem 1 by \cite{blomquist2017a}, although the idea was first presented by \cite{blomquist2015}. 
We discuss the impossibility result in this section and present our novel identification strategies in the next sections.

Figure \ref{fig:imposs} provides intuition behind this impossibility result. 
It illustrates that the observable PDF $f_y$ in Figure \ref{fig:imposs_panela} is generated by applying Equation \ref{eq:logsol-onekink} to two different combinations of latent variable distributions and elasticities, $f_{n^*,\varepsilon}$ and $f_{n^*,\varepsilon'}$ in Figures \ref{fig:imposs_panelc} and \ref{fig:imposs_paneld}, respectively. 
In fact, for any value of the elasticity $\eps>0$, there exists a  continuous PDF of $n^*$ that justifies the observed distribution $f_y$ according to Equation \ref{eq:logsol-onekink}.
The model assumptions imply no restrictions on $\eps$ over $(0,\infty)$.
Theorem 1 by \cite{blomquist2017a} clarifies that current bunching methods are either implicitly restricting $\mathcal{F}_{n^*}$ or simply inconsistent for the true elasticity.
Section \ref{sec:app:restriction_fn} in the supplement details the implicit restrictions on $\mathcal{F}_{n^*}$ made by the original bunching methods, namely, the affine PDF assumption of \cite{Saez2010} and the uniform PDF assumption of \cite{chetty2011}.
A direct consequence of the impossibility result 
is that restrictions on $\mathcal{F}_{n^*}$ are untestable. 
One may argue that the affine or uniform assumption is a good approximation to any potentially non-linear density $f_{n^*}$ if the bunching interval $\left[k - \eps s_0, k - \eps s_1 \right]$ is small. 
The problem with this argument is that the size of the interval is itself a function of the elasticity. 
It is impossible to state that the interval is small and the linear approximation is a good one without a priori knowledge of the elasticity.

We conclude this section by stating a 
sufficient condition on how flexible $\mathcal{F}_{n^*}$ may be for point identification of $\eps$ to be possible.

\begin{assumption}\label{aspt:cond_inv}
Let $\mathcal{F}_{n^*}$ be a set of all possible CDFs of $n^*$ that are continuously differentiable.
For any  $F_{n^*} \in \mathcal{F}_{n^*}$, consider the possible values of $e \geq 0 $ and $G_{n^*} \in \mathcal{F}_{n^*}$ that satisfy the following system of equations:
\begin{align}
G_{n^*}( u - e s_0 )
& = 
F_{n^*}( u - \eps s_0 )
~~\text{for } ~\forall u < k,
\label{eq:bunch_saez_param1}
\\
G_{n^*} ( u -  es_1 )
& = 
F_{n^*} ( u - \eps s_1  )
~~\text{for } ~\forall u \geq k.
\label{eq:bunch_saez_param2}
\end{align}
The set of distributions $\mathcal{F}_{n^*}$ is restricted to be such that  the only values of $e \geq 0 $ and $G_{n^*} \in \mathcal{F}_{n^*}$ that satisfy 
Equations \ref{eq:bunch_saez_param1} \textendash \ref{eq:bunch_saez_param2} are $e=\eps$ and $G_{n^*}=F_{n^*}$.
\end{assumption}

Intuitively, Assumption \ref{aspt:cond_inv} restricts $\mathcal{F}_{n^*}$ in such a way that knowledge of the tails of an unknown CDF in $\mathcal{F}_{n^*}$
is enough to reconstruct that entire CDF; and no other CDF in $\mathcal{F}_{n^*}$ has the same tails.
The assumption is easily verified in parametric families, e.g., $\mathcal{F}_{n^*} = \{G_{n^*}(n;\theta) ~,~ \theta \in \Theta \}$, for some set of parameters $\Theta \subseteq \mathbb{R}^p$.
In Section \ref{sec:app:restriction_fn}, we show how to verify Assumption \ref{aspt:cond_inv} using an example of the Gaussian family.


\section{Solutions}
\label{sec:solutions}

\indent 

The rest of the paper focuses on methods that identify the elasticity in the kink case.
We present three types of identification assumptions on the distribution of ability,
from less restrictive to more restrictive.
We start with a nonparametric shape restriction that bounds the slope magnitude of $f_{n^*}$,
which leads to partial identification of $\eps$.
Next, we connect bunching to the literature on censored regressions, 
where $n^*$ is the regression error. 
It becomes natural to use covariates to explain $n^*$, and we propose two types of semi-parametric restrictions on the distribution of $n^*$ that point-identify the elasticity.
The first restricts the distribution of $n^*$, conditional on covariates;
and the second restricts a quantile of the distribution of $n^*$, conditional on covariates. 
In general, more data variation and structure are needed to provide any information about the elasticity.

\subsection{Nonparametric Bounds} 
\label{sec:solutions:bounds}
 
\indent 

Our partial identification approach relies on restricting the class  $\mathcal{F}_{n^*}$ to distributions with PDFs, $f_{n^*}$, that are Lipschitz continuous with constant $M \in (0,\infty)$.
In other words, the slope magnitude of any $f_{n^*}$ in this class  is bounded by $M$.
The following theorem gives the partially identified set for $\eps$ as a function of identified quantities and the maximum slope magnitude $M$.

\begin{theorem}\label{theo_partial}
Assume $\mathcal{F}_{n^*}$ contains all distributions with PDF $f_{n^*}$ that are Lipschitz continuous with constant $M \in (0, \infty)$.
Then the elasticity $\eps \in \Upsilon$, where
\begin{gather*}
\Upsilon = 
\left\{
\begin{array}{ll}
\emptyset & \text{, if   } 
B < \frac{  \left|f_{y}(k^+) - f_{y}(k^-) \right| ~ \left[ f_{y}(k^+) + f_{y}(k^-) \right]}{2M} 
\\
\left[\underline{\eps} , \overline{\eps}\right]  & \text{, if   }
\frac{  \left|f_{y}(k^+) - f_{y}(k^-) \right| ~ \left[ f_{y}(k^+) + f_{y}(k^-) \right]}{2M}
\leq
B 
<
\frac{  f_{y}(k^+)^2 + f_{y}(k^-)^2 }{2M}
\\
\left[\underline{\eps} , \infty \right)  & \text{, if   }
\frac{  f_{y}(k^+)^2 + f_{y}(k^-)^2 }{2M} 
\leq
B 
\end{array}
\right.,
\end{gather*}
where $\emptyset$ is the empty set, and
\begin{gather*}
\underline{\eps} = 
\frac{2 \left[f_{y}(k^+)^2/2 + f_{y}(k^-)^2/2 + M ~ B \right]^{1/2} 
- \left( f_{y}(k^+) + f_{y}(k^-) \right) }
{M(s_0 - s_1)}
\\
\overline{\eps}
=\frac{-2 \left[f_{y}(k^+)^2/2 + f_{y}(k^-)^2/2 - M ~ B \right]^{1/2} 
+ \left( f_{y}(k^+) + f_{y}(k^-) \right) }
{M(s_0 - s_1)}.
\end{gather*}
\end{theorem}

Figures \ref{fig:imposs_panelc} and \ref{fig:imposs_paneld} provide the intuition behind the derivation of the bounds in $\Upsilon$.
For a fixed value of $\eps$, the length of the interval $[\underline{n},\overline{n}]$
is fixed. 
If the magnitude of the derivative of $f_{n^*}$ is bounded by $M$, we obtain maximum and minimum areas under $f_{n^*}$ over 
$[\underline{n},\overline{n}]$.
We repeat this exercise for every value of $\eps$ to get a range of possible areas associated with each $\eps$.
Given the probability of bunching $B$ is the area under the true $f_{n^*}$ over 
$[\underline{n},\overline{n}]$,
the partially identified set has all values of $\eps$ whose range of possible areas contains $B$.
The partially identified set is empty if $M$ is not big enough to allow for the existence of a 
continuous function $f_{n^*}$ which connects 
$f_{y}(k^-) = f_{n^*}(k-\eps s_0)$ to $f_{y}(k^+) = f_{n^*}(k-\eps s_1)$.
The partially identified set is unbounded if $M$ is large enough to allow $f_{n^*}$
to be zero inside the interval $[\underline{n},\overline{n}]$.

The uniform approximation made by one of the original estimators (Example \ref{example2} in Section \ref{sec:app:restriction_fn} of the supplement) says that $f_{n^*}$ has zero slope inside the bunching interval, that is, $M=0$.
The trapezoidal approximation (Example \ref{example1} in that same section) implicitly chooses $M=m_0$ such that
$m_0$ is the smallest value of $M$ for which we have bounds that are well defined.
Formally, $m_0$ solves $B= \left|f_{y}(k^+) - f_{y}(k^-) \right| ~ \left[ f_{y}(k^+) + f_{y}(k^-) \right]/ 2m_0$, which makes $\underline{\eps} = \overline{\eps}$ and point-identifies $\eps$.
Thus the exercise of computing bounds necessarily involves assumptions weaker than the uniform and trapezoidal approximations.

Smoothness assumptions such as slope restrictions are now common in the partial identification literature.
For example, \cite{kim2018} study partial identification of average treatment effects under smoothness conditions on the treatment response function; 
\cite{rambachan2023} derive bounds on treatment effects in difference-in-difference designs under smoothness conditions on the class of deviations of the parallel trend assumption.  
A common issue 
in this literature 
is that the researcher must choose the smoothness assumption.
In the case of Theorem \ref{theo_partial}, the researcher must specify the value of $M$.
The impossibility of identifying the elasticity without structure on $\m{F}_{n^*}$ (Section \ref{sec:identification:imposs})
implies that it is impossible to identify, and thus estimate, the value of $M$.

We recommend researchers to conduct a sensitivity analysis by plotting the bounds in Theorem \ref{theo_partial} as a function of  $M$, for a range of values of $M$ that is considered reasonable given the empirical context.
From above, we know that $M=m_0$ yields point identification. 
A  useful reference for $M$ comes from the maximum slope magnitude of the continuous part of $f_y$, say $m_1$.
The PDF $f_y$ is identified and is the shifted PDF of $n^*$.
Thus, the maximum slope of $f_{n^*}$ outside of the bunching interval is identified and equal to $m_1$.
If we assume that the slope of $f_{n^*}$ inside the bunching interval is never bigger than outside, then $M=m_1$.
Thus, a rule of thumb for the range of values of $M$ is to start at $m_0$ and go up to at least $m_1$.
One may also estimate the worst case PDFs of Figures \ref{fig:imposs_panelc} and \ref{fig:imposs_paneld} to evaluate the visual effect of $M$ on the shape of the latent distribution.
Sensitivity analysis of this kind are not new in the partial identification literature. 
The idea is to report what can be learned under a sequence of progressively weaker assumptions.
For examples, we refer the reader to 
\cite{kim2018} 
and
\cite{rambachan2023}.

Theorem \ref{theo_partial} is important to quantify the magnitude of the impossibility problem of identification using kinks. 
If the bounds plotted for a range of $M$ values admit elasticities that are too different in economic terms, then the identifying assumptions play a critical role in determining the elasticity.
We give full details and implement this sensitivity analysis in the empirical section using our \texttt{bunching} Stata package 
(Section \ref{sec:application}).\footnote{It is important to clarify that the problem of choosing $M$ is different than the typical problem of choosing a tuning parameter, e.g., a bandwidth or polynomial order in nonparametric estimation. 
The value of $M$ represents a choice of functional form assumption, while in nonparametric estimation, you typically 
choose the tuning parameter to achieve desirable properties of the estimator for a given functional form assumption.}

We developed Theorem \ref{theo_partial} independently of \cite{blomquist2017a}, who were the first to present a partial identification result for $\eps$.
While we assume the PDF has bounded slope, \cite{blomquist2017a} partially identify the elasticity by assuming the PDF of heterogeneity is monotone. 
Our approach has three valuable properties that make it novel.
The first is that the bounds of our partially-identified set have closed-form solutions.
Second, an observed mass point implies a positive elasticity even for large values of the slope $M$, which is in line with the theoretical prediction that agents respond to a change in incentives. 
Third, it nests and is easily comparable to the original bunching estimator based on the trapezoidal approximation.

We end this subsection with the case of a budget set with several kinks $k_j$,
$j=1,\ldots, J$, but no notches.
One may ask whether the existence of several kinks helps identify the elasticity. 
As noted above, the bunching intervals  do not  overlap across kinks, that is, 
$\overline{N}_{j} = K_{j}(1-t_{j-1})^{-\eps} < K_{j}(1-t_{j})^{-\eps}= \underline{N}_{j}$.
Multiple kinks do not necessarily point-identify $\eps$, because the distribution of ${n^*}$ may be very different across different bunching intervals.

Multiple kinks do help with the identification of $\eps$, as long as the researcher restricts the slope of $f_{n^*}$ and believes the model in Equation \ref{eq:util_saez} applies to all individuals. 
This arises from the fact that every individual is assumed to have the same elasticity parameter $\eps$, and that the bounds of Theorem \ref{theo_partial} vary in length as $B_j$, $f_y(k_j^\pm)$, $s_j$ vary across cutoffs $j=1,\ldots, J$. 
The partially identified set is narrowed down by the intersection of
bounds specific to each one of the multiple kinks.
\begin{corollary}
\label{cor_partial}
Assume the conditions of Theorem \ref{theo_partial} for each kink 
$k_j$, $j=1, \ldots, J$.
Then the elasticity $\eps \in \bigcap_{j=1}^{J} \Upsilon_j$,
where $\Upsilon_j$ is the partially identified set of Theorem \ref{theo_partial} applied to kink $k_j$.
\end{corollary}

\subsection{Semi-parametric Identification with Covariates}
\label{sec:solutions:cov}

\indent 

Identification with kinks is impossible when the distribution of ability $n^*$ belongs to the nonparametric class of all continuous distributions.
Parametric functional form assumptions identify the elasticity, but identification relies on fitting such functional form to non-bunching individuals and extrapolating the functional form to bunching individuals.

This section considers alternative identification assumptions that rely on the existence of additional covariates in the dataset. There is strong empirical evidence suggesting that ability is well explained by individual characteristics, such as
age, demographics,  filing status, etc.
For example, the ability distribution of young workers may have a very different mean and variance, compared to that of older workers. 
Extrapolations based on covariates that predict $n^*$ are much more reasonable than extrapolations solely based on the shape of the PDF of $n^*$. 
The key assumption is that covariates that help explain the distribution of $n^*$ for non-bunching individuals also help explain the distribution of $n^*$ for bunching individuals.

We start by connecting bunching to censored regression models. 
This allows us to relate to the vast econometrics literature in this area.
Consider again the data generating process given by Equation \ref{eq:logsol-onekink}.
The model for $y$ is a mid-censored model,
where the error term is $n^*$, the intercept to the left of the kink is $\eps s_0$, the intercept to the right of the kink is $\eps s_1$, and the censoring point is $k$.
The main difference between \eqref{eq:logsol-onekink} and a typical censored regression model is that the latter has the censoring point at either the minimum or maximum of the distribution of $y$ (see Equation \ref{eq_tobit_min_max} in the next subsection).
Identification, estimation, and inference in these models have been widely studied in econometrics since \cite{Tobin1958}.

There are many advantages of framing the  estimation of $\eps$ as estimation of a censored model.
Surveys of censoring models and their applications are provided by \cite{Maddala1983}, \cite{Amemiya1984}, \cite{Dhrymes1986}, \cite{Long1997}, \cite{Demaris2005}, and \cite{Greene2005}. There are straightforward extensions that account for optimizing frictions. Moreover,  censored models are easily estimated with a number of different techniques that are available in many computer packages.  Most importantly, it becomes extremely practical to add covariates as explanatory factors for the distribution of $n^*$.

Assume the researcher has access to a  vector of covariates $X \in \mathbb{R}^{1 \times (d+1)}$, where $X$ contains one intercept variable and $d$ slope variables, $\mme[X'X]$ has full rank, but otherwise the distribution of $X$ is unrestricted.
We build on censoring models with covariates to identify the elasticity by imposing two types of semi-parametric assumptions on the distribution of $n^*$.

The first type of assumption states that the distribution of $n^*$ is a certain mixture of normal distributions averaged over the distribution of covariates.
This assumption does not imply conditional normality of $n^*$ given $X$ but it is implied by conditional normality of $n^*$.
Although the Tobit model assumes  normality of the unobserved distribution conditional on covariates, we demonstrate that the Tobit estimator remains consistent under our semi-parametric class of normal mixtures \textemdash as long as the unconditional distribution for $y$ implied by the Tobit model matches the true distribution of $y$.
In addition, the researcher may estimate a truncated Tobit model on data in a small neighborhood of the kink point, which requires even weaker distribution assumptions for consistency.
Our main motivation to study the robustness of Tobit to lack of normality comes from practical reasons:
Tobit is extremely popular and easy to implement due to its likelihood function being globally concave and software being ubiquitous.

The second type of assumption imposes a parametric functional form on a quantile of the conditional distribution of $n^*$ given $X$.
Sufficient variation in covariates yields point-identification of the elasticity,
which is consistently estimated by mid-censored quantile regressions.

\subsubsection{Tobit Regression}
\label{sec:solutions:cov:tobit}

\indent

The first type of assumption is formally stated in Assumption \ref{aspt:tobit} below.
In the meantime, we construct the Tobit estimator by simply assuming that
$F_{n^*|X}(n,x) = \Phi\left(\frac{n-x\beta}{\sigma} \right)$
for 
$\beta \in \mmr^{(d+1) \times 1 }$ and   $\sigma >0$,
where $F_{n^*|X}$ denotes the true CDF of $n^*$ conditional on $X$,
and $\Phi(\cdot)$ is the CDF of a standard normal distribution. 
Again, $X$ contains one intercept variable and $d$ slope variables, and $\mme[X'X]$ has full rank.
Our goal is to first relate bunching to Tobit, which is the most popular censoring model.
Conditional normality is assumed for ease of exposition in defining the mid-censored Tobit estimator;
it will be relaxed in Assumption \ref{aspt:tobit} below.


Define the error term $U = n^\ast - X \beta$, the latent variables
$y_{0}^* = {\eps} s_0 + X \beta + U$
and
$y_{1}^* = {\eps} s_1 + X \beta + U$,
where $y_{1}^* < y_{0}^*$, since $\eps>0$ and $s_0>s_1$.
Using Equation \ref{eq:logsol-onekink}, we see that $y$ follows a mid-censored Tobit model: 
\begin{equation}
y = 
\left\{
\begin{array}{ccl}
 y_{0}^*   & \text{, if} & y_{1}^* < y_{0}^* < k
\\
k                  & \text{, if} & y_{1}^* \leq  k \leq y_{0}^* 
\\
y_{1}^*   & \text{, if} & k < y_{1}^* < y_{0}^*
\end{array}
\right\}
= \min\{ y_{0}^*   ; ~ \max \{k  ;~  y_{1}^* \}  \}. 
\label{eq_tobit_min_max}
\end{equation}

This is different from the classic Tobit model, where the censoring point is either at the minimum or at the maximum of the distribution of $y$. 
A possible estimation strategy is to adapt the two-step Heckit estimator to our setting \citep{Heckman1976,Heckman1979}. 
In the first step, we estimate a binary outcome for bunching and not bunching individuals including covariates.  
In the second step, we regress income of not bunching individuals on covariates and the equivalents of the inverse Mills ratio.
It is useful to relate the mid-censored Tobit model to two classic Tobit models (left\textendash  \: and right\textendash censored). 
To see that,
construct the variables $y_{0}=\min\{y,k \}$ and $y_{1}=\max\{k , y \}$.
It turns out that $y_{0}$ follows a right-censored Tobit 
with intercept $\eps s_0 + \beta_0$,
slope coefficients $\beta_1, \ldots, \beta_d$, where $\beta = (\beta_0, \beta_1, \ldots, \beta_d)$.
Similarly, 
$y_{1}$ follows a left-censored Tobit with intercept $\eps s_1 + \beta_0$,
and slope coefficients $\beta_1, \ldots, \beta_d$.
Thus, the elasticity is consistently estimated by the difference of both intercepts $(\eps s_1 + \beta_0) - (\eps s_0 + \beta_0) $  divided by $(s_1-s_0)$.
Although readily implementable in most statistical packages, 
this estimation strategy does not constrain the slope coefficients and variances to be equal on both sides of the kink, which translates into loss of efficiency. 
The mid-censored Tobit likelihood naturally takes these constraints into account and provides the most efficient estimates.
It is therefore our preferred implementation. 

Let $(y_i, X_i)$, $i=1,\ldots,n$, be an iid sample of observations. 
The maximum likelihood estimator (MLE) for the true parameters $(\eps, \beta, \sigma)$
is constructed by maximizing the log-likelihood function with respect to $(e,b,s)$ given the sample data,
\begin{align}
L(y_1,\ldots,y_n | X_1, \ldots, X_n ; e, b, s )   & 
\nonumber
\\
   & \hspace{-2cm}   = \frac{1}{n} \sum_{i=1}^{n}
\mmi\left\{ y_{i}<k \right\}
\log \left[ \frac{1}{s}\phi\left( \frac{ y_{i}- e s_{0}-X_i b }{ s } \right) \right]
\nonumber
\\
     & \hspace{-1.2cm} +
\mmi\left\{ y_{i}=k \right\} 
\log \left[ 
 \Phi\left(\frac{ k- e s_1 - X_i b}{s}\right) 
 -
 \Phi\left(\frac{ k- e s_{0} - X_i b }{s} \right) 
 \right]
 \nonumber
\\
  & \hspace{-1.2cm} +
\mmi\left\{ y_{i} > k \right\}
\log\left[ 
    \frac{1}{s}\phi\left(\frac{y_{i}- e s_{1}- X_i b}{s} \right)
\right]
\nonumber
\\
 & \hspace{-2cm} \equiv \frac{1}{n} \sum_{i=1}^n \ell_i( e, b, s).
\label{eq:likelihood}
\end{align}

Regardless of what the true distribution $F_{n^*|X}$ is, 
the MLE based on \eqref{eq:likelihood} is consistent for the parameters that 
maximize the population average of the log-likelihood function, 
that is,  
$(e^*, b^*, s^*) = \arg \max_{e,b,s} \mme[\ell_i(e, b, s)]$, where the solution is unique (see, e.g., \cite{hayashi2000}, Section 8.3.).
We say the elasticity is identified by a mid-censored Tobit when $e^* = \eps$.
This occurs when $F_{n^*|X}(n,x) = \Phi\left(\frac{n-x\beta}{\sigma} \right)$, but we would like to relax the normality assumption 
and still have  $e^* = \eps$.
We pursue this exercise in the next paragraphs.
The Tobit estimator is extremely practical to implement, and we believe it is worth investigating a set of assumptions on $F_{n^*|X}$ that are weaker than normality and yet sufficient for consistency. 
We start by describing these conditions on the set of possible distributions $\m{F}_{n^*|X}$.

\begin{assumption}\label{aspt:tobit}
The set of true conditional CDFs of $n^*$ given $X$ is denoted  $\mathcal{F}_{n^*|X}$.
We assume $\mathcal{F}_{n^*|X}$ is such that:
\begin{enumerate}
    \item[(i)] for every ${F}_{n^*|X} \in \mathcal{F}_{n^*|X}$, 
    the unconditional CDF 
    satisfies 
    $F_{n^*}(n) = \mme\left[ {F}_{n^*|X}(n,X) \right] = \mme\left[ \Phi\left(\frac{n-X b}{s} \right) \right]$
    for some $b \in \mmr^{(d+1) \times 1}$ and $s>0$.
    The set of all unconditional CDFs $\mathcal{F}_{n^*} = \left\{ F_{n^*}(n) = \mme\left[ {F}_{n^*|X}(n,X) \right] 
        \text{ for } {F}_{n^*|X}
        \in 
        \mathcal{F}_{n^*|X} 
        \right\}$
        satisfies Assumption \ref{aspt:cond_inv};
        
    \item[(ii)] recall that 
    $\underline{n} = k-\eps s_0$,
    $\overline{n} = k-\eps s_1$,
    $B = \mme\left[ {F}_{n^*|X}(\overline{n}, X) - {F}_{n^*|X}(\underline{n}, X)   \right]$;
    define 
    $D=\mmi\{n^* \geq \overline{n} \}$ and
    $B_N(X,\delta,\theta,s) = \Phi\left(( \overline{n} - \delta - X \theta  )/s \right) - \Phi\left(( \underline{n} - X \theta)/ s \right)$;
    for every ${F}_{n^*|X} \in \mathcal{F}_{n^*|X}$, 
    we have 
    $\mme\left[ {F}_{n^*|X}(n,X) \right] = \mme\left[ \Phi\left(\frac{n-X \beta}{\sigma} \right) \right]$
    for $\beta$ and $\sigma$ that satisfy
        \begin{align}
        (0,\beta,\sigma) =  \arg \min_{\delta, \theta, s} &  ~ \frac{ 1-B }{2} 
        \left\{ 
        \log(s^2) + \frac{1}{s^2} \mme\left[ \left(n^* - D \delta - X \theta \right)^2 | n^* \not \in [\underline{n}, \overline{n}] \right]
        \right\}
        \notag
        \\
        & ~ -
        B ~ \mme\left[
        \log\left( B_N(X,\delta,\theta,s) \right) | n^* \in [\underline{n}, \overline{n}]
        \right].
        \label{aspt:tobit2}
        \end{align}
\end{enumerate}

\end{assumption}

Assumption \ref{aspt:tobit}(i) says that if we take any conditional distribution of $n^*$ given $X$ from the set $\m{F}_{n^*|X}$ and integrate it over $X$ we obtain a marginal distribution of $n^*$
which is a mixture of normals. 
The mixture is averaged over the distribution of $X$ and 
the more variation in covariates one has, the richer the set $\m{F}_{n^*}$  is.
Assumption \ref{aspt:tobit}(i) also assumes the set of mixtures of normals satisfies Assumption \ref{aspt:cond_inv}, so it is sufficient for point identification of the elasticity.
To see that, note that each mixture of normals in $\m{F}_{n^*}$ is characterized by $b \in \mmr^{(d+1) \times 1}$ and $s>0$, that is, $G_{n^*}^{b,s}(n) = \mme\left[ \Phi\left(\frac{n-X b}{s} \right) \right]$.
That along with a value of the elasticity $e \geq 0$ imply a distribution of $y$,
\begin{equation*}
G_y^{e,b,s}(y) = 
\left\{
\begin{array}{ccl}
G_{n^*}^{b,s}(y - e s_0) 
& \text{, if } & y < k,
\\[0.25cm]
G_{n^*}^{b,s}(y - e s_1) 
& \text{, if } & y \geq k.
\end{array}
\right.
\end{equation*}
If we search for a mixture of normals in $\m{F}_{n^*}$ and an elasticity value that matches the observed distribution of $y$, we will solve uniquely for the true elasticity and a distribution $G_{n^*} \in \mathcal{F}_{n^*}$ (Assumption \ref{aspt:cond_inv}). 
Uniqueness of $G_{n^*}$ does not necessarily mean that $G_{n^*}$ is indexed by a unique set of parameter values $(b,s)$.
The assumption still allows for cases where different values of parameters $(b,s)$ lead to the same CDF, $G_{n^*}$.
For example, let $X=[1,W]$ and $(n^*,W)$ be distributed as a standard bivariate normal with correlation $\rho$.
The CDF of $n^*$ conditional on $X$ equals $\Phi( n-\rho W )$.
The unconditional CDF of $n^*$ is equal to $\Phi(n)$, which is the same as $\mme[\Phi( n-\rho W )]$ or $\mme[\Phi( n )]$. 
Thus, $\Phi(n)$ belongs to $\mathcal{F}_{n^*}$ with $s=1$ and either $b=(0,\rho)$ or $b'=(0,0)$.
Coming back to Assumption \ref{aspt:tobit}, part (i) says that $G_y^{e,b,s}=F_y$ for some value of $(e,b,s)$ where $e$ is always equal to $\eps$.

The mid-censored Tobit estimator produces a value for the elasticity $e^*$ and picks the mixture of normals characterized by $(b^*, s^*);$ both of these imply a distribution of $y$:
$G_y^{*} = G_y^{e^*,b^*,s^*}$.
The Tobit best-fit distribution $G_y^*$ may or may not match the observed distribution of $y$, even if Assumption \ref{aspt:tobit}(i) is true;
in case it does match, the equality $F_y=G_y^*$ implies $e^*=\eps$ by virtue of Assumption \ref{aspt:tobit}(i).
In general, we may not have $F_y=G_y^*$ because the Tobit MLE does not necessarily minimize
the distance between $F_{y}$ and $G_{y}^{e,b,s}$ for a choice of parameters $(e,b,s)$.
In fact, the Tobit MLE minimizes the Kullback-Leibler divergence between two conditional distributions of $y$ given $X$ averaged over $X$, 
where the two conditional distributions are 
$F_{y|X}$ and
$G_{y|X}^{e,b,s}$ for a choice of parameters $(e,b,s)$.
These are two different minimization problems in general.
Assumption \ref{aspt:tobit}(ii) states a condition on $F_{n^*|X}$ that guarantees 
that
minimizing the average Kullback-Leibler divergence is equivalent 
to minimizing
the distance between $F_{y}$ and $G_{y}^{e,b,s}$.

To interpret Assumption \ref{aspt:tobit}(ii), 
note that the objective function of the minimization problem in \eqref{aspt:tobit2} 
is a known function of $(\delta, \theta, s)$ once we fix $F_{n^*|X} \in \m{F}_{n^*|X}$ plus knowledge of the distribution of $X$, which is identified. 
That objective function  is the  average of two terms weighted by the bunching mass $B$ and $1-B$.
The term on the LHS is the objective function of the least-squares problem conditional on $n^* \not \in [\underline{n}, \overline{n}]$, and the term on the RHS is an average of the log of the bunching mass conditional on $X$ in case the distribution of $n^*$ given $X$ were normal.
Thus, the objective function is a ``penalized'' least-squares problem. 
To gain further insight, assume $\eps \approx 0$, so that $B \approx 0$ and $\underline{n} \approx \overline{n}$.
In this case, \eqref{aspt:tobit2} is approximately 
\[
(0,\beta,\sigma) =  \arg \min_{\delta, \theta, s}    
\left\{ 
\log(s^2) + \frac{1}{s^2} \mme\left[ \left(n^* - D \delta - X \theta \right)^2  \right]
\right\}.
\]

This is equivalent to saying that the population regression of $n^*$ on $(D,X)$ must produce $(0,\beta)$ as coefficients,
and that the variance of the regression error must be $\sigma^2$.
This translates to linear restrictions on the first two moments of the joint distribution of $\left( n^*, \mmi\{n^* \geq \overline{n} \}, X \right)$, but that distribution is otherwise unrestricted. 
In particular, a Gaussian distribution
${F}_{n^*|X}(n,X)  = \Phi\left(\frac{n-X b}{s} \right)$ belongs to the set $ \m{F}_{n^*|X}$,
but not every element of $\m{F}_{n^*|X}$ is Gaussian. 
We give examples of $F_{n^*|X}$ that satisfy Assumption \ref{aspt:tobit} and are not Gaussian in the simulation experiments at the end of this section 
(Figures \ref{fig:simulation1}\textendash  \ref{fig:simulation2}) 
and in Section \ref{sec:supp:tobit:robust} of the supplemental appendix.
Lemma \ref{lemma:tobit_robust} below  shows that the mid-censored Tobit MLE identifies the elasticity under Assumption \ref{aspt:tobit}, and  the proof is in Section \ref{sec:app:tobit} of Appendix \ref{app}.

\begin{lemma}
\label{lemma:tobit_robust}
Let $F_y$ be the CDF of the observed distribution of $y$ and
$G_y^*$ be the CDF of the Tobit best-fit distribution for $y$ as defined above.
Suppose Assumption \ref{aspt:tobit}(i) holds. 
If $G_y^* = F_y$, then $e^*=\eps$.
Moreover, suppose Assumption \ref{aspt:tobit}(ii) holds. 
Then, $G_y^* = F_y$ and $e^*=\eps$.

\end{lemma}

If the Tobit best-fit distribution of $y$ matches the true distribution of $y$,
Lemma \ref{lemma:tobit_robust} guarantees that the elasticity estimated by the Tobit is consistent for the true elasticity, regardless of whether $F_{n^*|X}$ is normal.
We provide an example of this in the second simulation experiment at the end of this section (Figure \ref{fig:simulation2}), as well as in Section \ref{sec:supp:tobit:robust} of the supplemental appendix (Figure \ref{fig:simulation3}). 
Standard quasi-MLE asymptotic inference procedures apply here. 
Namely, the MLE $(\hat \eps, \hat \beta, \hat \sigma)$ obtained from \eqref{eq:likelihood} and centered at $(e^*, b^*, s^*)$ 
is asymptotically normal, with  zero mean and the usual variance-covariance matrix in the ``sandwich form.''

One of the features of  bunching estimators is the reliance on data local to the kink point.
With the mid-censored Tobit model, the researcher may also restrict the sample to observations of $y$ lying in a small neighborhood of $k$ and estimate a truncated Tobit.\footnote{
The truncated Tobit model has log-likelihood that is slightly different from \eqref{eq:likelihood}. 
Instead of the log-likelihood of $y|X$, we maximize the log-likelihood of $y|X,k-\delta<y<k+\delta$ for $\delta>0$,
which has a truncated normal distribution censored at $k$.}
The truncated Tobit is an attractive estimation strategy, because consistency of $\hat \eps$ 
has weaker requirements in terms of Assumption \ref{aspt:tobit} and the $G_y^* = F_y$ condition.
Assumption \ref{aspt:tobit} only needs to hold for the distribution of $n^*$ conditional on $n^*$ being in a small interval containing $[\underline{n}, \overline{n}]$.
Moreover, the smaller the truncation window, the easier it is to fit the unconditional distribution of $y$ with a Tobit, and the stronger is the robustness result of Lemma \ref{lemma:tobit_robust}.

As a matter of routine, we recommend researchers  estimate 
a truncated Tobit model for various window sizes around the kink point and  examine two things:
first, the plot of the estimated elasticity as a function of the size of the truncation window;
second, the plot of the best-fit Tobit distribution of $y$ compared to the histogram of $y$ for various  sizes of truncation windows.
The distribution fit tends to improve as the size of the window decreases. 
The better the fit, the more likely the conditions of Lemma \ref{lemma:tobit_robust} are met, and the closer is the elasticity to the truth.
We illustrate this exercise with simulated data below and with real data in Section \ref{sec:application}.

To end this section, we carry out two simulation experiments to illustrate the robustness property of our Tobit estimator to lack of normality.
In both experiments, we use the Skewed Generalized Error Distribution (SGED) with parameters
$\mu$, $\sigma$, $k$, and $\lambda$ (see \cite{theodossiou2000}, Section 5A).
We denote the PDF $f_{n^*}(n)$ of a SGED as $SGED(n;\mu,\sigma,k,\lambda)$. 
Parameters $\mu \in \mmr$ and $\sigma \in \mmr_+$ equal the mean and standard deviation of the distribution, respectively.
The parameter $k \in \mmr_+$ relates to kurtosis, while $\lambda \in (-1,1)$ regulates skewness.
SGED nests the Laplace distribution ($\lambda=0$, $k=1$), the normal distribution ($\lambda=0$, $k=2$),
and the uniform distribution ($\lambda=0$, $k \to \infty$) as special cases.
The data generating processes of these experiments have some empirical features that resemble those of the EITC data in Section \ref{sec:application}, e.g., location of the kink and range of the distribution. 

Experiment 1 has two goals. First, the experiment shows that a mixture of normals can be almost any distribution as long as the distribution of $X$ is rich enough. 
The unconditional distribution of $n^*$ does not need to be ``locally normal'' or ``locally symmetric'' at the kink. 
The second goal of Experiment 1 is to show that truncation without covariates may require a much smaller truncation window to fit the distribution of $y$ compared to truncation with covariates.
The key parameters of Equation \ref{eq:logsol-onekink} are $\eps=1$, $k=2.0794$, $s_0=0.2624$, and $s_1=-0.1054$.
Figure \ref{fig:simulation1:a} displays the distribution of $n^*$, which is a mixture of two SGEDs: 
$f_{n^*}(n)= (1/2)SGED(n; 1.6,0.75,4,-.5) + (1/2)SGED(n; 6,0.75,1,.5)$.
The random variable $X$ is a scalar and $\beta=1$.
We specify $F_{n^*|X}(n^*) = \Phi\left((n-X)/0.0717\right)$ and solve numerically for the distribution of $X$ that satisfies $F_{n^*}(n) = \mme\left[\Phi\left((n-X)/0.0717\right) \right]$ (Figure \ref{fig:simulation1:b}).
We generate 50,000 observations of $(y,X)$ according to this model,
and the histogram of $y$ is displayed in Figure \ref{fig:simulation1:c}.
Although $f_{n^*|X}$ is normal, it is clear from the figures that $f_{n^*}$ is very far from Gaussian, even within smaller truncation windows.

We estimate two different Tobit models using data from Experiment 1. 
The first model is correctly specified with covariate $X$.
We start with the full sample of simulated data and produce estimates for truncation windows that are symmetric around the kink point and shrink in size. 
For example, Figures \ref{fig:simulation1:c}\textendash \ref{fig:simulation1:e} show the histogram of simulated data for $y$, and the best-fit Tobit distributions for three truncation sizes, 100\%, 60\%, and 20\%.
Figure \ref{fig:simulation1:f} displays the elasticity estimate as a function of the percentage of data used in each truncated estimation. 
As expected, the elasticity estimate is stable over all truncation windows, because the model is correctly specified. 
The Tobit fits the distribution of $y$ perfectly for all truncation windows, and the estimated elasticity is approximately equal to the truth.

The second model we estimate with data from Experiment 1 omits the covariate $X$.
The Tobit model is misspecified because $f_{n^*}$ is not normal.
As expected, estimation using all of the data does not fit the distribution of $y$ (Figure \ref{fig:simulation1:g}).
As a general rule, the smaller the truncation window, the better the fit to the distribution of $y$ (Figures \ref{fig:simulation1:g}\textendash \ref{fig:simulation1:i}), although a perfect fit is not always guaranteed for reasonable sample sizes. 
It is only in the last feasible truncation window of 20\% that the fit becomes reasonable and the elasticity estimate reaches the true value (Figure \ref{fig:simulation1:j}).
This experiment shows the importance of covariates for the fit of the distribution and that truncation does not always yield the same fit that the inclusion of relevant covariates does.

We move to Experiment 2, which again has two goals. 
First, not only normality of $n^*$ is not required for consistency; conditional normality of $F_{n^*|X}$ is not required either.
Second, consistency of the Tobit elasticity and perfect fit of the distribution of $y$ do not require truncation in models without conditional normality (i.e., models with $F_{n^*|X}$ misspecified).
Figure \ref{fig:simulation2:a} plots the PDF of $n^*$, which is approximately the mixture of two SGEDs,
$(1/2)SGED(n;1,0.75,4,-.5) + (1/2)SGED(n;6,0.75,1,.5)$.
The distribution of scalar $X$ is discrete with 20 mass points (Figure \ref{fig:simulation2:b}) and is chosen such that
$F_{n^*}(n) = \mme\left[\Phi\left((n-X)/0.1919\right) \right]$
approximates the CDF of the mixture of two SGEDs.
We solved numerically for non-normal conditional distributions of $n^*$ given $X$ 
that satisfy Equation \ref{aspt:tobit2}.
Figure \ref{fig:simulation2:d} displays the true PDFs $f_{n^*|X=x}$ in black and 
the normal PDFs $g_{n^*|X=x}$ assumed by the Tobit in gray, for all values of $x$.
We clearly see that $f_{n^*|X}$ is not normal.
We then generate 50,000 observations of $(y,X)$ and fit our Tobit model with the covariate $X$ to the entire sample. 
Despite the lack of conditional normality and truncation, the Tobit model fits the distribution of $y$ (Figure \ref{fig:simulation2:c}) and estimates the elasticity at $\ha\eps = 1.0083$ (S.E. 0.0073).
Section \ref{sec:supp:tobit:robust} in the supplement repeats this experiment for the case $n^*$ has uniform distribution. 

This section gives a practical rational for using the truncated Tobit model but other more flexible model assumptions, such as index models or the semi-parametric Tobit model of \cite{chen2011}  are also possible. 
For example, and in terms of our notation, Section 7.3 of \cite{chen2011} assumes $n^* = X \beta + \sigma(X) \nu$, with $G$ being the CDF of the distribution of $\nu$ conditional on $X$.
Both $\sigma$ and $G$ are unknown but smooth functions, and $\eps$ and $\beta$ are partially identified.
The more structure one imposes on $\sigma$ and $G$, the narrower the partially identified set gets, eventually getting to point identification as in our case.
\cite{chen2011} then construct confidence regions for $\eps$ and $\beta$ that are valid regardless of point or partial identification using a sieve MLE bootstrap. 
\cite{chen2018} provide a computationally attractive alternative to the sieve MLE bootstrap that is based on Monte Carlo simulations.
Overall these methods constitute helpful approaches to performing sensitivity analysis of model assumptions and suggest a rich area for future work. 

\subsubsection{Censored Quantile Regressions}\label{sec:solutions:cov:cqreg}
\indent 

Another type of  semi-parametric assumption on the ability distribution consists of restricting  a  quantile of the distribution of $n^*$, conditional on $X$. 
Namely, for $\tau \in (0,1)$, we assume
that there exists 
$\beta(\tau) \in \mmr^{1 \times (d+1)}$
such that
\begin{equation}
Q_{\tau} \left(n^{\ast} \mid X \right) =  X \beta(\tau),
\label{aspt:clad}
\end{equation}
where $Q_{\tau}$
denotes the $\tau$-th quantile of a distribution.
A common choice in applied work is $\tau=1/2$ or the median regression.
The restriction in \eqref{aspt:clad} may be a flexible one if one includes transformations of $X$ on the right-hand side, e.g., polynomials and interaction terms.

Equation  \ref{eq_tobit_min_max} leads to
$y= \min\{ \eps s_0 + n^*   ; ~ \max \{k  ;~  \eps s_1 + n^* \}  \}$,
which is an increasing and continuous function of $n^{\ast}$. 
The quantile of an increasing and continuous function of $n^*$ is equal to that same function evaluated at the quantile of  $n^*$.
Using Equation \ref{aspt:clad}, 
\begin{equation}
Q_{\tau} \left(y \mid X \right) = \min\{\eps s_0 +  X \beta(\tau); ~ \max \{k  ;~ \eps s_1 + X \beta(\tau)\}  \}.
\label{eq:quant_y_min_max}
\end{equation}

For those observations such that $X \beta(\tau)< k - \eps s_0$ or 
$X \beta(\tau)> k - \eps s_1$, the quantile $Q_{\tau} \left(y \mid X \right)$ varies linearly with $X$; otherwise, it is constant and equal to $k$. 
Intuitively, if there is enough variation in $X$ for uncensored observations, then the slope coefficients and the intercepts are identified. This leads to identification of $\eps$. 

\begin{lemma}
\label{lemma:clad}
Define 
$\tilde{X}=
\left[X, ~
\mmi\left\{ Q_{\tau} \left(y \mid X \right) > k \right\} \right]$,
a random vector in $\mathbb{R}^{1 \times (d+2) }$.
Assume  \\
$\mme
\left[
\mmi\left\{ Q_{\tau} \left(y \mid X \right) \neq k \right\}  
\tilde{X}' \tilde{X}
\right] $ has full rank and that Equation \ref{aspt:clad} holds.
Then $\eps$ is point identified.
\end{lemma}
 
The quantile method does not contradict the impossibility of point identification discussed in Section \ref{sec:identification:imposs}, that is, we \emph{do} need  restrictions on the distribution of $n^*$ conditional on $X$ to point identify the elasticity.
These restrictions are Equation \ref{aspt:clad} and the rank condition of Lemma \ref{lemma:clad}. 
To see that, note that in the absence of covariates, the rank condition is never satisfied.
When we have covariates, we are not free to specify $Q_{\tau} \left(n^* \mid X \right)$ as flexible as desired. 
For a fixed distribution of $X$, the rank condition eventually fails as we increase the flexibility of $Q_{\tau} \left(n^*\mid X \right)$. 

We illustrate this fact with a simple example. 
Suppose the researcher has two dummy variables, $W_1$ and $W_2$, and wants to be fully flexible. 
An unrestricted $Q_{\tau} \left( n^* \mid W_1,W_2 \right)$ contains four parameters, because the conditional quantile function takes at most four different values.
That is, $Q_{\tau} \left( n^* \mid W_1,W_2 \right) = \beta_0 + \beta_1 W_1 + \beta_2 W_2 + \beta_3 W_1 W_2$.
In terms of Lemma \ref{lemma:clad}, 
$X=[1,W_1,W_2,W_1 W_2]$ is $1 \times 4$,
$\ti{X}$ is $1 \times 5$,
which implies $d=3$.
In the best case scenario for identification, four values of $Q_{\tau} \left( n^* \mid W_1,W_2 \right)$ translate to four values of $Q_{\tau} \left( y \mid W_1,W_2 \right)$ that are all different from $k$, so that $\mmi\left\{ Q_{\tau} \left(y \mid X \right) \neq k \right\}=1$.
The matrix 
$\mme
\left[
\mmi\left\{ Q_{\tau} \left(y \mid X \right) \neq k \right\}  
\tilde{X}' \tilde{X}
\right] $
$=\mme
\left[
\tilde{X}' \tilde{X}
\right] $
is $5 \times 5$
but has rank equal to $4$ at most.
Thus, $Q_{\tau} \left(n^* \mid X \right)$ must be restricted to fewer parameters for point identification to be possible.

As seen in the example with dummies, the amount of variation in covariates trades off with the degree of flexibility in specifying 
the parametric functional form for $Q_{\tau} \left( n^* \mid X \right)$.
The more variation researchers have (e.g., continuous vs. discrete $X$), the more flexible the parametric functional form may be, and the less damaging misspecification errors will be.
Unfortunately, our data in Section \ref{sec:application} only have binary covariates, which severely limits our ability to obtain sensible estimates using the quantile identification method.
An interesting extension to Lemma \ref{lemma:clad} relates to the work by \cite{hong2003} and \cite{khan2009}, which connects censored quantile regressions to moment inequalities and partial identification. 
Although increasing the flexibility in the specification of $Q_{\tau} \left( n^* \mid X \right)$ may violate the rank condition of Lemma \ref{lemma:clad}, researchers may still use those methods to partially identify the elasticity.
For inference, researchers may use the simulation-based confidence regions of \cite{chen2018}, which are valid regardless of point or partial identification.

Theoretical work  on estimation and inference of parameters in censored quantile regression (CQR) models dates back to the 1980s (\cite{powell1984, powell1986}).
Recent advances include the computationally attractive three-step estimator by \cite{chernozhukov2002}, and 
CQR with endogeneity by  \cite{hong2003}, \cite{khan2009}, and \cite{chernozhukov2015}.
In the simpler case of $Q_{\tau} \left(y \mid X  \right) = X \beta(\tau)$,
\cite{koenker1978} show that a consistent estimator for $\beta(\tau)$ is obtained by the solution to the problem
\begin{equation}
\min_{ b \in \mathbb{R}^{d+1}}
\sum_{i=1}^n \left[\rho_{\tau} \left(y_{i} - X_i b\right) \right],
\end{equation}
where $(y_i,X_i)$ $i=1,\ldots,n$ is an iid sample and $\rho_{\tau} \left(u \right)= \left(\tau - 1 \left(u \le 0 \right)\right)u$ is the so-called ``check function.''
In our case, the parametric conditional quantile function $Q_{\tau} \left(y \mid X  \right)$ is given in Equation  \ref{eq:quant_y_min_max}.
The slope and intercept coefficients are estimated by
\begin{equation}
(\hat{b}(\tau),\hat{\delta}(\tau)) 
= \arg \min_{ b \in \mmr^{d}, \delta \in \mmr}
\sum_{i=1}^n  \left[\rho_{\tau} \left(y_{i} - \min\{X_i' b; ~ \max \{k  ;~ X_i' b + \delta \}  \} \right) \right],
\label{eq:opt_quan_min_max}
\end{equation}
where $\hat{b}(\tau)$ is consistent for $\beta(\tau) + [\eps s_0, ~ 0, ~  \ldots, ~ 0]'$,
and $\hat{\delta}(\tau)$ is consistent for $\eps(s_1 - s_0)$. 
Therefore the elasticity is consistently estimated by 
$\hat\eps = \hat\delta/(s_1-s_0)$ and is asymptotically normal. 

The optimization problem in Equation  \ref{eq:opt_quan_min_max} is computationally difficult.
For the left (or right) censored case, \cite{chernozhukov2002} proposed a fast and practical estimator that consists of three steps.
Our case of middle censoring requires a straightforward modification of their method. 
We delineate practical steps to obtain $\hat\eps$ and its standard error using CQR in Section \ref{sec:supp:clad:implement} of the supplemental appendix.

\section{Application to EITC}
\label{sec:application}

\indent 

We demonstrate and compare our new methods using bunching behavior created by kinks in the earned income tax credit (EITC). 
Each method differs in the assumptions they make about the unobserved distribution to achieve identification.
There is no way to determine which assumption is correct because the unobserved distribution is not fully identified.
Nevertheless, estimates that are stable across many methods indicate that different identifying assumptions do not play a major role in the construction of those estimates.   
On the contrary, estimates that are sensitive to different assumptions are dependent on the validity of those assumptions.
\cite{PatelSeegertSmith2016} provide an empirical illustration of this sensitivity. 

First, we use our nonparametric bounds to provide initial information about how sensitive the elasticity estimate is to different shapes of the underlying ability distribution. 
When the bounds are tight, then the shape of the underlying distribution is not critical.  
But when the bounds are wide, then the shape is critical. 
In this case, reducing the range of possible elasticities requires either stronger restrictions on the shape of the ability distribution or additional data on determinants of ability.

Second, we combine observed determinants of ability with our semi-parametric approach to point identify the elasticity. 
We compare the resulting best-fit Tobit income distribution to the observed distribution for alternative samples that range from using all observations to using only data local to the kink. 
When the best-fit Tobit distribution coincides with the observed distribution, the estimated elasticity is consistent (Lemma \ref{lemma:tobit_robust}).
Furthermore, if the Tobit elasticity is within narrow nonparametric bounds, then the identifying assumptions are inconsequential; if within wide bounds, then the identifying assumptions are not contradictory and the covariates provide point identification. 
In contrast, if the Tobit elasticity is outside of the bounds, then the elasticity estimate is not robust to the two alternative identifying assumptions. 
Finally, when the best-fit Tobit distribution does not coincide with the observed distribution, the determinants of ability used for estimation are uninformative or the semi-parametric assumption is inappropriate. 

We recommend that researchers examine the sensitivity of elasticity estimates across all available methods as a matter of routine. 
We illustrate these steps in the context of the EITC in the rest of this section. 

\subsection{Data}
\label{sec:application:data}
\indent

We use data from the Individual Public Use Tax Files, constructed by the IRS. The annual cross-section for each year 1995 to 2004 includes sampling weights which allow interpretation of any estimates as being based on the population of U.S. income tax returns. This data was initially used by \cite{Saez2010} to demonstrate how to use bunching to estimate an elasticity.

The income distribution for individuals with one child demonstrates clear bunching around the \$8,580 kink (year 2008 dollars) in the EITC schedule (Figure \ref{figure:trunc_est_all}). 
Because the marginal tax rate increases from $-34$ percent to $0$ percent at \$8,580, individuals have strong incentives to report more income up to the kink point.

Observed bunching in the distribution of income suggests that people do respond to changes in tax rates. To effectively set tax rates, however, it is imperative to quantify this response precisely. 
Small variation  in elasticity estimates imply large differences in optimal tax rates \citep{saez2001}. 
For example, a true elasticity of 0.2 implies the optimal top marginal tax should be 69\%. 
If instead, the true elasticity is 0.3, the optimal top marginal tax should be 60\%.\footnote{
    This example comes from \cite{saez2001}. In particular, Equation 9 states $\bar{\tau} = (1 - g)/(1 - g + \varepsilon^{u}  +\varepsilon^{c}(a -1)),$ where $g$ is defined as the value the government has for the marginal consumption of high income earners (often set to 0), $a$ is the Pareto parameter (with baseline value of 2), and $\varepsilon^{c}$ and $\varepsilon^{u}$ are the compensated and uncompensated elasticities of taxable income. For the calculation in the text, we utilize $\varepsilon^{u}= \varepsilon^{c}$, a Pareto parameter of 2, and a $g$ value of 0.1.} 
As demonstrated above, identifying the elasticity requires information on the amount of bunching and the income distribution. The following sections show how different methods leverage different types of variation to identify the elasticity. 



The methods of this paper are designed for data without friction errors or sharp bunching.
For examples of sharp bunching data, see Figure 4 by \cite{glogowsky2018},
Figure 1 by \cite{goncalves2018}, and Figure 8 by \cite{goff2022}.
Nevertheless, the IRS data do have friction errors as the excess mass due to bunching is visibly dispersed in a small interval near the kink (for example, Figure 5 by \cite{weber2016}).
Therefore, to apply our procedures to the IRS data, we first need to filter reported income out of friction error.
A proper deconvolution theory must be developed to tackle this problem, but it is beyond the scope of this paper.%
    \footnote{There are several works in progress investigating frictions. See, for example, \cite{aronsson2018,alvero2020,cattaneo2018}, and \cite{McCallumNavarrete2022}.
    }
For now, we simply need a practical way of removing friction error before applying the different bunching estimators, so that they may be properly compared.

Following the intuition of  \cite{chetty2011},
we fit a seventh-order polynomial to the empirical CDF of reported income with friction errors $\tilde{y}$. 
As does \cite{Saez2010}, we exclude observations that lie within \$1,500 of the kink and allow an intercept change at the kink. The extrapolation of the fitted polynomial to the excluded region results in a CDF with a jump discontinuity at the kink. This is an estimate for the CDF of income without friction error, that is, $F_{y}(y)$. The size of the discontinuity equals the bunching mass.  We then rely on the fact that $y = F_{y} \left( F_{\tilde{y}}^{-1}(\tilde{y}) \right)$ and use the estimated CDFs to transform $\tilde{y}$ into $y$.

Our filtering procedure is different from the polynomial strategy discussed in Section \ref{sec:model:counterfactual} and Section \ref{sec:app:friction_error} of the supplement.
We simply aim at removing the friction error from the sample, while the the polynomial strategy 
aims to remove friction error and recover the counterfactual distribution of income, which requires much stronger restrictions according to the discussion in Section \ref{sec:identification:imposs}.
Our filtering procedure consistently estimates the CDF of $y$ under the following conditions: 
(i) frictions only affect bunching individuals additively;
(ii) friction error is independent of unobserved heterogeneity $n^*$ and has known support containing zero, i.e., $[-1500,+1500]$;
(iii) the CDF of $y$ is represented by a seventh order polynomial with an intercept change at the kink.
Section \ref{sec:supp:filter} of the supplement has the proof of this claim. 
A more general filtering method is deferred to future work.\footnote{Section \ref{sec:supp:saez_filter} of the supplement recomputes our estimates using the filtering procedure employed by \cite{Saez2010}.}

\subsection{Estimates Across Methods}
\label{sec:estimates_across_methods}
\indent 

Table \ref{tbl:tobit} reports estimates of the elasticity of taxable income using a classic bunching method, nonparametric bounds, and Tobit models with covariates. Each of these estimates relies on a different set of assumptions to identify the elasticity of taxable income, and together they provide insights into which assumptions are most defensible in the context of the EITC.  

Column 1 reports our estimates of the elasticity of taxable income using a trapezoidal approximation (Example \ref{example1}).\footnote{
    We estimate the PDF of the variables in logs rather than in levels, which simplifies the elasticity formula based on the trapezoidal approximation in Example \ref{example1}.
} 
This method assumes the unobserved PDF is linear in the bunching region, which prior literature believed to approximate non-linear distributions well. 
In practice, the appropriateness of this approximation depends on the true distribution and length of the bunching region, which are both unobserved. 
Linearity may be inappropriate if the distribution is sufficiently non-linear or the bunching region is wide.

Column 1 demonstrates substantial heterogeneity in estimates across different subsamples. In particular, the elasticity estimate is 0.32 for the all filers sample, 0.81 for self-employed individuals, 0.77 for self-employed married individuals, and 0.84 for self-employed not married individuals.

The guidelines for implementation of our nonparametric bounds in Section \ref{sec:solutions:bounds} utilizes a range of values for $M$ that includes the maximum slope magnitude of $f_y$.
We reiterate that $M$ is unidentified and that the slope of $f_y$ provides a starting point. 
The \texttt{bunching} Stata package consistently estimates the maximum slope of $f_y$ by taking the maximum slope in the histogram of $y$ across all consecutive bins.
We find that the slope is never bigger than $0.5$ across our subsamples. 
For a more conservative view, we report our nonparametric bounds using $M=0.5$ and $M=1$ in Table \ref{tbl:tobit}, Columns 2 and 3, and plot bounds for $M$ up to $2$ in Figure \ref{figure:bounds_est}.
The vertical lines in these figures designate the minimum and maximum slope, such that both the upper and lower bounds are finite numbers. 
The first line is the smallest slope that allows a continuous PDF to be consistent with both the bunching mass and observed income distribution.
At the minimum slope, both lower and upper bounds are equal to the estimate based on the trapezoidal approximation, reported in column 1.

As $M$ increases, the set of possible PDF shapes in the bunching region becomes richer.
The second line is the maximum slope before the set of possible distributions allows for a PDF that touches zero in the bunching interval.
In that case, the bunching mass remains constant for arbitrarily large $\eps$, and the upper bound is infinity (Theorem \ref{theo_partial}).

A large range between lower and upper bounds in Figure \ref{figure:bounds_est} suggests the estimates change substantially with the shape of the unobserved distribution. For example, the bounds are uninformative for the self-employed married sample, even for small values of $M$. This indicates that the data will not provide precise information on the elasticity unless the researcher imposes further functional form restrictions on the distribution of $n^*$. 
In contrast, we learn the most in the case of all filers and self-employed not married, where the bounds are narrower than in other subsamples for $M=0.5$.  
The lower bound is always defined for larger choices of $M$, which gives partial information on the elasticity without the need of being precise with the choice of $M$. 
For the exceedingly high value of $M=2$, the lower bound is about 0.25 for all filers and at least 0.4 for the other three subsamples.

Columns 4\textendash7 report our estimates of the Tobit model using the full sample and truncated samples at 75\%, 50\%, and 25\% of the data.
Figures \ref{figure:trunc_est_all}\textendash\ref{figure:trunc_est_self_notmarried} complement these estimates by graphing the actual distribution and the implied distribution from the Tobit estimates at different levels of truncation in panels a through e.  

These estimates incorporate a list of indicator variable covariates for whether a tax return filer: filed as married, used a tax preparer, claimed a real estate interest deduction, claimed mortgage interest as a business expense, received unemployment compensation, contributed to charity, received social security benefits, took the educator expenses deduction, paid into an Individual Retirement Account (IRA) for the primary filer, claimed a student loan interest deduction, claimed exemption for children living away from home, received pension income, capital gains, farm profit or loss, a positive tax credit, positive wages, owed positive income tax before tax credits but zero after, took a self-employed health insurance deduction, paid into a Keogh retirement plan, received income from a business or farm, filed a 1040 form instead of simpler form, claimed secondary taxpayer exemption, and paid into an IRA for a secondary filer. 
We  include year dummies variables with 1995 as the excluded year. 

The fit of the Tobit model generally improves as we truncate the sample closer to the kink,  which implies that the semi-parametric assumption of mixed normals is more reasonable locally than globally.
The minimum truncation necessary for a reasonable fit varies by subsample. 
For example, for self-employed not married, the fit seems reasonable using 60\% or less of the data, but for all filers, the fit only becomes reasonable at around 20\%.
It is interesting to observe that the Tobit with covariates fit the distribution better in narrower cuts of the data than for all filers.

Panel f in Figures \ref{figure:trunc_est_all}\textendash\ref{figure:trunc_est_self_notmarried} graph the elasticity estimate as a function of the percentage of data used. 
The estimates tend to plateau as the distribution fits improve.
For example, in the self-employed not married sample depicted in Figure \ref{figure:trunc_est_self_notmarried}, the estimates are all around 0.75,  using less than 60\% of the data. 
It is worth pointing out that truncated samples with less than 20\% of the data lead to numerical issues, such as perfect collinearity of covariates and lack of convergence in the likelihood maximization. 
This leads to imprecise estimates, as indicated by an upward bend in left extremity of the curves depicted in panel f of Figures \ref{figure:trunc_est_all}\textendash\ref{figure:trunc_est_self_notmarried}.\footnote{
    We note that the filtering procedure that removes optimizing errors dictates the shape of the distribution close to the kink.
    Researchers should keep in mind, as part of their robustness checks, that different filtering methods may yield different shapes of the distribution close to the kink, which may in turn lead to different elasticity estimates.
 }

\subsection{Comparisons Across Methods}

\indent

Comparisons across methods provide insights into the reasonableness of different assumptions used to estimate the elasticity. The trapezoidal approximation is always within the bounds, because its estimate is based on a linear interpolation of the PDF in the bunching region.
The slope of such line equals the minimum slope for which the bounds are defined. 
In contrast, the Tobit model using 100\% of the data is often below the lower bounds, but the Tobit distribution fails to fit the observed distribution of income globally.
Truncated Tobit estimates generally enter the bounds as the truncation window decreases and, as a result, the fit of the Tobit distribution improves.
For the all filers sample, an M larger than 0.5 is needed for the bounds to cover the Tobit estimate truncated at 25\%.  This reiterates our previous discussion that the Tobit fit for all filers is poor until we use 20\% or less of the data.

Consider self-employed married and self-employed not married filers. 
Figures \ref{figure:trunc_est_self_married} and \ref{figure:trunc_est_self_notmarried} demonstrate that the bunching mass is approximately 6 times larger for self-employed not married individuals than for self-employed married individuals. 
This difference in bunching mass might lead a researcher to conjecture that the elasticity is much larger for self-employed not married individuals.
Whether this conjecture is true depends, however, on differences in the underlying distribution of heterogeneity. 
Estimates based on the trapezoidal approximation in column 1 of Table \ref{tbl:tobit} indicate a higher elasticity for self-employed not married individuals and so do the lower bounds of our partial identification method and Tobit estimates.
They disagree in the magnitudes.
The trapezoidal estimate in column 1 says the elasticity for self-employed not married individuals is 9\% higher compared to self-employed married individuals;
conservative lower bounds in column 3 say it is 53\% higher while the truncated Tobit estimates in column 7 point to a 46\% difference.
The disagreement across methods for these subsamples indicates that assumptions on the distribution of heterogeneity are critical to obtain informative elasticity estimates.

\section{Conclusion}
\label{sec:conclusion}

\indent 

We show how to use bunching from piecewise-linear budget constraints 
to identify elasticities, under conditions weaker than those used in the literature on kinks and notches.
The key theoretical point is that bunching is determined by the elasticity parameter and the shape of an unobserved distribution.
Additional assumptions or data are needed to identify the elasticity. 

We propose a suite of estimation techniques that allows researchers to tailor their estimation to 
different assumptions 
and data variation.  
These include nonparametric bounds and semi-parametric censored models with covariates.
The nonparametric bounds are the least restrictive method and also nest estimators from the previous literature.

These techniques have wide applicability, because piecewise-linear budget constraints are common across fields, from public finance and labor, to industrial organization and accounting.
Our estimation strategies also provide a foundation for future advances in techniques that will account for different empirical hurdles. 
Of particular interest are extensions that consider optimization and friction errors, extensive margin responses, and panel data methods.

\ifid
{
\section{Acknowledgements}
\label{sec:acknow}
The views expressed in this paper are those of the authors and do not necessarily reflect the views of the Federal Reserve Board or the Federal Reserve System.
We would like to thank 
Matias Cattaneo,
Bill Evans,
Roger Gordon,
Jim Hines,
Dan Hungerman, 
Michael Jansson,
Henrik Kleven,
Brian Knight,
Erzo Luttmer,
Byron Lutz,
Dayanand Manoli,
Magne Mogstad,
Marcelo Moreira,
Whitney Newey,
Andreas Peichl,
Emmanual Saez,
Dan Silverman,
and
Joel Slemrod
for valuable comments and discussions.
The paper also benefited from feedback received from seminar participants at 
the UCSD Workshop on Bunching Estimators,
Econometric Society,
International Association for Applied Econometrics,
International Institute of Public Finance,
National Tax Association,
Dartmouth College,
Federal Reserve Board,
and 
University of Michigan.
Jessica C. Liu, 
Michael A. Navarrete, 
and Alexis M. Payne 
provided excellent research assistance. 
All remaining errors are our own.
Bertanha acknowledges financial support received while visiting the Kenneth C. Griffin Department of Economics, University of Chicago.}
\fi

\singlespacing

\bibliographystyle{chicago.bst}
\bibliography{02references.bib}


{\color{white}
\section{Figures and Tables}
}

\clearpage
\singlespacing


\begin{figure}[tbp]
\caption{Identification of the Elasticity in the Case of a Kink}
\label{fig:imposs}
\begin{subfigure}[b]{0.49\textwidth}
\caption{\centering Distribution of Observed Income}
\includegraphics[width=1\linewidth]{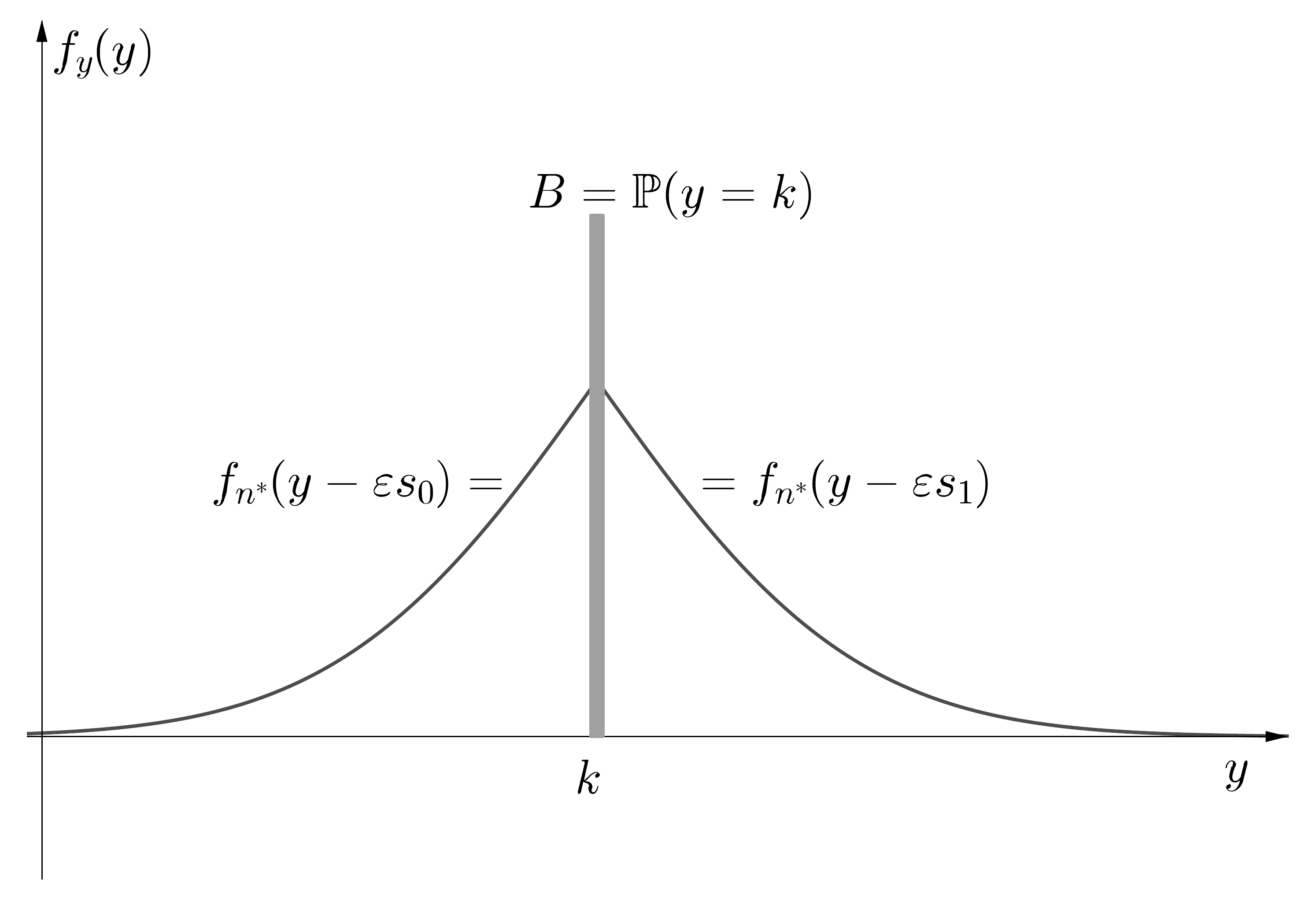}
\label{fig:imposs_panela}
\end{subfigure}
\hfill
\begin{subfigure}[b]{0.49\textwidth}
\caption{\centering Counterfactual Distribution of Income in the Absence of Kink}
\includegraphics[width=1\linewidth]{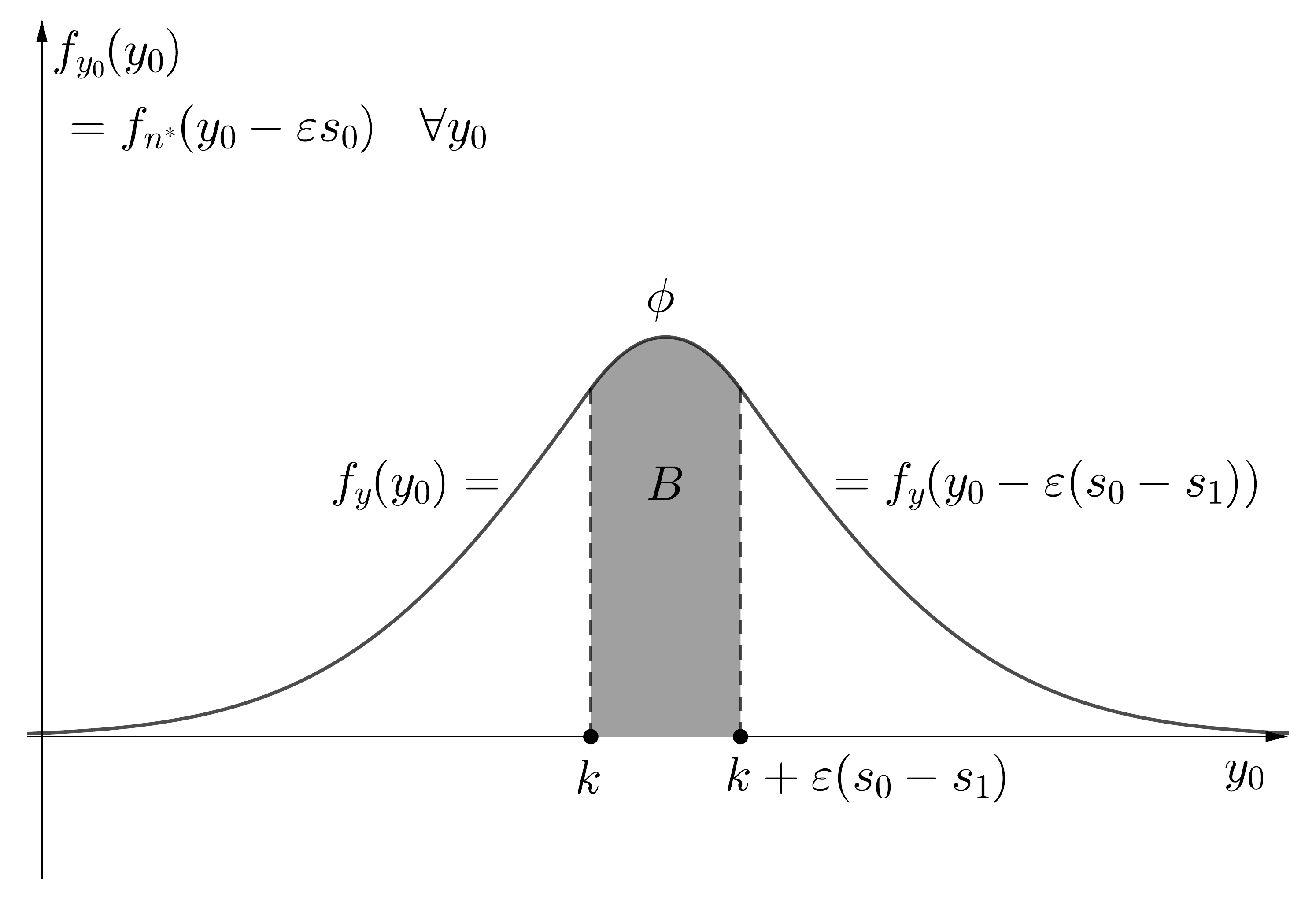}
\label{fig:imposs_panelb}
\end{subfigure}
\medskip
\begin{subfigure}[b]{0.49\textwidth}
\caption{\centering Distribution of Ability Consistent with Observed Income and Higher Elasticity}
\includegraphics[width=1\linewidth]{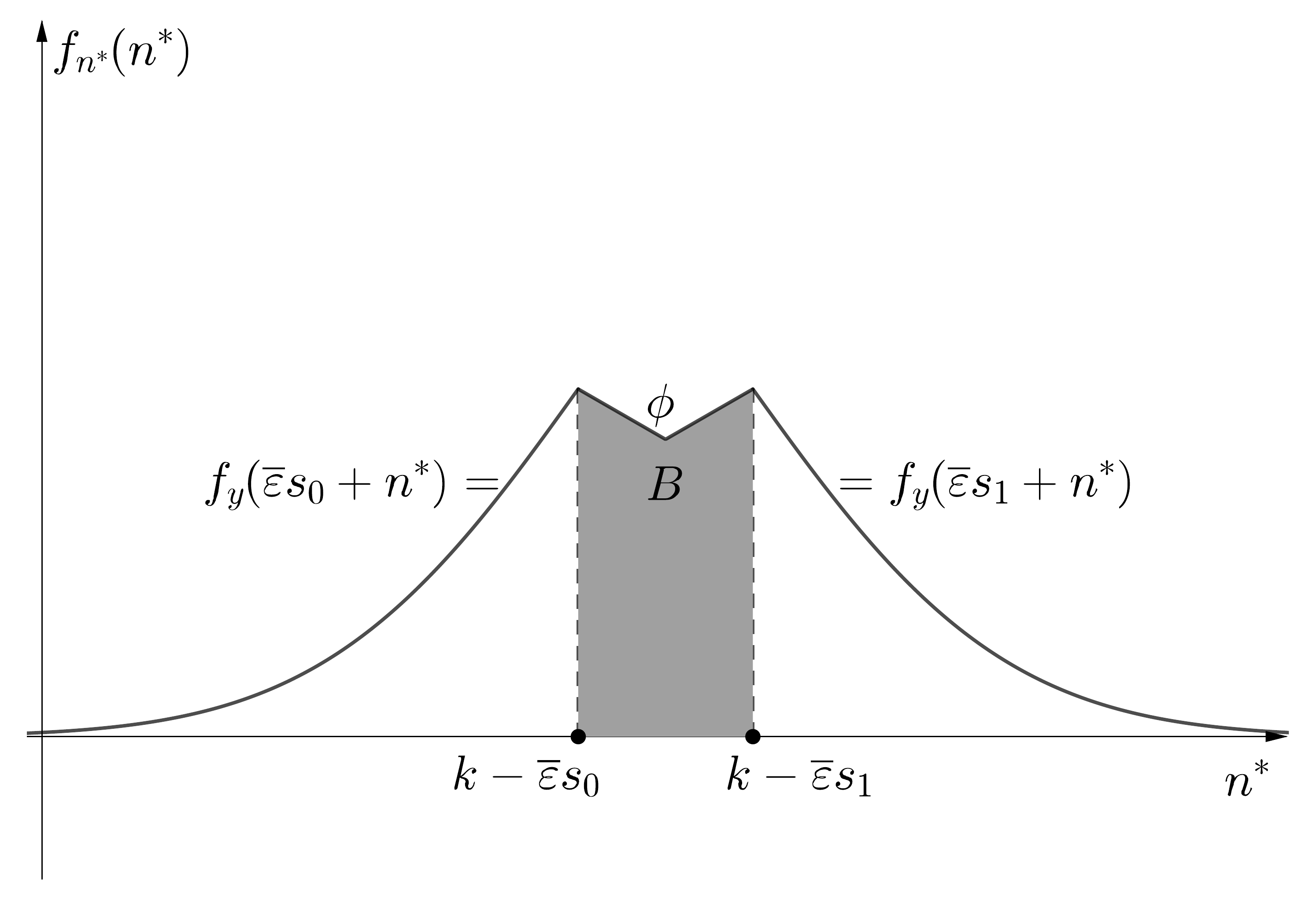}
\label{fig:imposs_panelc}
\end{subfigure}
\hfill
\begin{subfigure}[b]{0.49\textwidth}
\caption{\centering Distribution of Ability Consistent with Observed Income and Lower Elasticity}
\includegraphics[width=1\linewidth]{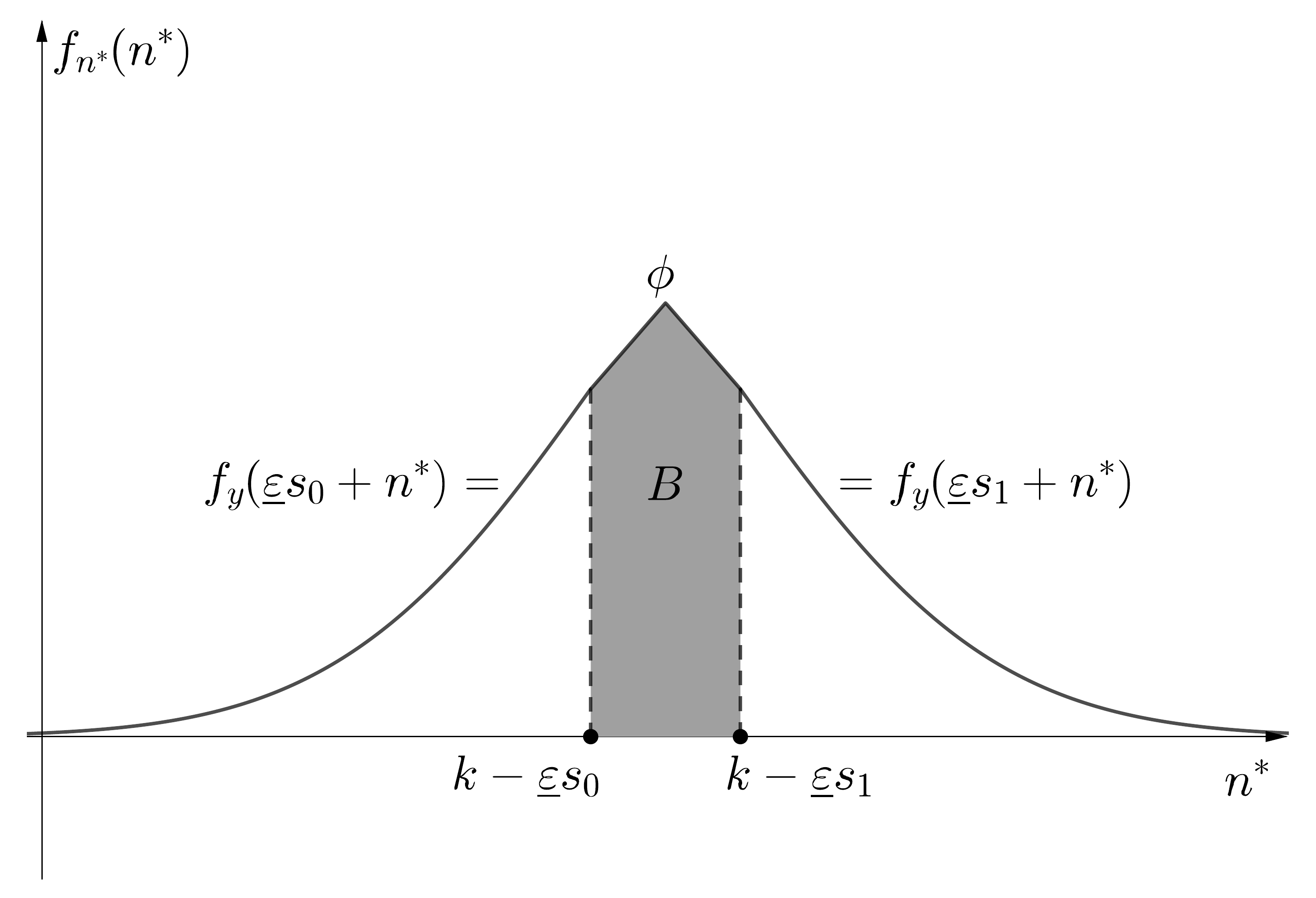}
\label{fig:imposs_paneld}
\end{subfigure}
\flushleft
{\footnotesize
\textit{Notes:}
Panel~\ref{fig:imposs_panela} plots an example of PDF of $y$. The continuous portions are equal to the PDF of ability $n^*$ shifted by $\eps s_0$ for $y < k$, and by $\eps s_1$ for $y > k$, respectively. 
The shaded area represents a discrete mass point with probability $B=\mmp\left(y=k\right)$, that is, the probability of bunching. 
Panel~\ref{fig:imposs_panelb} shows the counterfactual PDF of $y_0$, that is, the distribution of income if tax rates did not change at the kink. 
The PDF of $y_0$ is continuous, and equals the PDF of $n^*$ shifted by $\eps s_0$.
It is also equal to the PDF of $y$ before the kink, and to the shifted PDF of $y$ after the kink. However, the distribution of $y$ does not reveal the shape of the PDF of $y_0$ in the bunching region (i.e., $\phi$).
The shaded area under $\phi$ integrates to the probability of bunching $B$.
The last two panels (Panels~\ref{fig:imposs_panelc}-\ref{fig:imposs_paneld}) display two different distributions of $n^*$ that generate the same distribution of income $y$ 
(Panel~\ref{fig:imposs_panela})
with two different elasticities, $\underline{\eps}<\bar{\eps}$, according to Equation \ref{eq:logsol-onekink}.
The PDF of $n^*$ outside of the bunching region is equal to the PDF of $y$ shifted by $\eps s_0$, if $n^*<k-\eps s_0$; or shifted by $\eps s_1$, if $n^*>k-\eps s_1$.
Aside from $B$, the distribution of income does not contain any information about the shape of $\phi$ in the PDF of $n^*$.
If we assume $f_{n^*}$ is Lipschitz continuous with known constant, it is possible to derive upper and lower bounds for $\phi$, which correspond, respectively, to lower and upper bounds on the elasticity (Theorem \ref{theo_partial}).
\vspace{-29pt}} 
\end{figure}

\newgeometry{headheight=1mm,tmargin=10mm,bmargin=15mm, lmargin=10mm,rmargin=10mm}

\begin{landscape}
\begin{figure}[tbp]
\caption{Robustness of Tobit Estimates to Lack of Normality\textemdash  Experiment 1}
\label{fig:simulation1}
\centering

\begin{subfigure}[b]{0.30\linewidth}
\centering
\caption{Probability Density Function of $n^*$}
\includegraphics[width=0.7\linewidth]{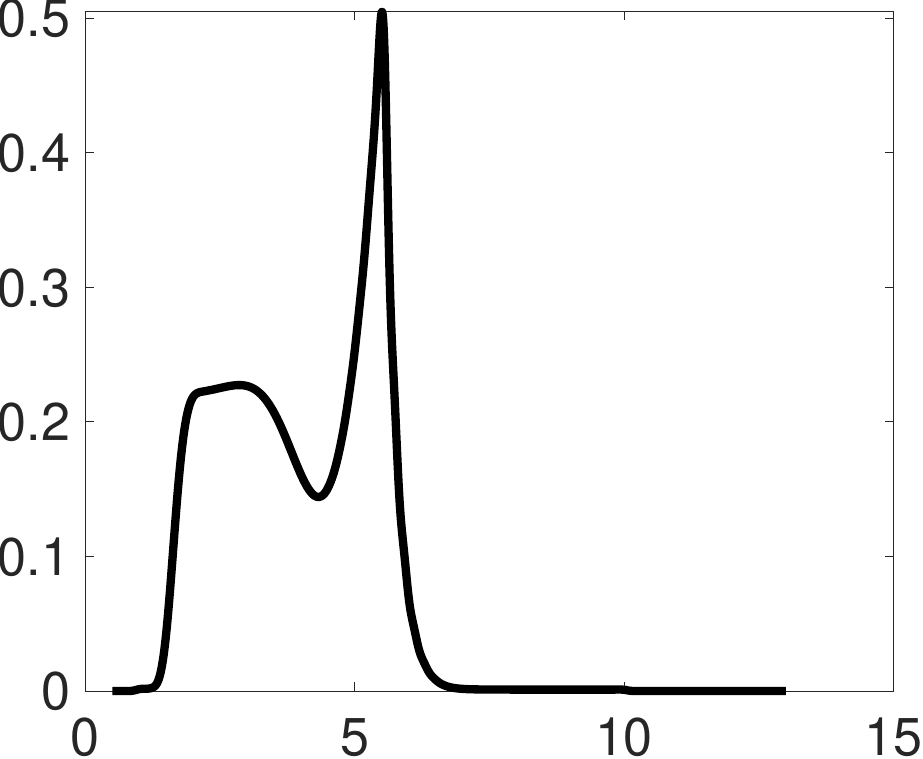}
\label{fig:simulation1:a}
\end{subfigure}
\begin{subfigure}[b]{0.30\linewidth}
\centering
\caption{Probability Mass Function of $X$}
\includegraphics[width=0.7\linewidth]{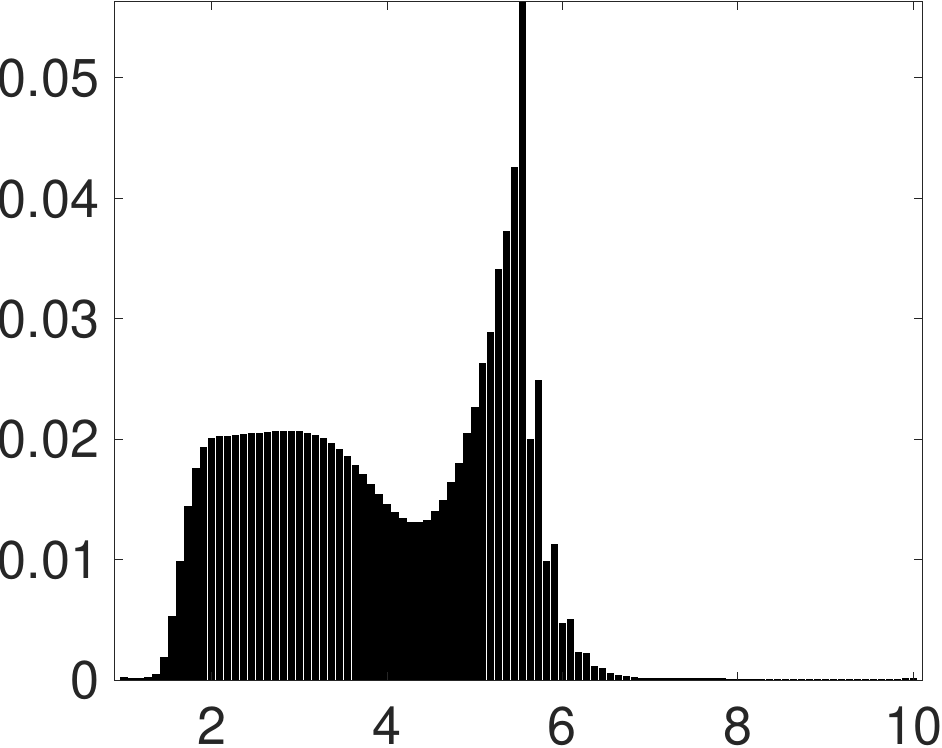}
\label{fig:simulation1:b}
\end{subfigure}

\medskip


\begin{subfigure}[b]{0.22\linewidth}
\centering
\caption{100\% of the data used}
\includegraphics[width=\linewidth]{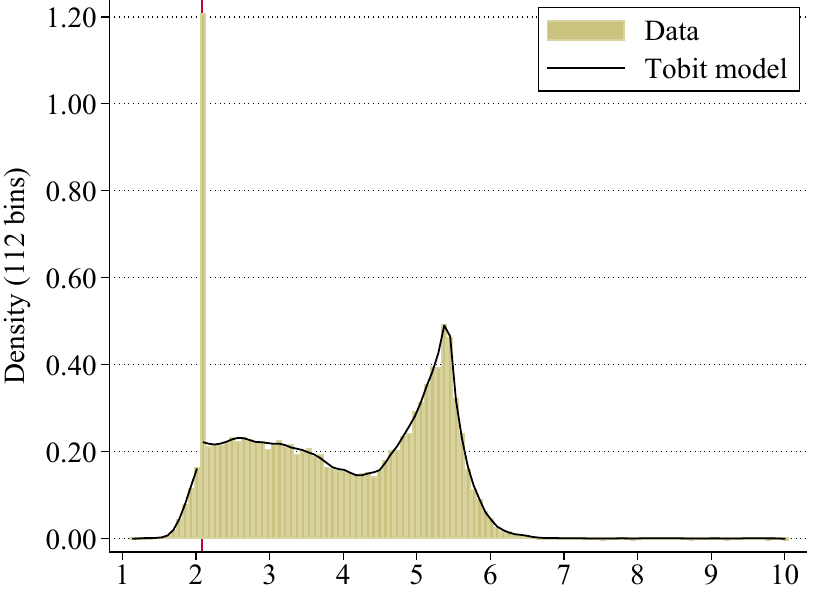}
\label{fig:simulation1:c}
\end{subfigure}
\begin{subfigure}[b]{0.22\linewidth}
\centering
\caption{60\% of the data used}
\includegraphics[width=\linewidth]{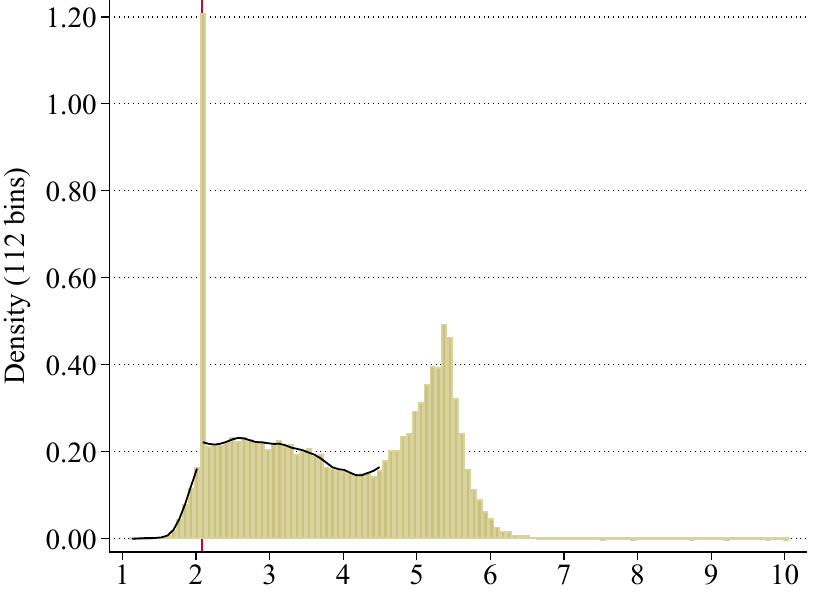}
\label{fig:simulation1:d}
\end{subfigure}
\begin{subfigure}[b]{0.22\linewidth}
\centering
\caption{20\% of the data used}
\includegraphics[width=\linewidth]{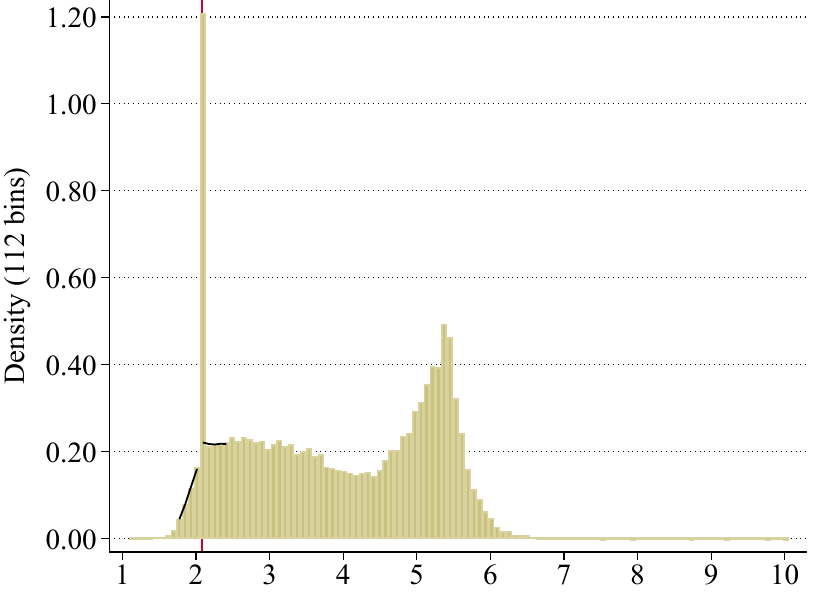}
\label{fig:simulation1:e}
\end{subfigure}
\begin{subfigure}[b]{0.22\linewidth}
\centering
\caption{Elasticity by percent used}
\includegraphics[width=\linewidth]{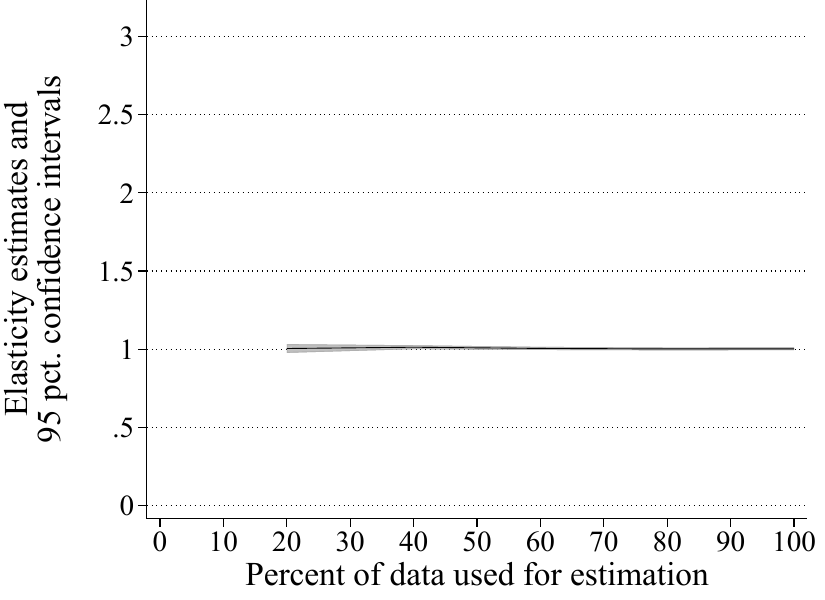}
\label{fig:simulation1:f}
\end{subfigure}

\medskip

\begin{subfigure}[b]{0.22\linewidth}
\centering
\caption{100\% of the data used}
\includegraphics[width=\linewidth]{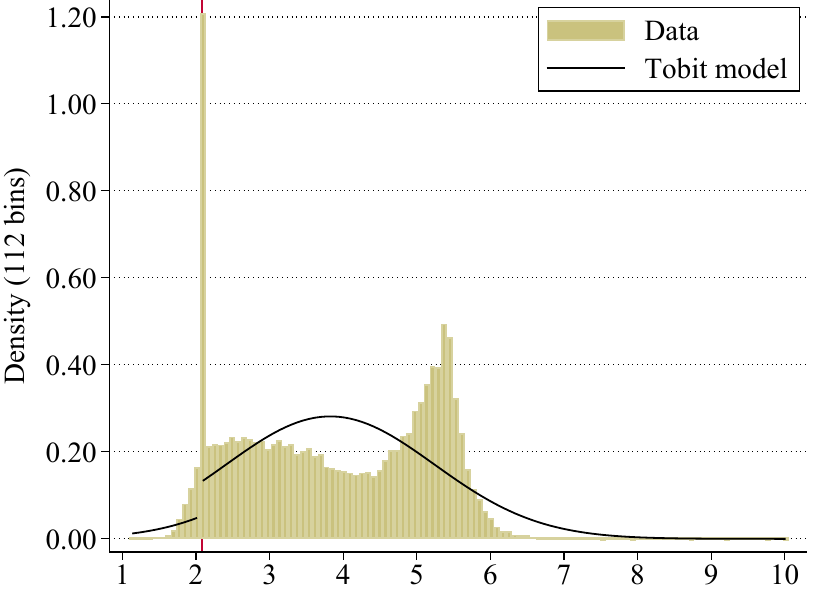}
\label{fig:simulation1:g}
\end{subfigure}
\begin{subfigure}[b]{0.22\linewidth}
\centering
\caption{60\% of the data used}
\includegraphics[width=\linewidth]{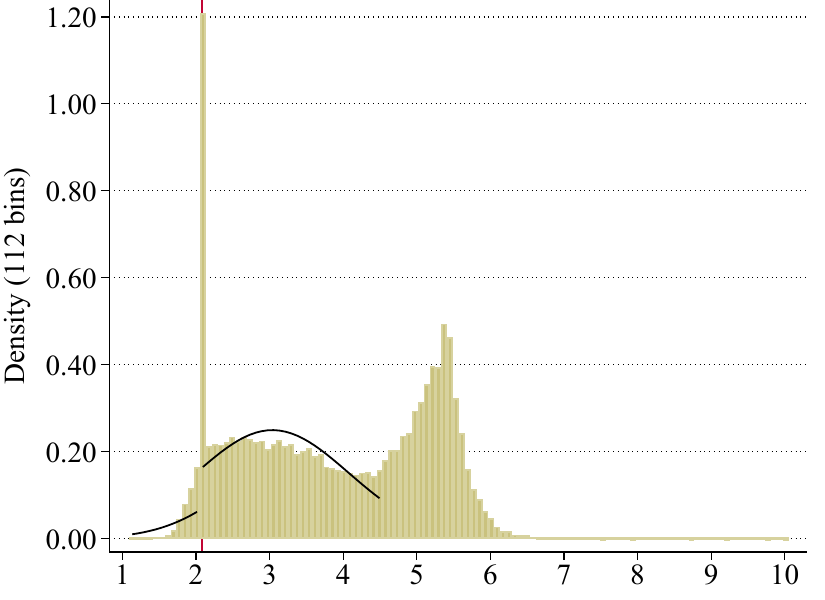}
\label{fig:simulation1:h}
\end{subfigure}
\begin{subfigure}[b]{0.22\linewidth}
\centering
\caption{20\% of the data used}
\includegraphics[width=\linewidth]{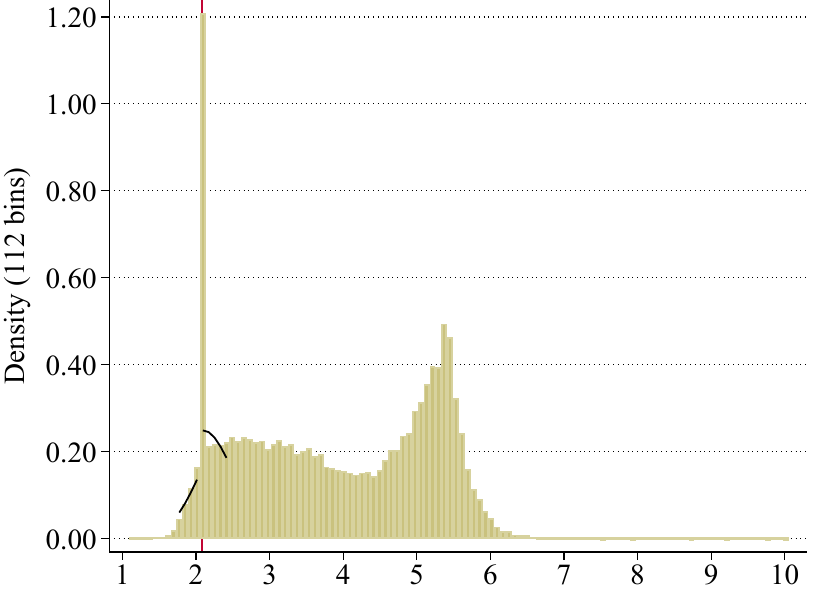}
\label{fig:simulation1:i}
\end{subfigure}
\begin{subfigure}[b]{0.22\linewidth}
\centering
\caption{Elasticity by percent used}
\includegraphics[width=\linewidth]{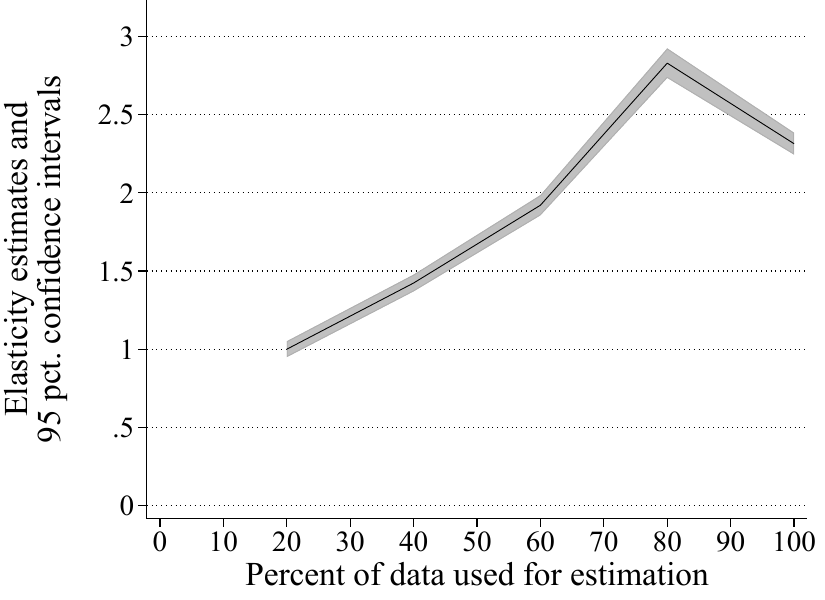}
\label{fig:simulation1:j}
\end{subfigure}
\caption*{
\footnotesize 
\textit{Notes:} This simulation experiment illustrates that the mid-censored Tobit model is able to fit non-normal distributions of $n^*$ and retrieve the right elasticity using covariates and truncation. 
We generate 50,000 observations of $y$ and a scalar $X$ following Experiment 1 detailed in Section \ref{sec:solutions:cov:tobit}. 
The variable $n^*$ is distributed as a mixture of two Skewed Generalized Error Distributions (Panel a), $n^*|X$ is Gaussian, 
the kink point is at $k=2.0794$, and $\eps=1$.
Panels c\textendash f correspond to estimates from a Tobit model that is correctly specified with covariate $X$.
Panels c\textendash e show the histogram of simulated data for $y$ and the best-fit Tobit distributions
for three truncation sizes.
Panel f displays the elasticity estimate as a function of the percentage of data used in each truncated estimation, along with 95\% confidence intervals.
Similarly, Panels g\textendash j correspond to estimates from a Tobit model that is incorrectly specified without covariate $X$.
}
\end{figure}


\end{landscape}

\restoregeometry

\newgeometry{headheight=1mm,tmargin=10mm,bmargin=15mm, lmargin=10mm,rmargin=10mm}

\begin{landscape}

\begin{figure}[tbp]
\caption{Robustness of Tobit Estimates to Lack of Normality\textemdash  Experiment 2}
\label{fig:simulation2}
\centering

\begin{minipage}{0.2\linewidth}
\vspace{6mm}
\centering

\begin{subfigure}[b]{\linewidth}
\centering
\caption{PDF of $n^*$}
\includegraphics[width=\linewidth]{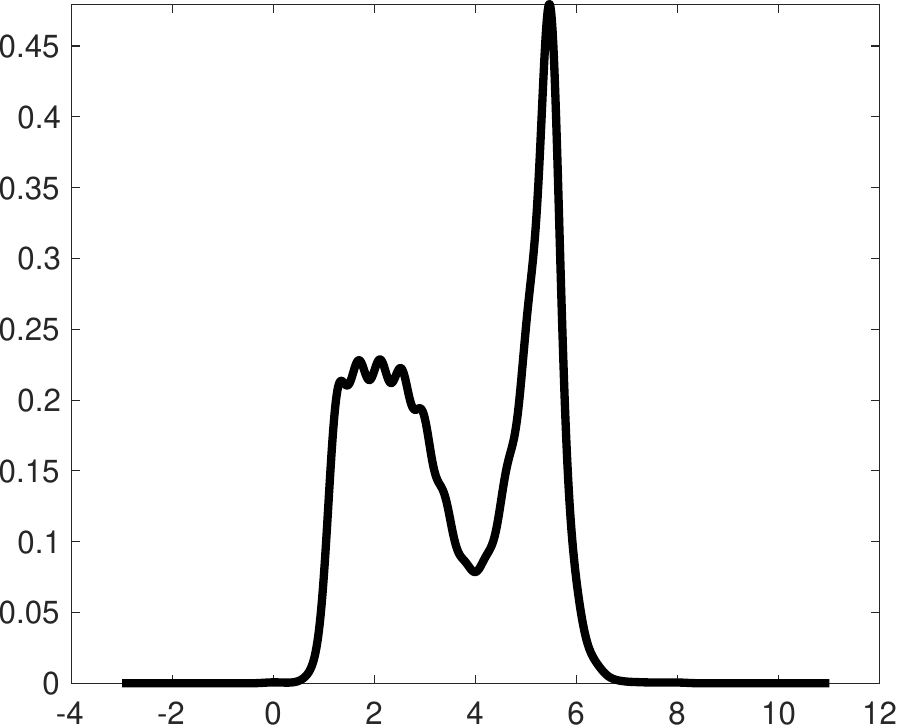}
\label{fig:simulation2:a}
\end{subfigure}

\begin{subfigure}[b]{\linewidth}
\centering
\caption{PMF of $X$}
\includegraphics[width=\linewidth]{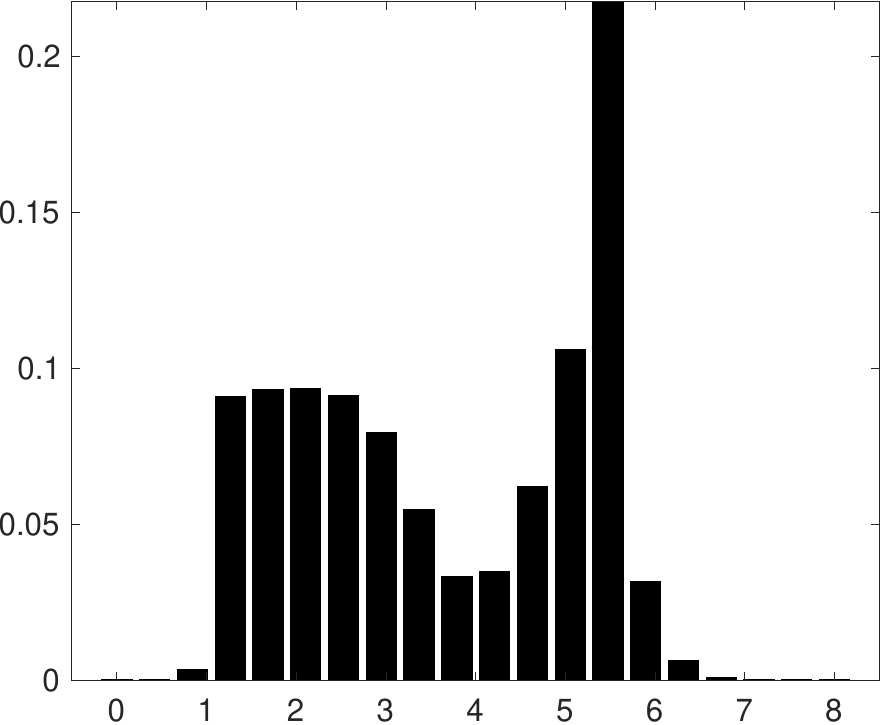}
\label{fig:simulation2:b}
\end{subfigure}

\begin{subfigure}[b]{\linewidth}
\centering
\caption{100\% of the data used}
\includegraphics[width=\linewidth]{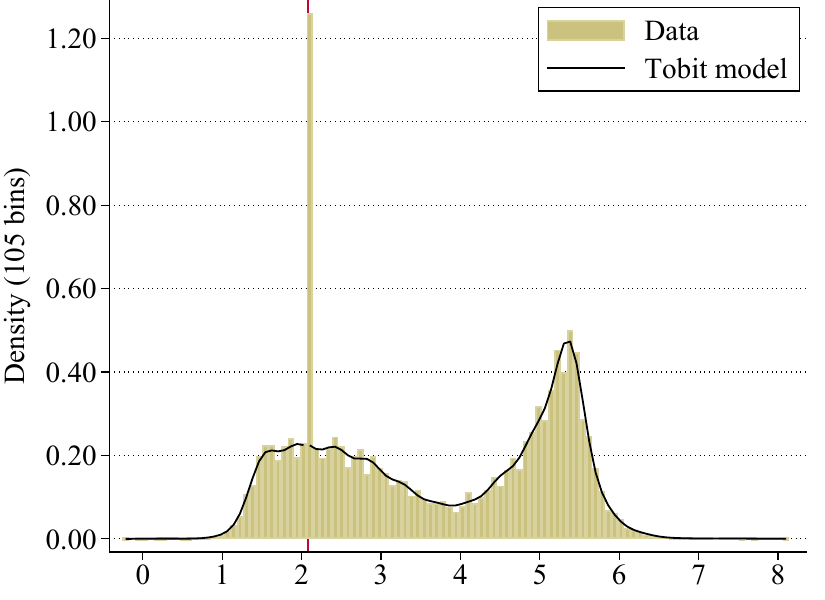}
\label{fig:simulation2:c}
\end{subfigure}

\end{minipage}
\hspace{5mm}
\begin{minipage}{0.67\linewidth}
\centering

\begin{subfigure}[c]{\linewidth}
\centering
\caption{Conditional Probability Density Functions $n^* | X=x$}
\includegraphics[width=\linewidth]{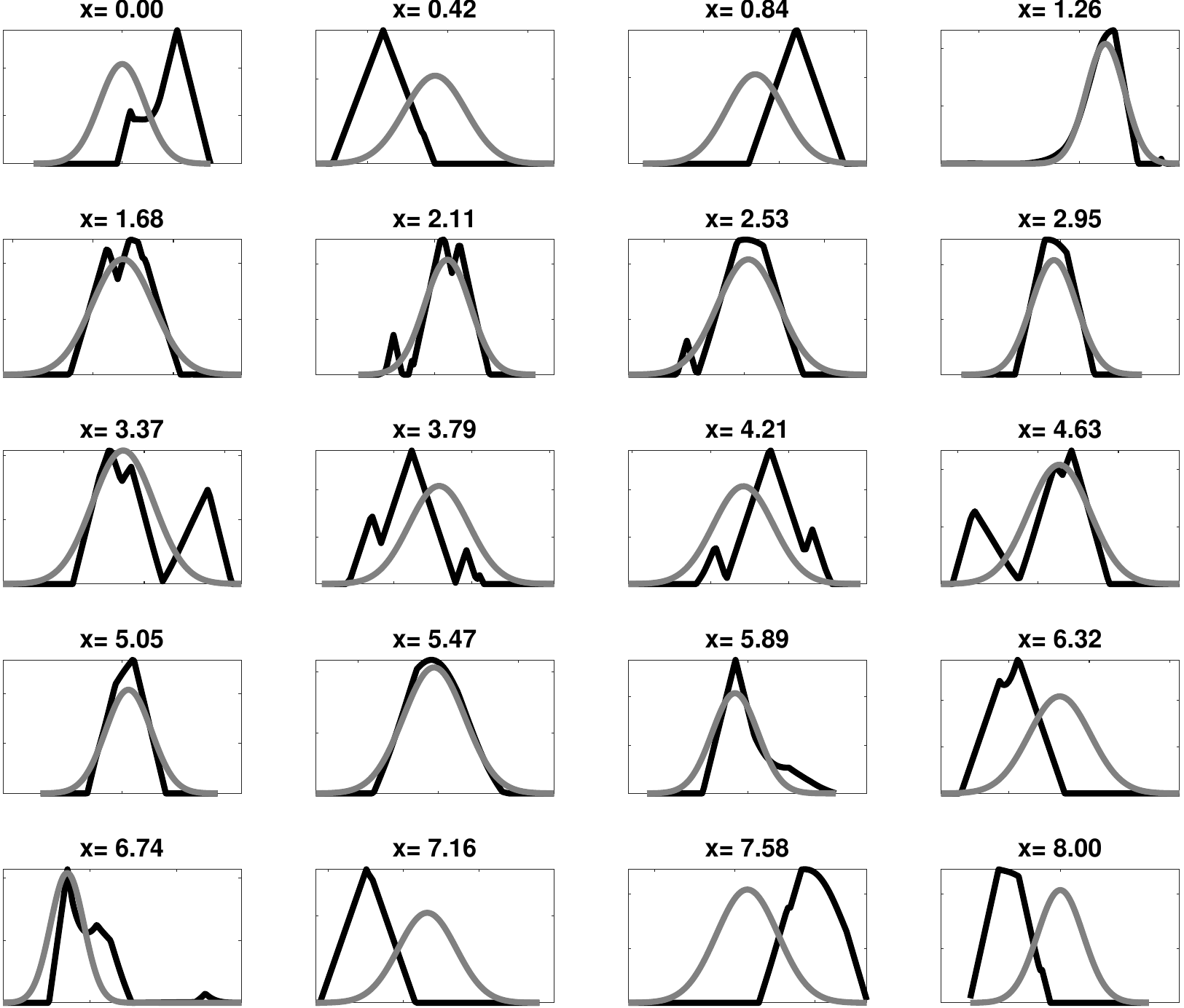}
\label{fig:simulation2:d}
\end{subfigure}

\end{minipage}

\caption*{
\footnotesize 
\textit{Notes:} 
This simulation experiment illustrates that the mid-censored Tobit model is able to fit non-normal distributions of $n^*$ and retrieve the right elasticity even when the conditional distribution $n^*|X$ is not Gaussian and there is no truncation.
We generate 50,000 observations of $y$ and a scalar $X$ following Experiment 2 detailed in Section \ref{sec:solutions:cov:tobit}. 
The variable $n^*$ is approximately a mixture of two Skewed Generalized Error Distributions (Panel a), the distribution of $X$ is discrete (Panel b), 
the kink point is at $k=2.0794$, and $\eps=1$.
Panel c shows the histogram of simulated data for $y$ and the best-fit Tobit distribution using covariate $X$ and no truncation.
Panel d displays the true conditional PDFs of $n^*|X$ in black along with the Gaussian PDFs in gray that are assumed by the Tobit model.
The elasticity is estimated at $\ha\eps = 1.0083$ (S.E. 0.0073).
}
\end{figure}
\end{landscape}

\restoregeometry

\clearpage
\begin{landscape}
\begin{table}
\caption{Estimates Using U.S. Tax Returns 1995--2004}%
\label{tbl:tobit}
\centering
\scalebox{0.94}{
\begin{tabular}{lccccccc|lc}
\hline\hline
 & (1) & (2) & (3) & (4) & (5) & (6) & (7) & \multicolumn{2}{c}{(8)} \\ 
Statistical Model & Trapezoidal & Theorem \ref{theo_partial} & Theorem \ref{theo_partial} & Tobit & Tobit & Tobit & Tobit & \\ 
 & Approximation & Bounds & Bounds & Full Sample & Trunc. 75\% & Trunc. 50\% & Trunc. 25\% & \multicolumn{2}{c}{Sample} \\ 
 & & M = 0.5 & M = 1 & & & & & \multicolumn{2}{c}{details} \\ \hline 
\textit{All} & & & & & & & & Obs. & 188.3m \\ 
\hspace{2mm} Elasticity $\left(\varepsilon\right)$ & 0.323 & $\left[ 0.297, 0.364 \right]$ & $\left[ 0.276, 0.448 \right]$ & 0.163 & 0.242 & 0.250 & 0.278 & Avg. & \$53.5k \\ 
 &  & & & (0.0001) & (0.0002) & (0.0002) & (0.0002) & Std. & \$64.6k \\ 
 & & & & & & & & \\ 
\textit{Self-employed} & & & & & & & & Obs. & 33.4m \\ 
\hspace{2mm} Elasticity $\left(\varepsilon\right)$ & 0.811 & $\left[0.686,1.183\right]$ & $\left[0.612,\infty \right]$ & 0.568 & 0.752 & 0.746 & 0.750 & Avg. & \$60.7k \\ 
 & & & & (0.0005) & (0.0007) & (0.0008) & (0.0008) & Std. & \$77.2k \\ 
 \textit{Self-employed, } & & & & & & & & & \\ 
 \textit{married} & & & & & & & & Obs. & 23.9m \\ 
\hspace{2mm} Elasticity $\left(\varepsilon\right)$ & 0.770 & $\left[0.563,\infty \right]$ & $\left[0.475, \infty \right]$ & 0.304 & 0.464 & 0.540 & 0.554 & Avg. & \$73.6k \\ 
 &  & & & (0.0006) & (0.0009) & (0.0010) & (0.0010) & Std. & \$84.4k \\ 
\textit{Self-employed,} & & & & & & & & \\ 
\textit{not married} & & & & & & & & Obs. & 9.5m \\ 
\hspace{2mm} Elasticity $\left(\varepsilon\right)$ & 0.842 & $\left[0.777, 0.936 \right]$ & $\left[0.728, 1.107\right]$ & 0.875 & 0.774 & 0.745 & 0.807 & Avg. & \$28.3k \\ 
 &  & & & (0.0010) & (0.0009) & (0.0010) & (0.0015) & Std. & \$39.6k \\ 

 & & & & & & & & \\ \hline
\end{tabular}
}
\exhibitnote{\textit{Notes:} The table shows estimates of the elasticity for four different subsamples of the IRS data, and using three different approaches discussed in the paper. 
The first approach (column 1) uses the trapezoidal approximation to point-identify the elasticity (Example \ref{example1}). 
We obtained non-parametric estimates of the side limits of $f_y$ at the kink using the method of \cite{cattaneo2019}.
The estimate for the bunching mass equals the sample proportion of $y$ observations that equals the kink point (see discussion on friction errors in Section \ref{sec:application:data}).
The second approach (columns 2 and 3) uses the same estimates of the bunching mass and side limits to compute partially identified sets for the elasticity (Theorem \ref{theo_partial}).
Upper and lower bounds are calculated for two choices of M, that is, the maximum slope of the PDF of the unobserved heterogeneity $n^*$.
Column 4 has Tobit MLE estimates of the elasticity that utilizes the full sample of data, along with robust standard errors in parentheses.
Columns 5 through 7 report truncated Tobit MLE estimates.
As we move from column 5 to column 7, we restrict the estimation sample to shrinking symmetric windows around the kink that utilizes 75\% to 25\% of the data.
The set of covariates that enters the Tobit estimation is kept constant across different truncation windows and are listed in Section \ref{sec:estimates_across_methods}.
}
\end{table}
\end{landscape}

\begin{figure}[tbp]
\caption{Partial Identification Bounds for the Elasticity}
\label{figure:bounds_est}
\centering

\begin{subfigure}[b]{0.495\textwidth}
\centering
\caption{All Filers}
\includegraphics[width=\linewidth]{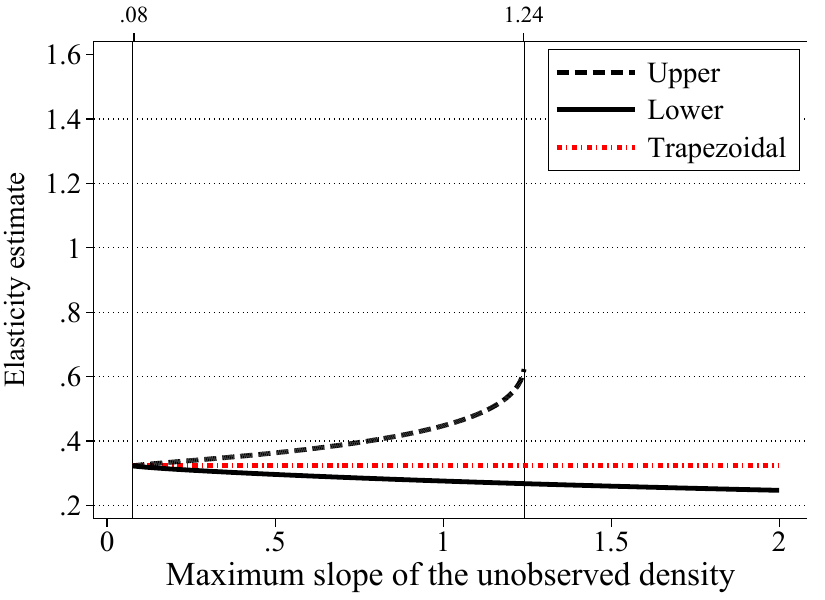}
\label{fig:bounds_est_panela}
\end{subfigure}
\begin{subfigure}[b]{0.495\textwidth}
\centering
\caption{Self-Employed Filers}
\includegraphics[width=\linewidth]{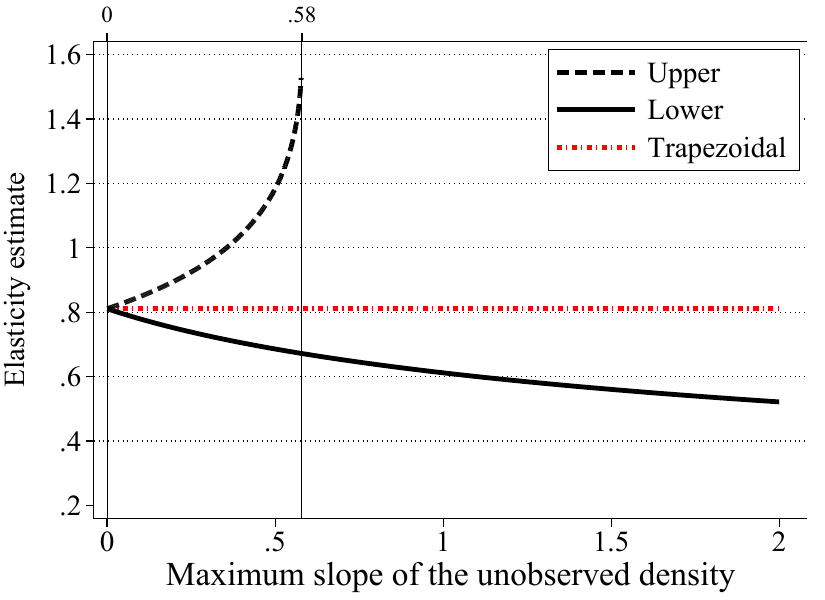}
\label{fig:bounds_est_panelb}
\end{subfigure}
\begin{subfigure}[b]{0.495\textwidth}
\centering
\caption{Self-employed and Married Filers}
\includegraphics[width=\linewidth]{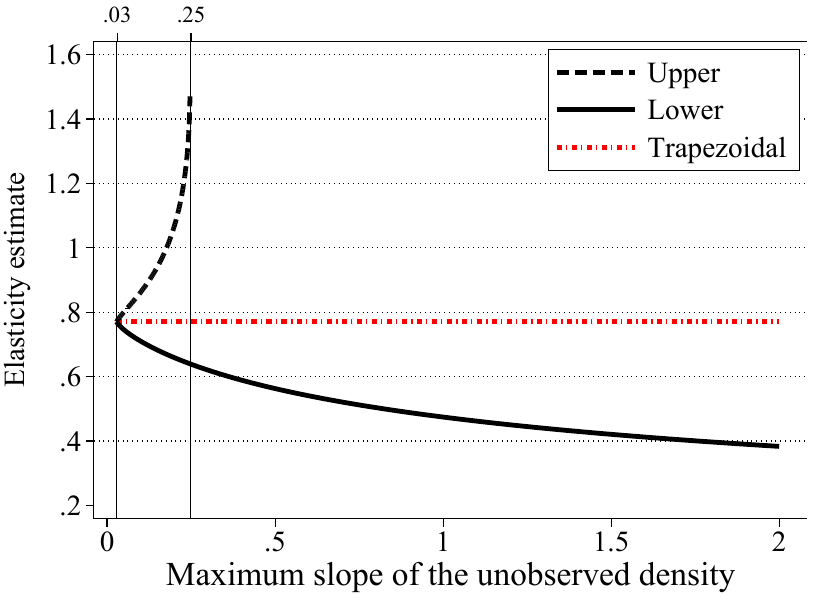}
\label{fig:bounds_est_panelc}
\end{subfigure}
\begin{subfigure}[b]{0.495\textwidth}
\centering
\caption{Self-employed and Not Married Filers}
\includegraphics[width=\linewidth]{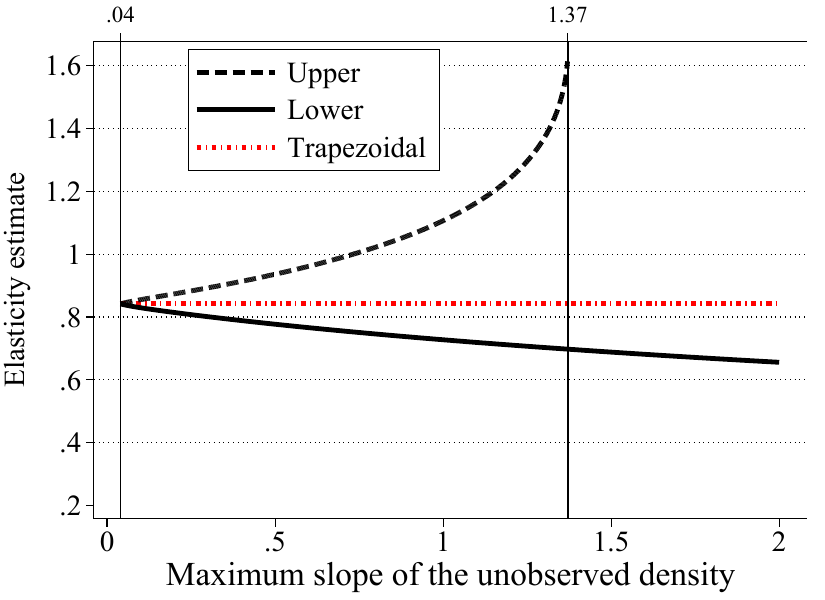}
\label{fig:bounds_est_paneld}
\end{subfigure}
\caption*{
\footnotesize 
\textit{Notes:} Panels a through d display partially identified sets for the elasticity for all filers with one child, and three other subsamples defined by employment and marital status.
The y-axis has elasticity values between lower and upper bounds given various choices of $M$ on the x-axis, that is, the maximum slope magnitude of the PDF of the unobserved heterogeneity $n^*$ 
(Theorem \ref{theo_partial}).
Each panel has two vertical lines.
The line on the left corresponds to the smallest choice  of $M$ for which the bounds are defined.
At the smallest $M$, upper and lower bounds are equal to the elasticity estimate based on the trapezoidal approximation (Example \ref{example1}).
The vertical line on the right corresponds to the largest choice of $M$ for which the upper bound is finite.
Higher slopes allow for the possibility of PDFs that are zero in the bunching window. As a result, we may have a finite bunching mass for any arbitrarily large elasticity.  
}
\end{figure}


\begin{landscape}

\begin{figure}[tbp]
\caption{Truncated Tobit -  All Filers}
\label{figure:trunc_est_all}
\centering

\begin{subfigure}[b]{0.455\textwidth}
\centering
\caption{100\% of the data used for estimation}
\includegraphics[width=\linewidth]{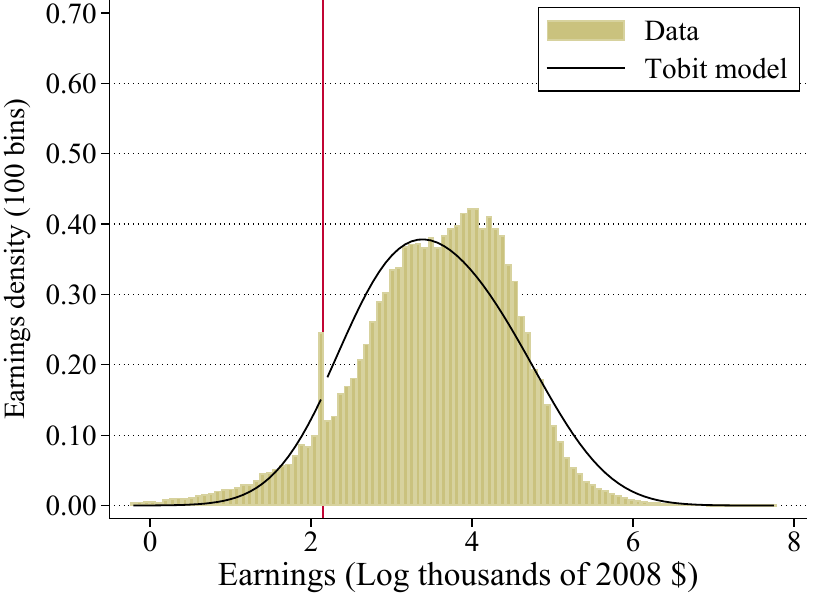}
\label{figure:trunc_est_all_1}
\end{subfigure}
\begin{subfigure}[b]{0.455\textwidth}
\centering
\caption{80\% of the data used for estimation}
\includegraphics[width=\linewidth]{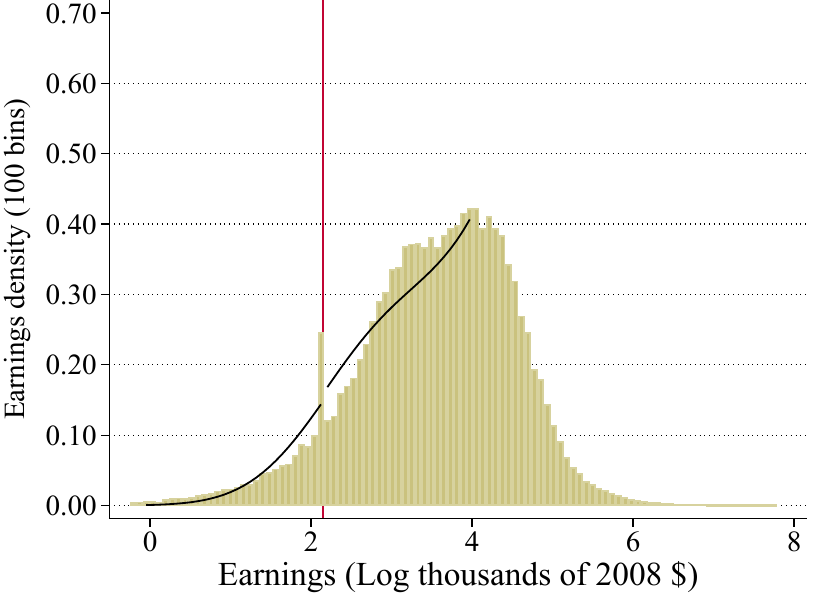}
\label{figure:trunc_est_all_2}
\end{subfigure}
\begin{subfigure}[b]{0.455\textwidth}
\centering
\caption{60\% of the data used for estimation}
\includegraphics[width=\linewidth]{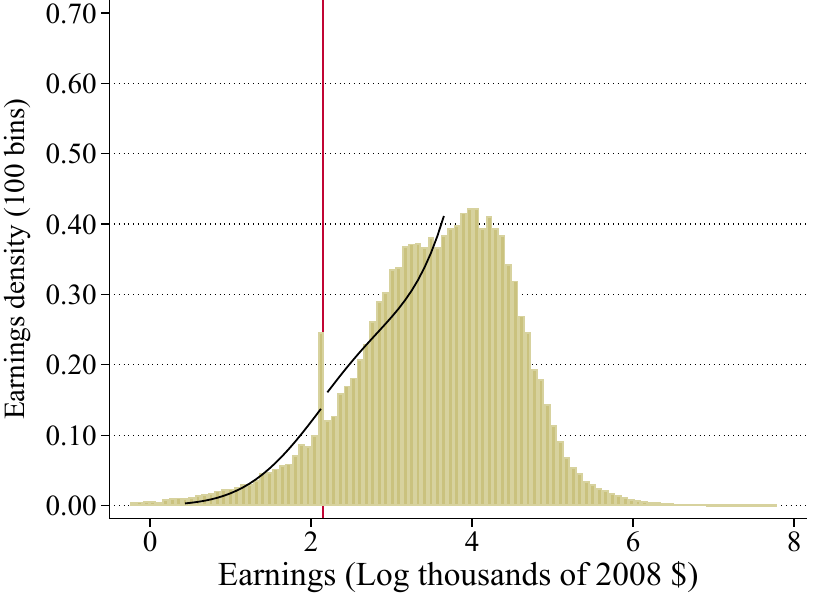}
\label{figure:trunc_est_all_3}
\end{subfigure}
\medskip

\begin{subfigure}[b]{0.455\textwidth}
\centering
\caption{40\% of the data used for estimation}
\includegraphics[width=\linewidth]{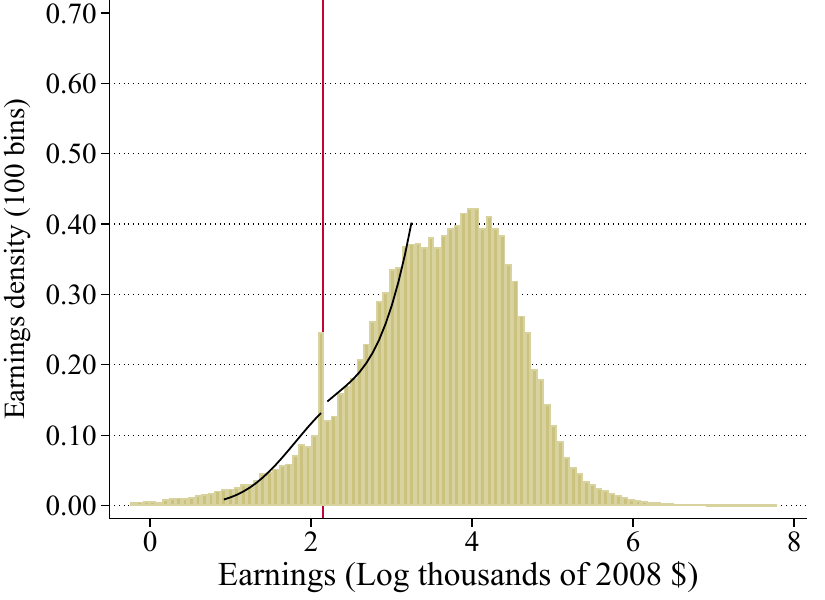}
\label{figure:trunc_est_all_4}
\end{subfigure}
\begin{subfigure}[b]{0.455\textwidth}
\centering
\caption{20\% of the data used for estimation}
\includegraphics[width=\linewidth]{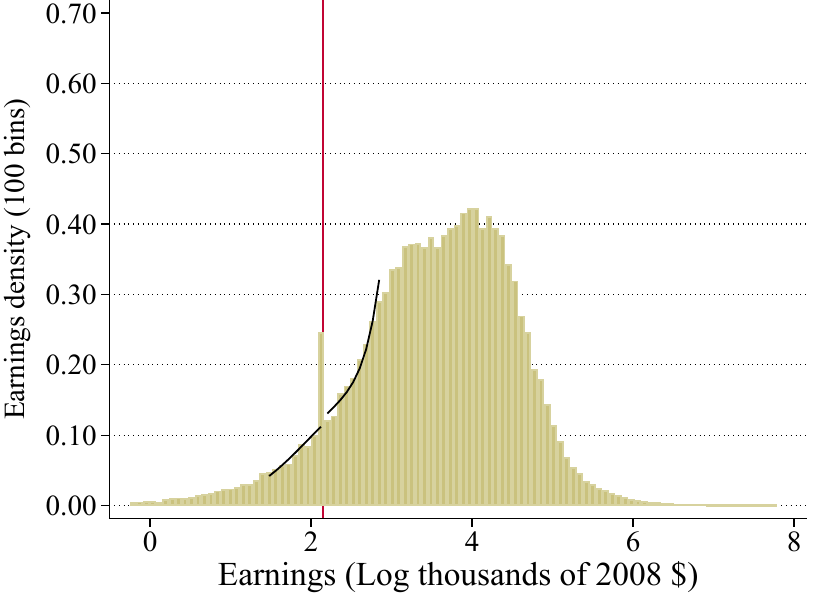}
\label{figure:trunc_est_all_5}
\end{subfigure}
\begin{subfigure}[b]{0.455\textwidth}
\centering
\caption{Elasticity by percent used}
\includegraphics[width=\linewidth]{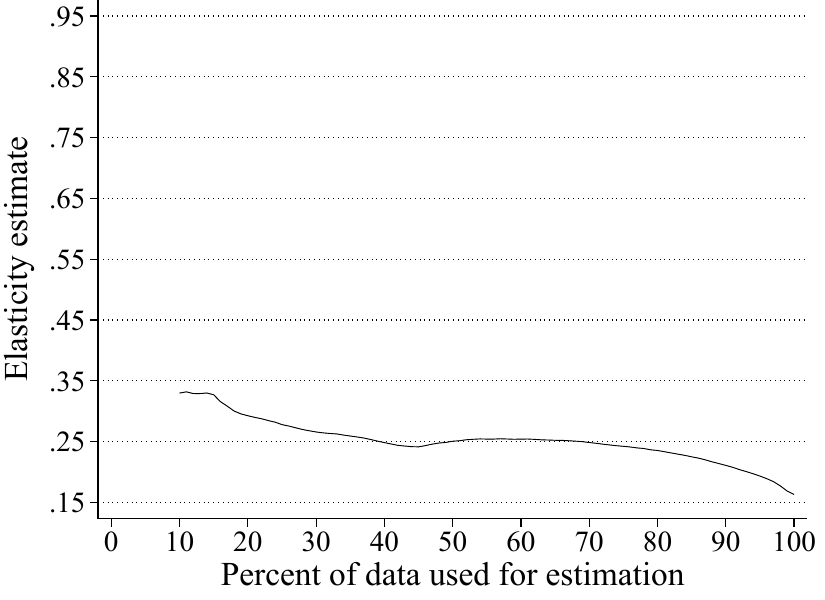}
\label{figure:trunc_est_all_6}
\end{subfigure}
\caption*{
\footnotesize 
\textit{Notes:} the figure displays best-fit Tobit distributions and elasticity estimates for various choices of a symmetric truncation window around the kink point.
The set of covariates that enters the Tobit estimation is kept constant across different truncation windows and are listed in Section \ref{sec:estimates_across_methods}.
Panels a through e show the histogram of income for all filers  (bars), along with the best-fit Tobit PDF for each truncation window (line).
The best-fit PDF is constructed using the truncated Tobit likelihood averaged over  covariate values in the sample.
Panel f displays the Tobit elasticity estimate as a function of the percentage of data used in estimation. 
}
\vspace{-41pt}
\end{figure}

\end{landscape}

\begin{landscape}

\begin{figure}[tbp]
\caption{Truncated Tobit -  Self-employed Filers}
\label{figure:trunc_est_self}
\centering

\begin{subfigure}[b]{0.455\textwidth}
\centering
\caption{100\% of the data used for estimation}
\includegraphics[width=\linewidth]{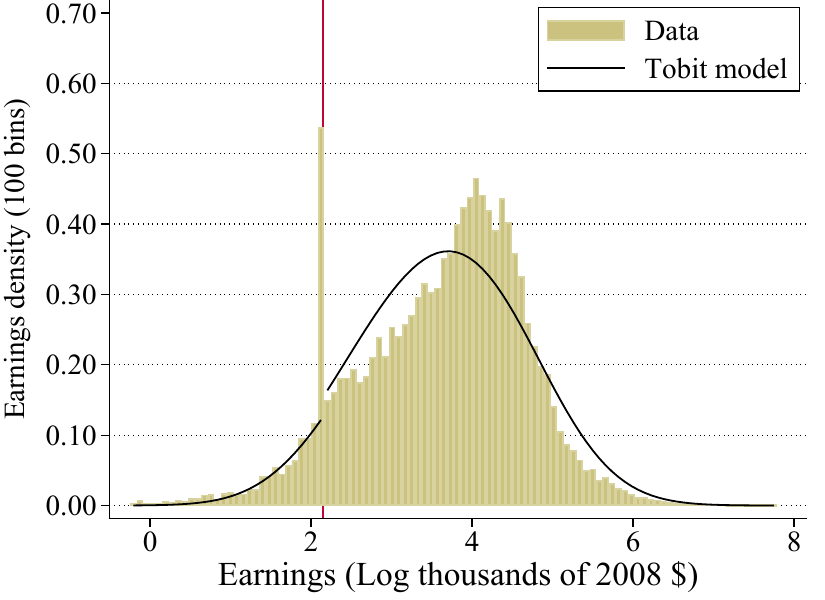}
\label{figure:trunc_est_self_1}
\end{subfigure}
\begin{subfigure}[b]{0.455\textwidth}
\centering
\caption{80\% of the data used for estimation}
\includegraphics[width=\linewidth]{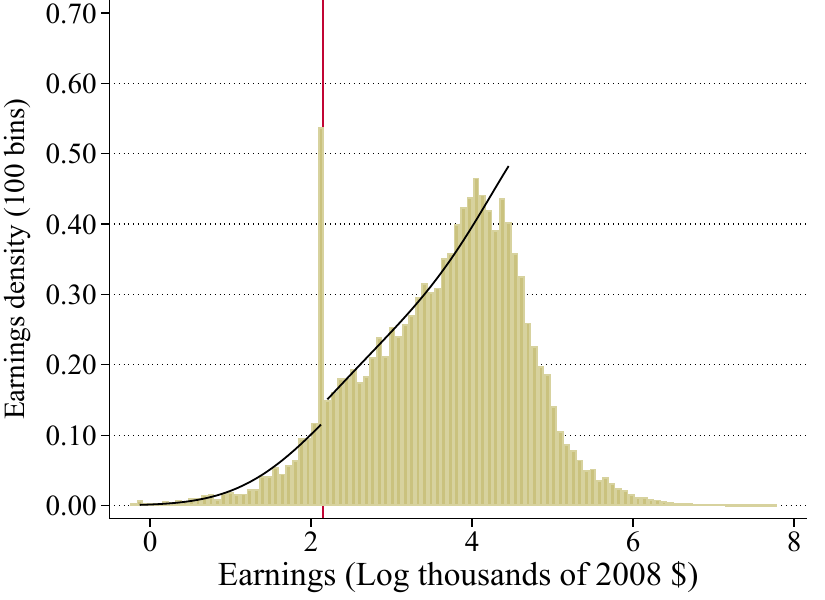}
\label{figure:trunc_est_self_2}
\end{subfigure}
\begin{subfigure}[b]{0.455\textwidth}
\centering
\caption{60\% of the data used for estimation}
\includegraphics[width=\linewidth]{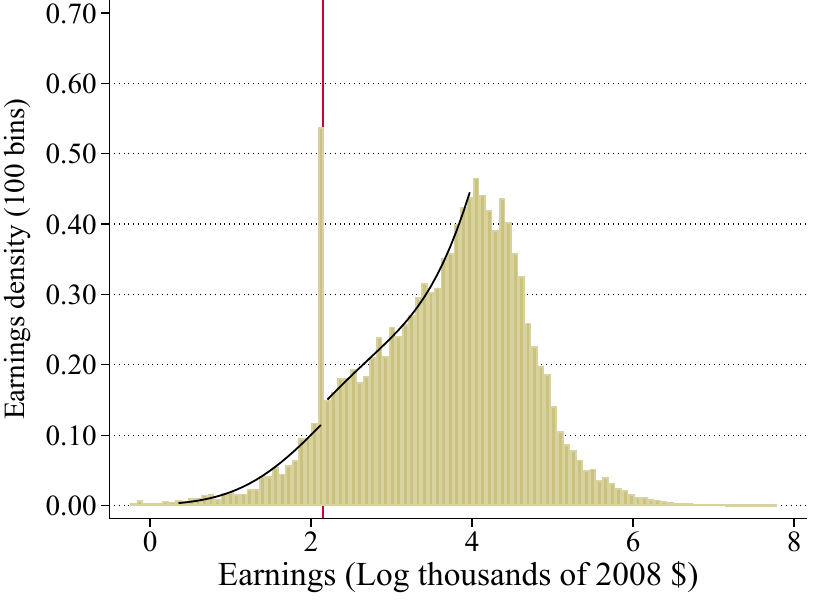}
\label{figure:trunc_est_self_3}
\end{subfigure}
\medskip

\begin{subfigure}[b]{0.455\textwidth}
\centering
\caption{40\% of the data used for estimation}
\includegraphics[width=\linewidth]{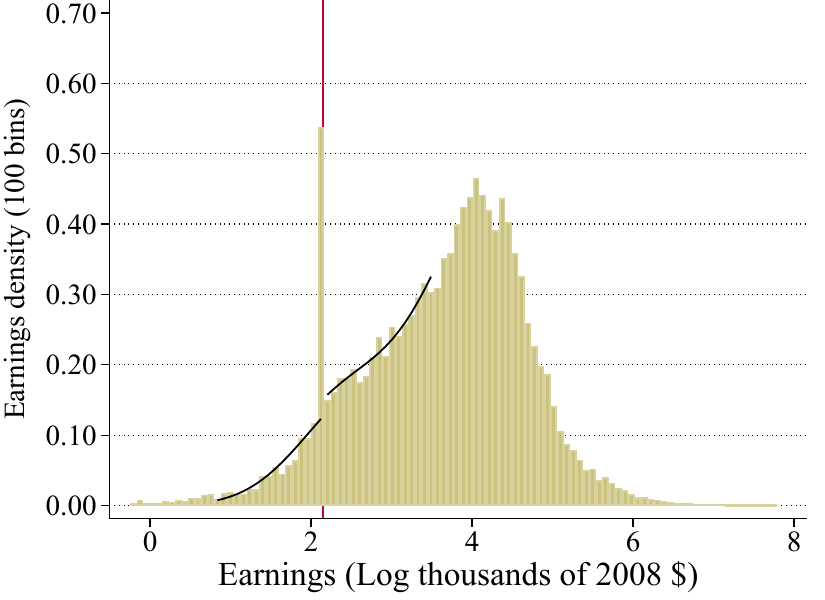}
\label{figure:trunc_est_self_4}
\end{subfigure}
\begin{subfigure}[b]{0.455\textwidth}
\centering
\caption{20\% of the data used for estimation}
\includegraphics[width=\linewidth]{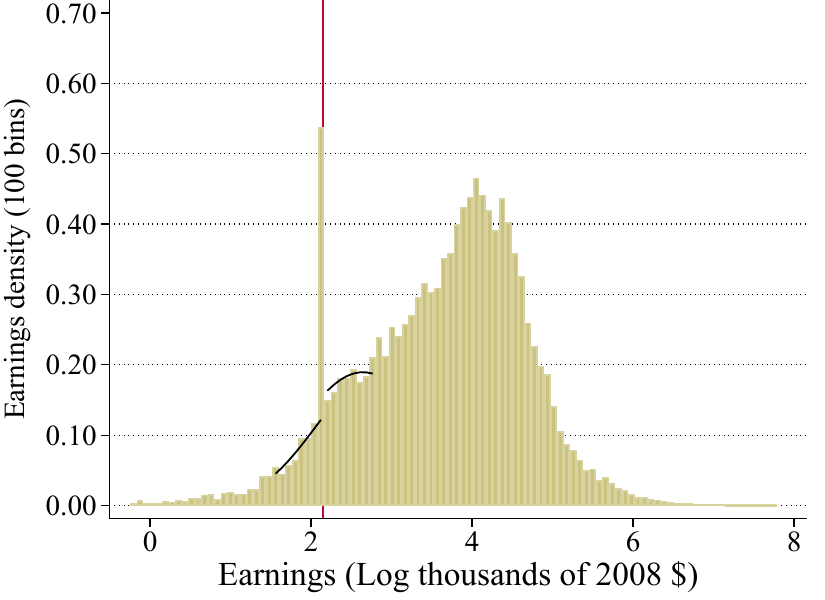}
\label{figure:trunc_est_self_5}
\end{subfigure}
\begin{subfigure}[b]{0.455\textwidth}
\centering
\caption{Elasticity by percent used}
\includegraphics[width=\linewidth]{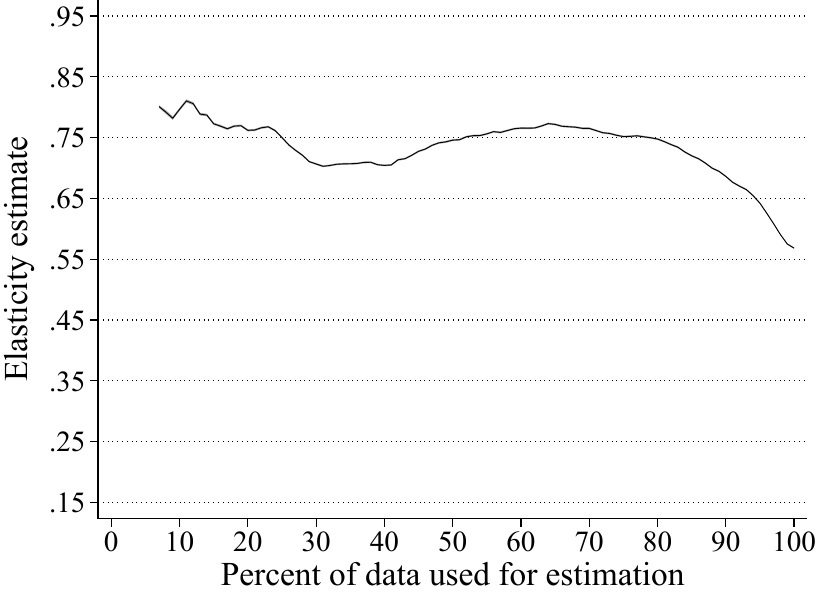}
\label{figure:trunc_est_self_6}
\end{subfigure}
\caption*{
\footnotesize 
\textit{Notes:} 
the figure displays best-fit Tobit distributions and elasticity estimates for various choices of a symmetric truncation window around the kink point.
The set of covariates that enters the Tobit estimation is kept constant across different truncation windows and are listed in Section \ref{sec:estimates_across_methods}.
Panels a through e show the histogram of income for self-employed filers  (bars), along with the best-fit Tobit PDF for each truncation window (line).
The best-fit PDF is constructed using the truncated Tobit likelihood averaged over  covariate values in the sample.
Panel f displays the Tobit elasticity estimate as a function of the percentage of data used in estimation. 
\vspace{-41pt}
}
\end{figure}

\end{landscape}

\begin{landscape}

\begin{figure}[tbp]
\caption{Truncated Tobit -  Self-employed and Married Filers}
\label{figure:trunc_est_self_married}
\centering

\begin{subfigure}[b]{0.455\textwidth}
\centering
\caption{100\% of the data used for estimation}
\includegraphics[width=\linewidth]{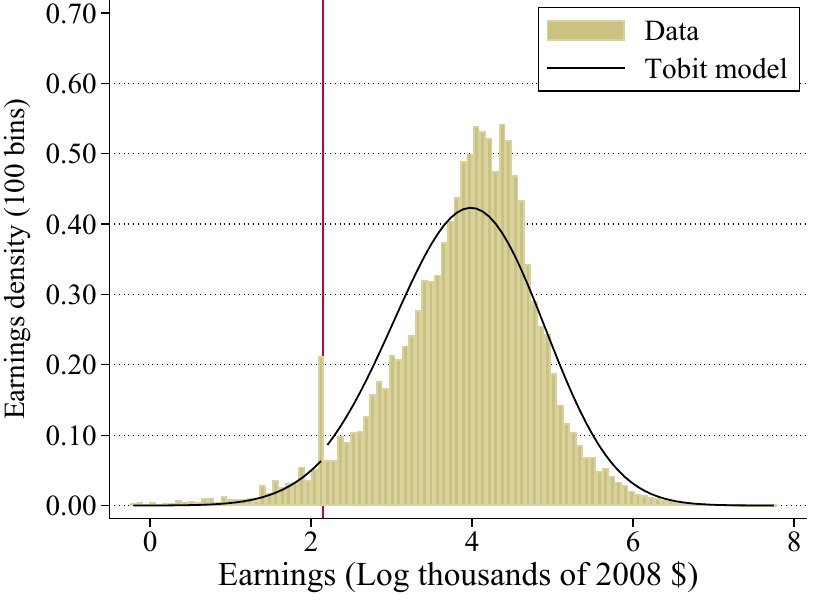}
\label{fig:trunc_est_self_married_1}
\end{subfigure}
\begin{subfigure}[b]{0.455\textwidth}
\centering
\caption{80\% of the data used for estimation}
\includegraphics[width=\linewidth]{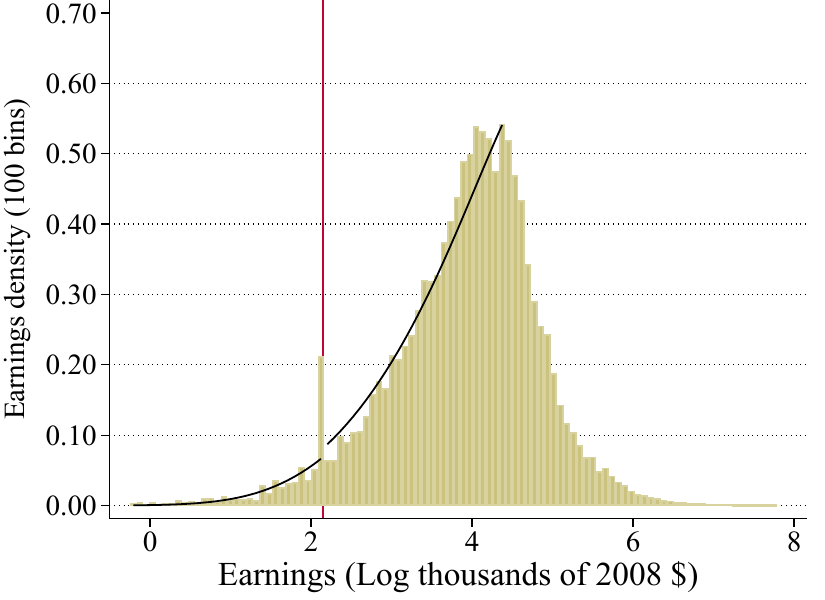}
\label{fig:trunc_est_self_married_2}
\end{subfigure}
\begin{subfigure}[b]{0.455\textwidth}
\centering
\caption{60\% of the data used for estimation}
\includegraphics[width=\linewidth]{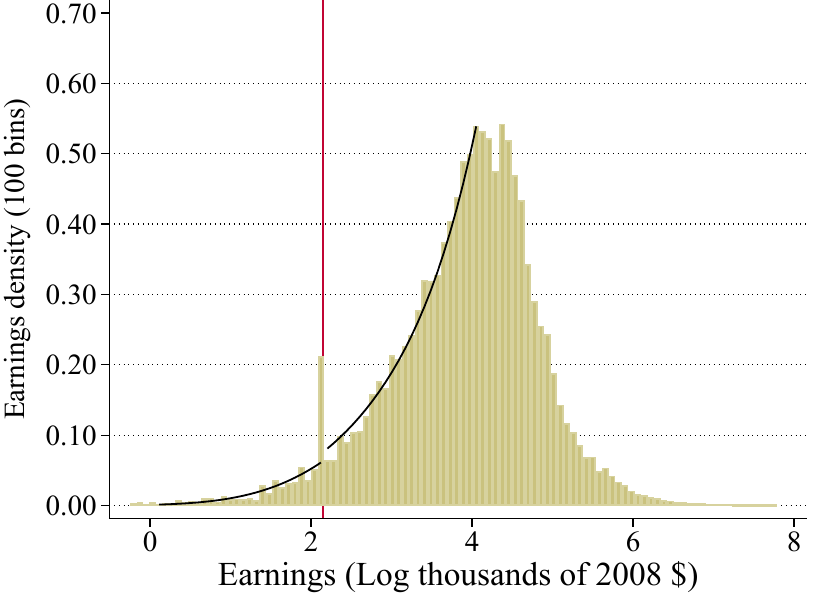}
\label{fig:trunc_est_self_married_3}
\end{subfigure}
\medskip

\begin{subfigure}[b]{0.455\textwidth}
\centering
\caption{40\% of the data used for estimation}
\includegraphics[width=\linewidth]{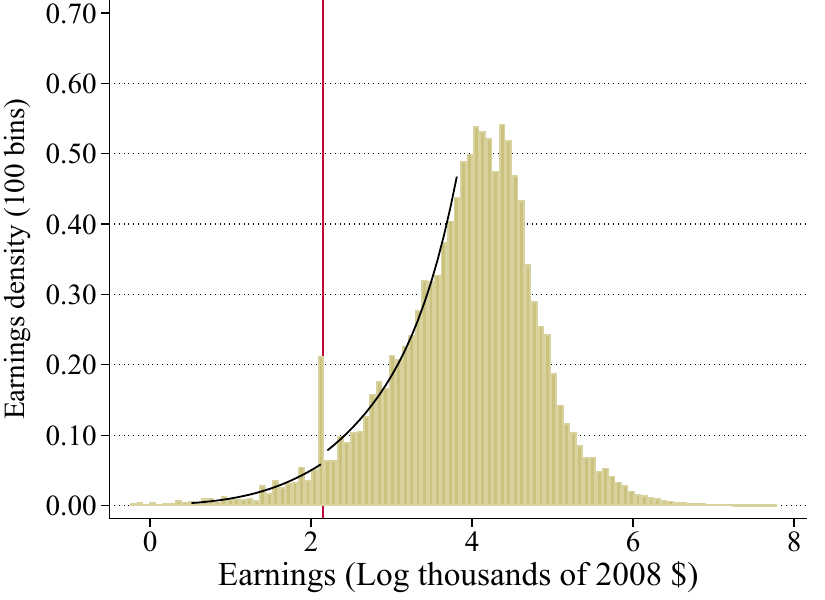}
\label{fig:trunc_est_self_married_4}
\end{subfigure}
\begin{subfigure}[b]{0.455\textwidth}
\centering
\caption{20\% of the data used for estimation}
\includegraphics[width=\linewidth]{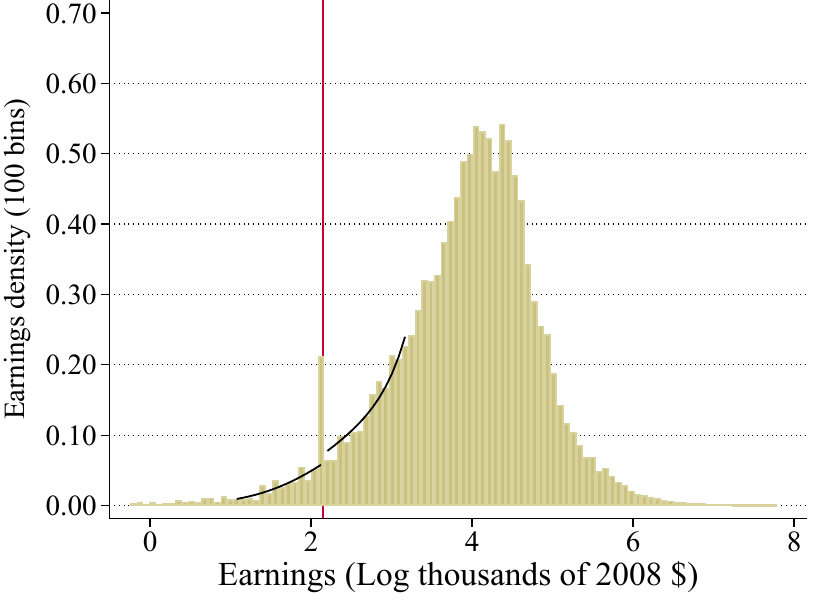}
\label{fig:trunc_est_self_married_5}
\end{subfigure}
\begin{subfigure}[b]{0.455\textwidth}
\centering
\caption{Elasticity by percent used}
\includegraphics[width=\linewidth]{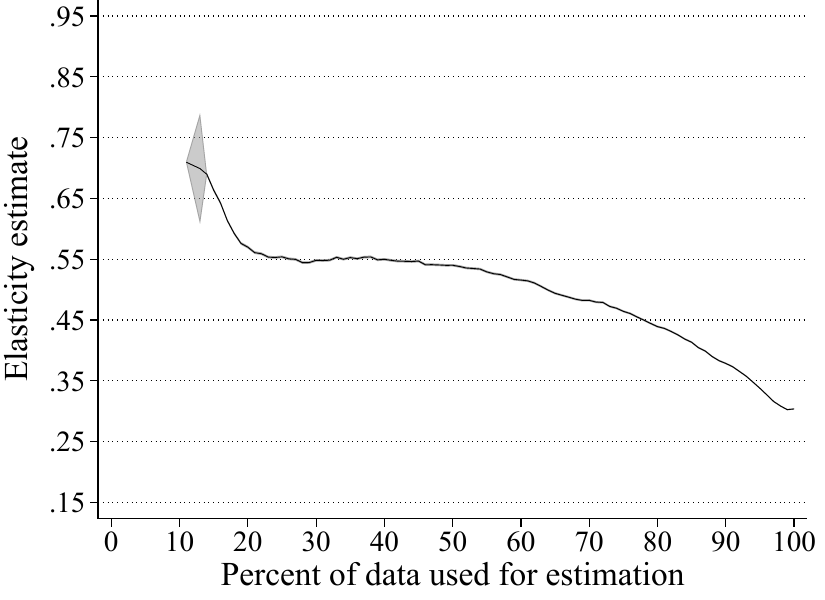}
\label{fig:trunc_est_self_married_6}
\end{subfigure}
\caption*{
\footnotesize 
\textit{Notes:} 
the figure displays best-fit Tobit distributions and elasticity estimates for various choices of a symmetric truncation window around the kink point.
The set of covariates that enters the Tobit estimation is kept constant across different truncation windows and are listed in Section \ref{sec:estimates_across_methods}. 
Panels a through e show the histogram of income for self-employed and married  filers  (bars), along with the best-fit Tobit PDF for each truncation window (line).
The best-fit PDF is constructed using the truncated Tobit likelihood averaged over  covariate values in the sample.
Panel f displays the Tobit elasticity estimate as a function of the percentage of data used in estimation. 
\vspace{-41pt}}
\end{figure}

\end{landscape}

\begin{landscape}

\begin{figure}[tbp]
\caption{Truncated Tobit - Self-employed and Not Married Filers}
\label{figure:trunc_est_self_notmarried}
\centering

\begin{subfigure}[b]{0.455\textwidth}
\centering
\caption{100\% of the data used for estimation}
\includegraphics[width=\linewidth]{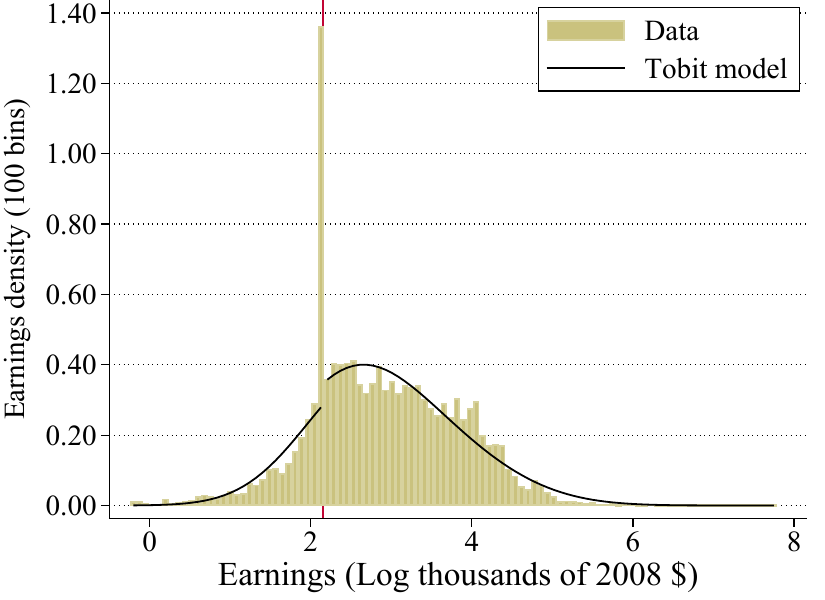}
\label{fig:trunc_est_self_notmarried_1}
\end{subfigure}
\begin{subfigure}[b]{0.455\textwidth}
\centering
\caption{80\% of the data used for estimation}
\includegraphics[width=\linewidth]{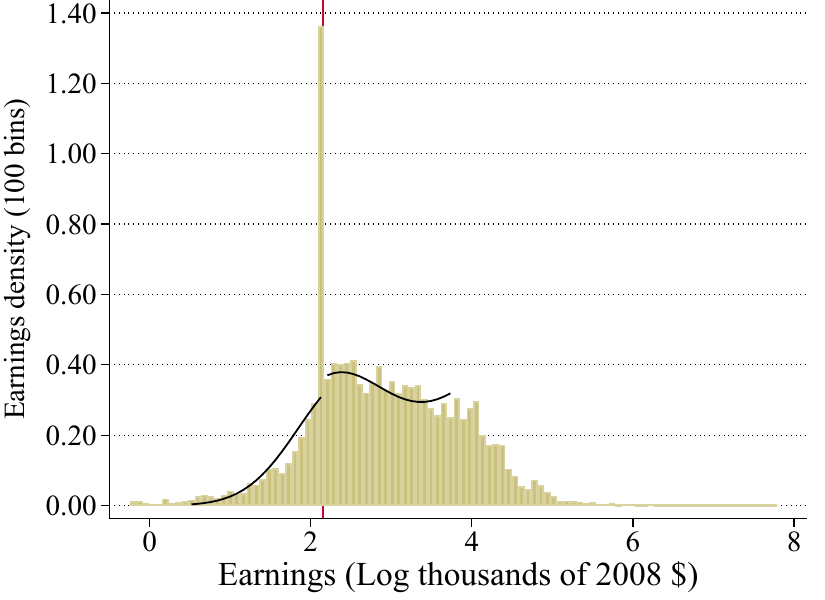}
\label{fig:trunc_est_self_notmarried_2}
\end{subfigure}
\begin{subfigure}[b]{0.455\textwidth}
\centering
\caption{60\% of the data used for estimation}
\includegraphics[width=\linewidth]{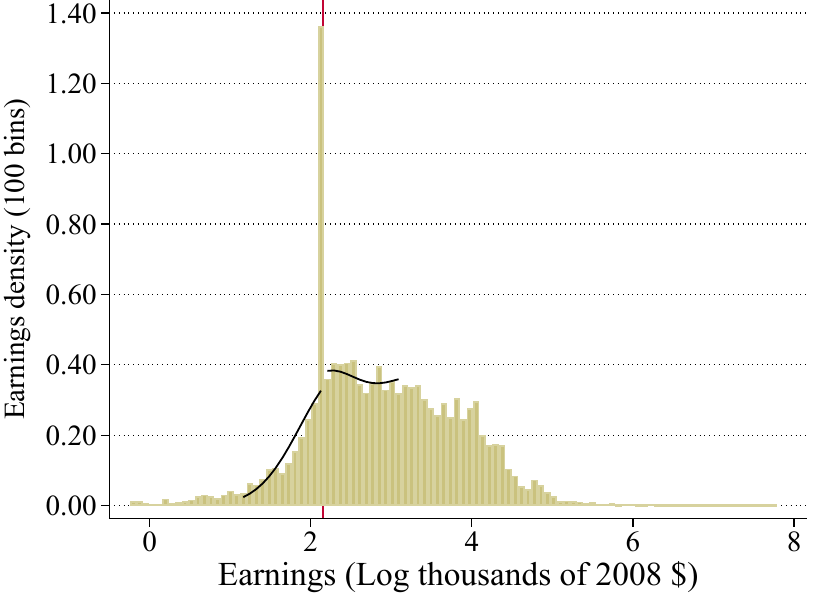}
\label{fig:trunc_est_self_notmarried_3}
\end{subfigure}
\medskip

\begin{subfigure}[b]{0.455\textwidth}
\centering
\caption{40\% of the data used for estimation}
\includegraphics[width=\linewidth]{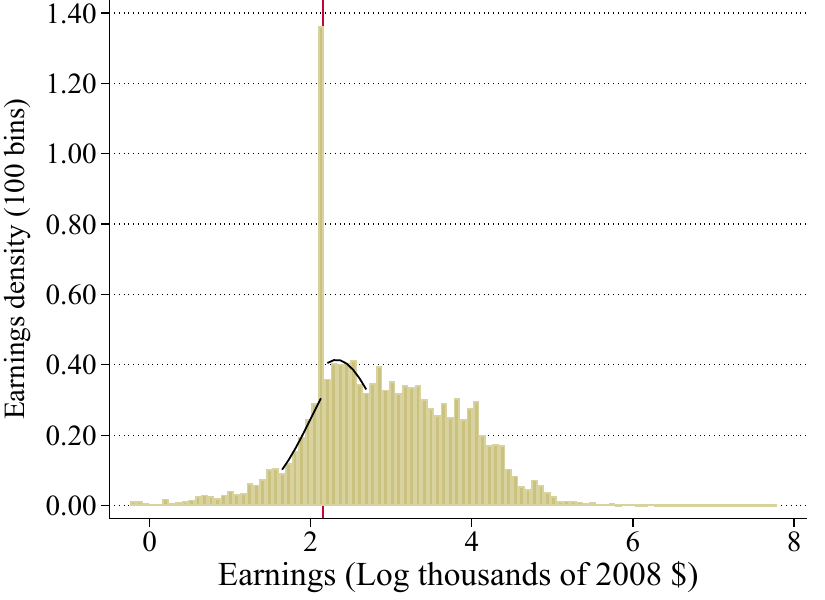}
\label{fig:trunc_est_self_notmarried_4}
\end{subfigure}
\begin{subfigure}[b]{0.455\textwidth}
\centering
\caption{19\% of the data used for estimation}
\includegraphics[width=\linewidth]{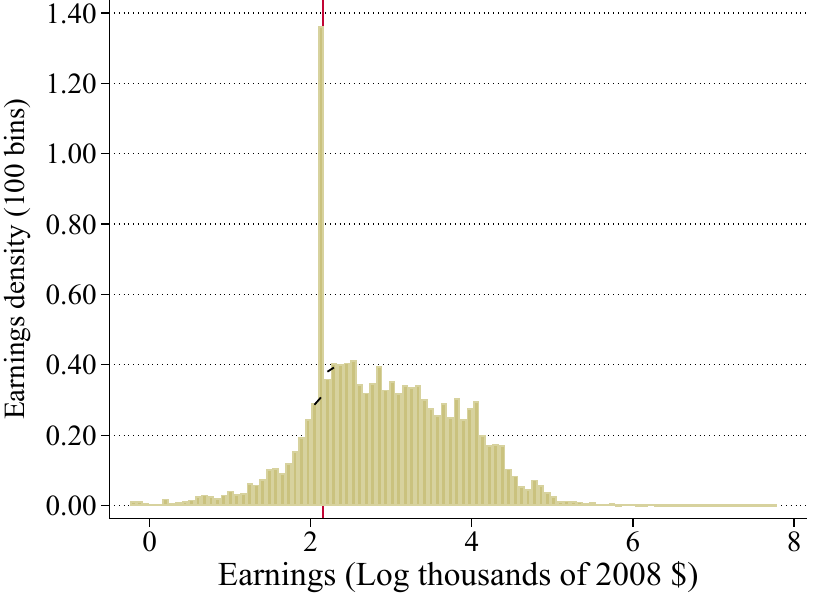}
\label{fig:trunc_est_self_notmarried_5}
\end{subfigure}
\begin{subfigure}[b]{0.455\textwidth}
\centering
\caption{Elasticity by percent used}
\includegraphics[width=\linewidth]{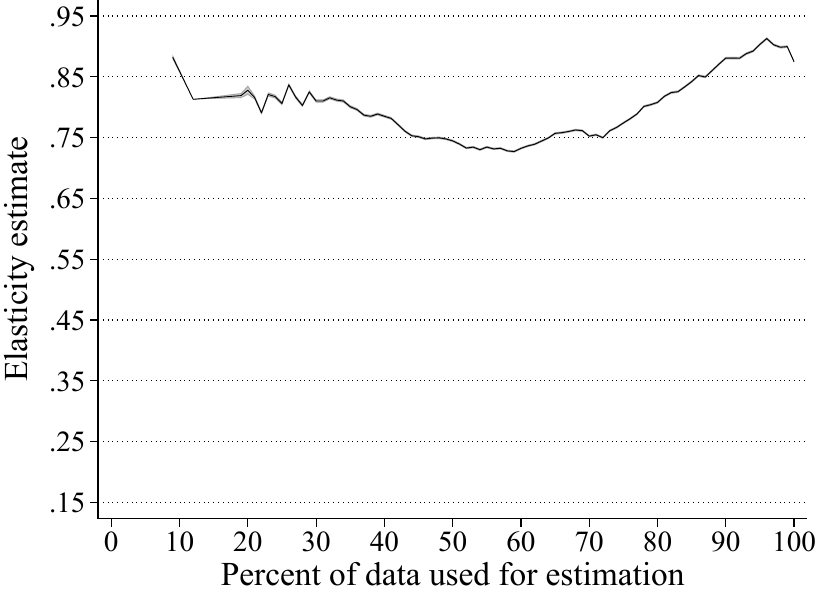}
\label{fig:trunc_est_self_notmarried_6}
\end{subfigure}
\caption*{
\footnotesize 
\textit{Notes:} the figure displays best-fit Tobit distributions and elasticity estimates for various choices of a symmetric truncation window around the kink point.
The set of covariates that enters the Tobit estimation is kept constant across different truncation windows and are listed in Section \ref{sec:estimates_across_methods}.
Panels a through e show the histogram of income for self-employed and not married filers  (bars), along with the best-fit Tobit PDF for each truncation window (line).
The best-fit PDF is constructed using the truncated Tobit likelihood averaged over  covariate values in the sample.
Panel f displays the Tobit elasticity estimate as a function of the percentage of data used in estimation. 
\vspace{-41pt}}
\end{figure}

\end{landscape}

\appendix
\singlespacing

\numberwithin{equation}{section}
\numberwithin{lemma}{section}
\numberwithin{theorem}{section}

\section{Appendix}
\label{app}

\subsection{Identification with a Notch - Proof of Theorem \ref{theo:notch-poss}}
\label{app:notch-poss}

\indent 

We present the proof of Theorem \ref{theo:notch-poss} in the more general case of multiple tax changes with at least one notch (Sections \ref{sec:app:general_prob} and \ref{sec:app:general_prob_sol} in the supplemental appendix).
Let $p \in \{1, \ldots, L\}$ be the index of the smallest notch $K_{p}$.
As explained in the text, the presence of a notch may remove the next tax change $K_{p+1}$ from the solution to the utility maximization problem with multiple kinks and notches (Lemma \ref{lemma:sol} in the supplemental appendix). 
Let $q \in \{p+1, \ldots, L\}$ be 
the index of the next tax change that appears in the solution. 
Following the proof of Lemma \ref{lemma:sol}, the distribution of $Y$ does not have any mass in the interval 
$( K_{p} ~;~ Y^I_{p} ]$ where $Y^I_{p} = N^I_{p} (1-t_{q-1} )^{\eps}$,
and $N^I_{p}$ is defined as part of the solution in Equation  \ref{eq:gen-sol} in the supplemental appendix.
The econometrician observes the value of $Y^I_p$, which is  between $K_{q-1}$ and $K_{q}$.
The goal is to solve for $\eps$ using this information.

The proof of Lemma \ref{lemma:sol} says $N^I_{p}$ satisfies the equation below.
\begin{gather*}
N^I_{p} (1-t_{q-1})^{1+\eps} + \eps \left( N^I_{p} \right)^{-1/ \eps} (K_{p})^{\frac{1+\eps}{\eps}} 
= (1+\eps) 
\left[ C_{p} - I_{q-1} + K_{q-1}(1 - t_{q-1}) \right]
\end{gather*}
Use the fact that $Y_p^I = N^I_{p} (1 - t_{q-1})^{\eps}$ and 
$\left( Y_p^I \right)^{-\frac{1}{ \eps}} (1 - t_{q-1}) = \left( N^I_{p} \right)^{-\frac{1}{ \eps}} $ and substitute these in the equation above to get
\begin{gather}
Y^I_{p} (1-t_{q-1}) + \eps \left( Y_p^I \right)^{-\frac{1}{ \eps}} (1 - t_{q-1})  (K_{p})^{\frac{1+\eps}{\eps}} 
= (1+\eps) 
\left[ C_{p} - I_{q-1} + K_{q-1}(1 - t_{q-1}) \right]
\nonumber
\\ 
Y^I_{p}  + \eps K_{p} \left( \frac{K_p}{Y_p^I} \right)^{\frac{1}{ \eps}}    
= \left( {1+\eps } \right) 
 \left(
 \frac{C_{p} - I_{q-1} + K_{q-1}(1 - t_{q-1})}
 {1 - t_{q-1}} 
 \right)
 \label{eq:notch-poss:1}
\end{gather}

The elasticity $\eps$ is identified if there exists an unique solution for $\eps$ in Equation \ref{eq:notch-poss:1}
as function of $Y_p^I$, $K_p$, $C_p$, $I_{q-1}$, $t_{q-1}$. We know a solution exists, and we show it must be unique. 
Consider the left-hand and right-hand sides of (\ref{eq:notch-poss:1}) as functions of $\eps$.
The solution occurs at the value of $\eps$ where both of these functions intersect. Uniqueness is equivalent to single-crossing of these functions.

The function on the right-hand side (RHS) of (\ref{eq:notch-poss:1}) has positive intercept equal to $[ C_{p} - I_{q-1} + K_{q-1}(1 - t_{q-1}) ]/(1 - t_{q-1})$. The function on the left-hand side (LHS) has intercept equal to 
$Y_p^I$ because
$\eps K_{p} \left( \frac{K_p}{Y_p^I} \right)^{\frac{1}{ \eps}}$ converges to zero as $\eps \downarrow 0$.
The intercept of the LHS is strictly bigger than the intercept of the RHS:
\begin{gather*}
Y_p^I  
\gtrless
 \frac{C_{p} - I_{q-1} + K_{q-1}(1 - t_{q-1})}
 {1 - t_{q-1}} 
\\
I_{q-1} + (Y_p^I  -K_{q-1}) (1 - t_{q-1}) 
\gtrless 
C_p
\\
C_p^I \gtrless C_p
\end{gather*}
where $C_p^I$ is the consumption value on the budget frontier when income is equal to $Y_p^I$ which is strictly greater than 
$C_p$. In fact, the consumer is indifferent between $(C_p^I,Y_p^I)$ and $(C_p,Y_p)$ where $Y_p^I>Y_p$. Since utility is strictly decreasing in $Y$ and increasing in $C$, we must have  $C_p^I > C_p$. 
Therefore, $Y_p^I  >  \left[ {C_{p} - I_{q-1} + K_{q-1}(1 - t_{q-1})} \right]/ (1 - t_{q-1}) $.

The function on the RHS of (\ref{eq:notch-poss:1}) has positive slope equal to 
$[ C_{p} - I_{q-1} + K_{q-1}(1 - t_{q-1}) ]/(1 - t_{q-1})$.
The function on the LHS has strictly positive derivative for any positive $\eps$,
\begin{gather*}
\frac{\partial}{\partial \eps} LHS = K_p \left( \frac{K_p}{Y_p^I} \right)^{\frac{1}{ \eps}}    
\left[ 1 - \frac{1}{\eps} \ln \left( \frac{K_p}{Y_p^I} \right) \right] 
\end{gather*}
which is strictly positive because $K_p>0$,  $ \left( \frac{K_p}{Y_p^I} \right)^{\frac{1}{ \eps}} \in (0,1) $,
and $- \frac{1}{\eps} \ln \left( \frac{K_p}{Y_p^I} \right) >0$.
The derivative is strictly increasing with $\eps$,
\begin{gather*}
\frac{\partial^2}{\partial \eps^2} LHS = K_p \left( \frac{K_p}{Y_p^I} \right)^{\frac{1}{ \eps}}    
 \frac{1}{\eps^3} \left[ \ln \left( \frac{K_p}{Y_p^I} \right) \right]^2 
\end{gather*}
which is also strictly positive. The limit of $\frac{\partial}{\partial \eps} LHS$ as $\eps \to \infty$ is equal to $K_p$.
Therefore, the slope of the LHS is positive, strictly increasing but always less than $K_p$. Next, we show that $K_p$
is strictly less than the constant slope of the RHS.
\begin{gather*}
K_p  
\gtrless
 \frac{C_{p} - I_{q-1} + K_{q-1}(1 - t_{q-1})}
 {1 - t_{q-1}} 
 \\
 I_{q-1} + (K_p - K_{q-1}) (1 - t_{q-1}) \gtrless C_p.
\end{gather*}
The value  $C_p^*=I_{q-1} + (K_p - K_{q-1}) (1 - t_{q-1})$ is what consumption would be if income were  equal to $K_p$ and the budget segment between $K_{q-1}$ and $K_q$ were extrapolated back to $K_p$.
We know that the indifference curve touches this budget segment at one point $(C_p^I, Y_p^I)$, and every other point on the extrapolated budget segment has strictly lower utility. We also know that $(C_p, K_p)$ is on such indifference curve,
so that 
$(C_p, K_p)$ is strictly preferred to $(C_p^*, K_p)$.
Therefore, $C_p^*< C_p$,
and $K_p < [ C_{p} - I_{q-1} + K_{q-1}(1 - t_{q-1}) ]/(1 - t_{q-1})$,
and the slope of the function on the LHS of (\ref{eq:notch-poss:1}) is always less than the slope of the function of the RHS.

In summary, the intercept of the function on the LHS of (\ref{eq:notch-poss:1}) is greater than the intercept of the function on the RHS.
Both functions are strictly increasing: the one on the RHS has constant slope, and the one on the LHS has increasing slope that is smaller than the slope of the RHS function. Therefore, the intersection of these two functions is unique. 

$\square$








\subsection{Partial Identification with Non-parametric Restrictions - Proof of Theorem \ref{theo_partial}}

\indent


First, let's fix $\eps>0$. We look at all possible PDFs in $\mathcal{F}_{n^*}$
and compute the maximum and minimum integrals over the interval $[\underline{n} ,\overline{n} ]$. 
The length of this interval is $\eps(s_0 - s_1)$.
Thus, without loss of generality, we restrict our attention to $f_{n^*}$ over the interval 
$[0,\eps(s_0 - s_1)]$ such that:

(i) $f_{n^*}$ is continuous, and it connects the point $(0, f_y(k ^-))$ to $(\eps(s_0 - s_1), f_y(k ^+))$ in the (x,y) plane;

(ii) the absolute value of the slope of $f_{n^*}$ is bounded by M.

First, start with $f_{n^*}$ being a line. The magnitude of the slope is
$\frac{|f_y(k ^+) - f_y(k ^-)|}{\eps(s_0-s_1)}$. Suppose this magnitude is bigger than $M$. Then, any $f_{n^*}$ satisfying (i) will have a slope magnitude higher than $M$ at some point. Therefore, we need to look at $\eps \geq \eps_1$ where 
$\eps_1 = \frac{|f_y(k ^+) - f_y(k ^-)|}{M (s_0-s_1) } $.

For fixed $\eps \geq \eps_1$, the slope of the line will be less or equal to $M$.
The maximum possible area is attained when the function has the shape of a hat with two line segments that attain the maximum slope. 
The first line segment starts at $(0, f_y(k ^-))$ and has slope $+M$;
the second line segment ends at $(\eps(s_0 - s_1), f_y(k ^+))$ and has slope $-M$. 
Call this function $\overline{f}_{n^*}$.
These lines intersect at $x^*$ where
\[
x^* = \frac{f_y(k ^+) - f_y(k ^-) + M \eps (s_0-s_1) }{2M}.
\]
Note that $x^*$ is always such $0 \leq x^* \leq \eps (s_0-s_1)$ because $\eps \geq \eps_1$.
Note that it is impossible to find another $f_{n^*}$ that satisfies (i), it is greater than $\overline{f}_{n^*}$, and that has slope magnitude less or equal than $M$.
The maximum area is
\begin{gather*}
\overline{A}(\eps) = \int_{0}^{\eps(s_0-s_1)} \overline{f}_{n^*} (v) ~dv
\\
=(1/4M) \left[ M^2 \eps^2 s_0^2 - 2 M^2 \eps^2 s_0 s_1 + M^2 \eps^2 s_1^2 + 2 M \eps f_y(k ^-) \right.
\\
\left.
s_0 - 2 M \eps f_y(k ^-) s_1 + 2 M \eps f_y(k ^+) s_0 - 2 M \eps f_y(k ^+) s_1 - f_y(k ^-)^2 + 2 f_y(k ^-) f_y(k ^+) - f_y(k ^+)^2)
\right]
\end{gather*}

The function $\overline{A}(\eps)$ is strictly increasing with respect to $\eps$
over $\eps \geq \eps_1$. In fact, the derivative is
$
((s_0 - s_1) (f_y(k ^-) + f_y(k ^+) + M\eps (s_0 - s_1))/2
$
which is strictly positive.

The minimum possible area is attained when the function has the shape of an inverted hat whose lines attain the maximum slope.
That is, a combination of two line segments.
One that starts $(0, f_y(k ^-))$ and has slope $-M$, and another that
ends at $(\eps(s_0 - s_1), f_y(k ^+))$ and has slope $+M$. 
Differently the hat function,
the intersection $(x^{**},y^{**})$ of this inverted hat function may or may not be above the x-axis. That is, $y^{**}$ may be negative, but $f_{n^*}$ is always positive.
In that case, we simply set the function to zero in the region where it would be negative. 
Call this function $\underline{f}_{n^*}$.

The intersection occurs at
\[
x^{**} = \frac{f_y(k ^-) - f_y(k ^+) + M \eps (s_0-s_1) }{2M}.
\]
Note that $x^{**}$ is always such $x^{**} \geq 0$ because $\eps \geq \eps_1$. The y-value of the intersection is
\[
y^{**} = \frac{f_y(k ^-) + f_y(k ^+) - M \eps (s_0-s_1) }{2M}.
\]
and this is positive as long as $\eps \leq \eps_2$ where $\eps_2 = \frac{|f_y(k ^+) + f_y(k ^-)|}{M (s_0-s_1) } $. Note also that $\eps_1 < \eps_2$. 

For $\eps_1 \leq \eps \leq \eps_2$, the minimum area is
\begin{gather*}
\underline{A}(\eps) = \int_{0}^{\eps(s_0-s_1)} \underline{f}_{n^*} (v) ~dv
\\
=(-1/4M) \left[ M^2 \eps^2 s_0^2 - 2 M^2 \eps^2 s_0 s_1 + M^2 \eps^2 s_1^2 - 2 M \eps f_y(k ^-) s_0 \right.
\\
\left.
+ 2 M \eps f_y(k ^-) s_1 - 2 M \eps f_y(k ^+) s_0 + 2 M \eps f_y(k ^+) s_1 - f_y(k ^-)^2 + 2 f_y(k ^-) f_y(k ^+) - f_y(k ^+)^2 \right]
\end{gather*}

The function $\underline{A}(\eps)$ is strictly increasing with respect to $\eps$
over $\eps_1 \leq  \eps < \eps_2$. In fact, the derivative is
$
((s_0 - s_1)*(f_y(k ^-) + f_y(k ^+) - M \eps ( s_0 - s_1) ))/2
$
which is strictly positive once we take into account $\eps < \eps_2$. The function
$\underline{A}(\eps)$ is constant with respect to $\eps$ over $\eps \geq \eps_2$.

Therefore, we have characterized the maximum and minimum areas $\underline{A}(\eps)$ and $\overline{A}(\eps)$ for any given $\eps$. These areas are undefined if $\eps<\eps_1$, they are equal if $\eps=\eps_1$, they are strictly increasing wrt $\eps$
and $\underline{A}(\eps) \leq \overline{A}(\eps)$
for $\eps \in (\eps_1,\eps_2)$. For $\eps \geq \eps_2$, $\overline{A}(\eps)$ continues to grow wrt $\eps$ but $\underline{A}(\eps)$ stays constant at $\underline{A}(\eps_2)$.
The expression for $\underline{A}(\eps_2)$ is
$(f_y(k ^-)^2 + f_y(k ^+)^2)/2M$.
Finally, we define the partially identified set.

\textbf{Case I:} 
If $B  < \underline{A}(\eps_1)=\overline{A}(\eps_1)$, there does not exist any function $f_{n^*}$ consistent with any elasticity $\eps$, so the set is empty.
The expression for $\underline{A}(\eps_1)=\overline{A}(\eps_1)$ is 
$( | f_y(k ^-) - f_y(k ^+) | (f_y(k ^-) + f_y(k ^+)) )/(2 M)$.

\textbf{Case II:} 
Suppose $B  \geq \underline{A}(\eps_1)$ and $B  < \underline{A}(\eps_2)$.
There is an interval range for $\eps$ such that for any $\eps$ in this interval there exists a function $f_{n^*}$ whose integral equals $B $. The minimum possible elasticity solves $\overline{A}(\underline{\eps}) = B $. That gives
\[
\underline{\eps} = 
\frac{2 \left[f_y(k ^+)^2/2 + f_y(k ^-)^2/2 + M ~ B  \right]^{1/2} 
- \left( f_y(k ^+) + f_y(k ^-) \right) }
{M(s_0 - s_1)}.
\]
The maximum possible elasticity solves $\underline{A}(\overline{\eps}) = B $. That gives
\[
\overline{\eps}
=\frac{-2 \left[f_y(k ^+)^2/2 + f_y(k ^-)^2/2 - M ~ B  \right]^{1/2} 
+ \left( f_y(k ^+) + f_y(k ^-) \right) }
{M(s_0 - s_1)}
\]

\textbf{Case III:} 
Suppose $B  \geq \underline{A}(\eps_2)$.
It is still possible to find a minimum elasticity that solves
$\overline{A}(\underline{\eps}) = B $.
However, for any elasticity $\eps \geq \underline{\eps}$ we have $\underline{A}(\eps)\leq B $, so $\overline{\eps}$ is infinity.

$\square$

\subsection{Tobit Regression - Proof of Identification Lemma \ref{lemma:tobit_robust}}
\label{sec:app:tobit}


\bigskip

\noindent \textbf{Part I:}

\bigskip

\indent

The probability limit of the Tobit MLE is $(e^*,b^*,s^*)$ and   $(b^*,s^*)$ correspond to a unconditional  distribution of $n^*$:
$G_{n^*}^*(n) = \mme\left[ \Phi\left( (n - X b^*)/s^* \right) \right]$,
where $G_{n^*}^* \in \mathcal{F}_{n^*}$.
The condition $F_y(y) = G_y^*(y)$ is equivalent to:
\begin{align*}
G_y^*(y) = G_{n^*}^*( y - e^* s_0 )
& = 
F_{n^*}( y - \eps s_0 ) = 
F_y (y)
~~\text{for } ~\forall y < k,
\\
G_y^*(y) =  G_{n^*}^* ( y -  e^* s_1 )
& = 
F_{n^*} ( y - \eps s_1  ) = 
F_y(y) 
~~\text{for } ~\forall y \geq k.
\end{align*}

Assumption \ref{aspt:tobit}(i) says that $\mathcal{F}_{n^*}$ satisfies Assumption \ref{aspt:cond_inv}.
Assumption \ref{aspt:cond_inv} says that the only values of $e \geq 0$ and $G_{n^*}\in \mathcal{F}_{n^*}$ that solve the equations
\begin{align*}
G_{n^*}( y - e s_0 )
& = 
F_{n^*}( y - \eps s_0 ) = 
F_y(y)
~~\text{for } ~\forall y < k,
\\
G_{n^*} ( y -  es_1 )
& = 
F_{n^*} ( y - \eps s_1  ) = 
F_y(y) 
~~\text{for } ~\forall y \geq k,
\end{align*}
are $e=\eps$ and $G_{n^*}=F_{n^*}$.
Therefore, $e^*=\eps$ and $G_{n^*}^*=F_{n^*}$.

\bigskip

\noindent \textbf{Part II:}

\bigskip

Start with a distribution $F_{n^*|X} \in \m{F}_{n^*|X}$.
By Assumption \ref{aspt:tobit}(ii), 
we have that 
$\mme\left[ F_{n^*|X}(n,X) \right] = \mme\left[  \Phi\left( (n - X \beta )/\sigma \right) \right] $
and
$(0,\beta,\sigma)$ is the solution to the minimization problem in \eqref{aspt:tobit2}.
We show below that this implies that 
$(e^*,b^*,s^*)=(\eps,\beta,\sigma)$, where $(e^*,b^*,s^*)$ is the probability limit of the Tobit MLE.

We show that the minimization problem in Equation \ref{aspt:tobit2} maps to the mid-censored Tobit maximization problem. 
\begin{align}
\min_{\delta, \theta, s} &   \frac{ 1-B }{2} 
\left\{ 
\log(s^2) + \frac{1}{s^2} \mme\left[ \left(n^* - D \delta - X \theta \right)^2 | n^* \not \in [\underline{n}, \overline{n}] \right]
\right\}
\notag
\\
& -
B \mme\left[
\log\left( B_N(X,\delta,\theta,s) \right) | n^* \in [\underline{n}, \overline{n}]
\right],
\label{eq:proof:rmk2:min1}
\\
\Longleftrightarrow & 
\notag
\\
\max_{\delta, \theta, s} &   - \frac{ 1-B }{2} 
\left\{ 
\log(s^2) + \frac{1}{s^2} \mme\left[ \left(n^* - D \delta - X \theta \right)^2 | n^* \not \in [\underline{n}, \overline{n}] \right]
\right\}
\notag
\\
& 
+
B \mme\left[
\log\left( B_N(X,\delta,\theta,s) \right) | n^* \in [\underline{n}, \overline{n}]
\right]
\notag
\\
\Longleftrightarrow & 
\notag
\\
\max_{\delta, \theta, s} &   ~( 1-B ) 
\left\{ 
\mme\left[\left. \log \left( \frac{1}{s} \phi\left( \frac{n^* - D \delta - X \theta}{s} \right) \right) \right| n^* \not \in [\underline{n}, \overline{n}] \right]
\right\}
\notag
\\
& 
+
B \mme\left[
\log\left( B_N(X,\delta,\theta,s) \right) | n^* \in [\underline{n}, \overline{n}]
\right]
\notag
\\
\Longleftrightarrow & 
\notag
\\
\max_{\delta, \theta, s} ~\mme & \left\{    
  \mmi\{ n^* < \underline{n} \} \log \left( \frac{1}{s} \phi\left( \frac{n^*  - X \theta}{s} \right) \right)  \right.
\notag
\\
& +
\mmi\{ n^* > \overline{n} \} \log \left( \frac{1}{s} \phi\left( \frac{n^* -  \delta - X \theta}{s} \right) \right)   
\notag
\\
& 
+
\left. \mmi\{  \underline{n} \leq n^* \leq  \overline{n}   \}
\log\left( \Phi\left( \frac{k-\eps s_1 - \delta - X \theta }{s} \right) - \Phi\left( \frac{ k- \eps s_0 - X \theta }{s} \right) \right) 
\right\}.
\notag
\end{align}

Re-parametrize the  maximization problem using the one-to-one mapping:
\[
\left(
\begin{array}{c}
e
\\
b
\\
s
\end{array}
\right)
=
\left(
\begin{array}{c}
e
\\
b_0
\\
b_1
\\
\vdots
\\
b_d
\\
s
\end{array}
\right)
\longrightarrow
\left(
\begin{array}{c}
(e-\eps)(s_1-s_0)
\\
b_0+s_0(e-\eps)
\\
b_1
\\
\vdots
\\
b_d
\\
s
\end{array}
\right)
=
\left(
\begin{array}{c}
\delta
\\
\theta_0
\\
\theta_1
\\
\vdots
\\
\theta_d
\\
s
\end{array}
\right)
=
\left(
\begin{array}{c}
\delta
\\
\theta
\\
s
\end{array}
\right),
\]
where we use that $b=(b_0,b_1,\ldots, b_d)$ and $\theta=(\theta_0,\theta_1,\ldots, \theta_d)$.
The parameters $b_0$ and $\theta_0$ correspond to the intercept variable in $X$.
The maximization problem becomes:
\begin{align*}
\max_{e, b, s} ~\mme & \left\{    
  \mmi\{ n^* < \underline{n} \} \log \left( \frac{1}{s} \phi\left( \frac{n^* - s_0(e-\eps)  - X b}{s} \right) \right)  \right.
\\
& +
\mmi\{ n^* > \overline{n} \} \log \left( \frac{1}{s} \phi\left( \frac{n^* -  s_1(e-\eps)  - X b }{s} \right) \right)   
\\
& 
+
\left. \mmi\{  \underline{n} \leq n^* \leq  \overline{n}   \}
\log\left( \Phi\left( \frac{k-\eps s_1 - s_1(e-\eps)  - X b }{s} \right) - \Phi\left( \frac{ k - \eps s_0 - s_0(e-\eps) - X b }{s} \right) \right) 
\right\}.
\end{align*}
Use that $y=n^* - \eps s_0$ if $n^*<\underline{n} \Leftrightarrow y<k$,  $y=n^* - \eps s_1$ if $n^*>\overline{n} \Leftrightarrow y>k$,
and
$\underline{n} \leq n^* \leq  \overline{n} \Leftrightarrow y=k$ to get
\begin{align}
\max_{e, b, s} ~\mme & \left\{    
  \mmi\{ y < k \} \log \left( \frac{1}{s} \phi\left( \frac{y -  e s_0    - X b}{s} \right) \right)  \right.
\notag
\\
& +
\mmi\{ y > k \} \log \left( \frac{1}{s} \phi\left( \frac{y -  e s_1    - X b }{s} \right) \right)   
\notag
\\
& 
+
\left. \mmi\{  y=k  \}
\log\left( \Phi\left( \frac{k- e s_1    - X b }{s} \right) - \Phi\left( \frac{ k -  e s_0   - X b }{s} \right) \right) 
\right\},
\label{eq:proof:rmk2:max2}
\end{align}
which is the mid-censored Tobit likelihood maximization problem.

By the equivalence of the optimization problems and the one-to-one re-parametrization mapping, we have that the solution of \eqref{eq:proof:rmk2:min1} maps to the solution of \eqref{eq:proof:rmk2:max2}.
That is, $(\delta^*,\theta^*,s^*)=(0,\beta,\sigma)$ maps to $(e^*,b^*,s^*)=(\eps,\beta,\sigma)$. 

Therefore, $G_{n^*}^*(n) = \mme\left[ \Phi\left( (n - X b^*)/s^* \right) \right] = \mme\left[ \Phi\left( (n - X \beta)/\sigma \right) \right] = F_{n^*}(n)$
and
$G_{y}^*(y) = F_y(y)$.

$\square$

\bigskip

\subsection{Censored Quantile Regression - Proof of Identification Lemma \ref{lemma:clad}}
Call $D = \mmi\left\{ Q_{\tau} \left(y \mid X \right) \neq k  \right\}$.
Let $\beta(\tau)=[\beta_0(\tau), \beta_1(\tau), \ldots, \beta_d(\tau)]' $.
Define $\tilde\beta(\tau) = [\beta_0(\tau) + \eps s_0, \beta_1(\tau), \ldots, \beta_d(\tau), \eps(s_1-s_0)]' $.
Multiplying Equation 
\ref{eq:quant_y_min_max} by $D$ yields
\begin{equation}
   D  Q_{\tau} \left(y \mid X \right)
   =
   D \tilde{X} \tilde{\beta}(\tau).
\nonumber
\end{equation}
Pre-multiplying it by $\tilde{X}'$ and taking expectations leads to
\begin{gather}
   D  \tilde{X}' Q_{\tau} \left(y \mid X \right)
   =
   D \tilde{X}' \tilde{X} \tilde{\beta}(\tau)
\nonumber
   \\
   \mme \left[
   D  \tilde{X}' Q_{\tau} \left(y \mid X \right)
   \right]
   =
   \mme \left[
   D \tilde{X}' \tilde{X} \tilde{\beta}(\tau)
   \right]
\nonumber
   \\
   \tilde{\beta}(\tau) = 
   \mme \left[
   D \tilde{X}' \tilde{X}
   \right]^{-1}
   \mme \left[
   D  \tilde{X}' Q_{\tau} \left(y \mid X \right)
   \right].
   \label{eq:id_beta_tau}
\end{gather}

An infinite amount of data identifies the joint distribution of $(y,X)$. This identifies the function
$Q_{\tau} \left(y \mid X =x \right)$ for every $x$ in the support of $X$, and the joint distribution of 
$(y,Q_{\tau} \left(y \mid X \right), X,\tilde{X}, D )$. 
Therefore, $\tilde{\beta}(\tau)$ is identified by Equation \ref{eq:id_beta_tau}. 
Finally, $\eps = \tilde{\beta}_{d+1}(\tau)/(s_1 - s_0)$. $ \square $

\subsection{Concave Kinks}
\label{sec:app:convex:kinks}

\indent 

The discussion in the main text focuses on ``convex kinks'', where the tax rate increases at the kink point (the tax schedule is convex).
They are simply referred to as kinks. 
A possible extension is to consider ``concave kinks'', that is, when the tax rate decreases at the kink point (the tax schedule is concave).
The budget line has a kink at $Y=K$,
\begin{equation}
 C= \mmi\{ Y \leq K \}[ I_0 + (1 - t_0) Y] 
 + \mmi\{ Y > K \} \left[ I_1 + \left( 1- t_{1} \right) (Y - K) \right].
\end{equation}
 where $t_0 > t_1$ and $I_1 = I_0 + (1 - t_0)K$ ensures continuity.

The shape of an agent's indifference curve depends on the agent's type $N^*$.
As $N^*$ increases, the point of tangency of the highest indifference curve shifts to the right along the budget line;
we start with just one tangency point at $Y<K$, go to two tangency points $Y<K$ and $Y'>K$ for one value of $N^*$, and finally to one tangency point $Y'>K$.
The threshold level $\underline{N}$ is that marginal value of $N^*$ between the first and second regimes. 
The convention we adopt is to resolve the agent's indifference in the second regime towards the smaller value of $Y$, that is, less labor supply. 

The solution for $Y$ has two cases in terms of $N^*,$ whereas the convex kink case in the main text had three cases (Equation \ref{eq:levelsol-onekink}).
Optimal income depends on the value of agent heterogeneity with respect to the threshold value $\underline{N},$
\begin{equation}
Y = 
\left\{
\begin{array}{ccl}
N^* (1-t_0)^{\eps} & \text{, if} & 0 < N^* \leq \underline{N}
\\

N^* (1-t_1)^{\eps} & \text{, if} & \underline{N} < N^*.
\end{array}
\right.
\end{equation}

We find $\underline{N}$ by solving an indifference condition.
Let $\underline{Y} = \underline{N}(1 - t_0)^\eps$ be biggest tangency point before the kink, 
and let $\overline{Y} = \underline{N}(1 - t_1)^\eps$ be the smallest tangency point after the kink.
The indifference condition requires the utility at $\underline{Y}$ to be the same as the utility at $\overline{Y}$, 
\begin{align}
 & I_0 + (1-t_0) \underline{Y} - \frac{ \underline{N} }{1+1/\eps} \left( \frac{ \underline{Y} }{\underline{N}} \right)^{1+1/\eps}
 \notag
 \\
 & = I_1 + (1-t_1) \overline{Y} - \frac{ \underline{N} }{1+1/\eps} \left( \frac{ \overline{Y} }{\underline{N}} \right)^{1+1/\eps},
 \label{eq:indiff_convex_kink}
\end{align}
where 
$I_0 + (1-t_0) \underline{Y} $ is consumption when income equals $\underline{Y}<K$, and
$I_1 + (1-t_1) \overline{Y}$ is consumption when income equals $\overline{Y}>K$.
We obtain an expression for $\underline{N}$ by plugging in 
$\underline{Y} = \underline{N}(1 - t_0)^\eps$
and
$\overline{Y} = \underline{N}(1 - t_1)^\eps$
in Equation \ref{eq:indiff_convex_kink}: 
\begin{equation}\label{eq:nthreshold}
 \underline{N} = \dfrac{(1 + \eps) K \left[ (1 - t_0)
- (1 - t_1)\right]}{\left[
 (1 - t_0)^{\eps + 1}
- (1 - t_1)^{\eps + 1}
\right]}.
\end{equation}

 A continuous distribution of $N^*$ implies a distribution of $Y$ that has a gap $[\underline{Y},\overline{Y}]$ of zero mass around the kink point but is continuous otherwise.
 The limits of this gap are: $\underline{Y} = \underline{N}(1 - t_0)^\eps$ and $\overline{Y} = \underline{N}(1 - t_1)^\eps,$
 where $\overline{Y} > \underline{Y}$ because $t_0 > t_1$.
 The researcher observes the distribution of $Y$ and thus knows those two limits.
 Taking the difference of logs of $\overline{Y}$ and $\underline{Y}$, we obtain: 
 $\overline{y} - \underline{y} = \eps (s_1 - s_0)$, 
 where lower-case variables denote the log of upper-case variables, and $s_j = \log (1-t_j)$, $j=1,2$.
 Therefore, the elasticity is identified,
 \begin{align*}
 \eps = \frac{\overline{y} - \underline{y} }{s_1 - s_0}.
 \end{align*}
Without having to impose much structure on the distribution of $N^*$,
it is possible to non-parametrically identify the elasticity, unlike the case of a convex kink (Theorem 1 by \cite{blomquist2017a}).



\clearpage
\onehalfspacing

\numberwithin{equation}{section}
\numberwithin{lemma}{section}
\numberwithin{theorem}{section}

\newpage
\setcounter{page}{1}

\bigskip

\begin{center}
 \Large \uppercase{``Better Bunching, Nicer Notching''}

\ifid
\normalsize Marinho Bertanha, Andrew McCallum, Nathan Seegert
\fi

\end{center}

\bigskip

\section{Supplemental Appendix for Online Publication}
\label{app:supp}

\subsection{General Utility Maximization Problem with Multiple Kinks and Notches}
\label{sec:app:general_prob}

\indent

This section generalizes  Problem \ref{eq:util_saez} in the main text to  multiple kinks and notches, and next section presents the general solution.
The exercise brings new insights to the identification of the elasticity,
when compared to the particular solution in the case of one kink or notch presented in the main text. 
First, in a budget set with multiple kinks but no notches, the general solution is simply a combination of solutions local to each kink. 
The bunching intervals of consecutive kinks do not overlap 
(Equation \ref{eq:gen-sol}).
As a result, inference methods for the elasticity  that are valid in the case of one kink may still be used locally to each kink.

Second, a notch at $k$ creates an empty interval in the support of the distribution of $y$ right after $k$. 
Such an empty interval may or may not contain the next tax change point $k'>k$, depending on the value of $\eps$.
For example, eligibility for Medicaid benefits in the United States creates a sizeable notch that may overshadow the next tax change in the budget set of some individuals.
In this case, inference methods that focus on kinks without accounting for other notches may produce misleading conclusions about the elasticity.
Of course, both of these insights change if the model allows the elasticity to vary with ability.

To generalize the objective function in Problem \ref{eq:util_saez}, we update the budget set to have $J$ different tax regimes that change at cutoff points $0< K_1< \ldots < K_J$ on pre-tax labor income $Y$.
Each tax regime has income tax $t_j$ such that $0 \leq t_0 \leq t_1 \leq \ldots \leq t_J < 1$. There are two possible tax changes. 
A change in tax rate is a kink. A lump-sum tax change is called a notch.
Agent type $N^\ast$ maximizes utility $U(C,Y;N^\ast)$ as follows
\begin{eqnarray}
\max_{C,Y} & & C - \frac{N^*}{1 + 1/\varepsilon} \left(\frac{Y}{N^*} \right)^{1 + \frac{1}{\varepsilon}}
\label{eq:util_saez_gen}
\\
 & s.t. & C= \sum_{j=0}^J \mmi\{ K_j < Y \leq K_{j+1} \} \left[ I_j + \left( 1- t_{j} \right) (Y - K_j) \right],
\label{eq:bf_gen}
\end{eqnarray}
where 
$K_0=0$, $K_{J+1}=\infty$, $\mmi\{ \cdot \}$ is the indicator function, the solution is always on the budget frontier (Equation \ref{eq:bf_gen}),
and we assume the agent resolves indifference by choosing the smallest value of $Y$. 
The elasticity of income $Y$ with respect to $(1-t_j)$ is equal to $\eps$ when the solution is interior. 

The budget frontier is continuous except when there is a notch. 
The limit of the budget frontier when $Y \downarrow K_j$ is equal to $I_j$,
but equal to $I_{j-1} + \left( 1- t_{j-1} \right) (K_j - K_{j-1})$ when $Y \uparrow K_j$.
The size of the jump discontinuity at a notch location $K_j$ is equal to $I_j - I_{j-1} - \left( 1- t_{j-1} \right) (K_j - K_{j-1})$.
The intercepts $I_j$ and $I_{j-1}$ are assumed to be such that jump discontinuities at notches are negative.

\subsection{General Solution with Multiple Kinks and Notches}
\label{sec:app:general_prob_sol}

\indent 

Lemma \ref{lemma:sol} below provides a general solution to Problem \ref{eq:util_saez_gen} with any combination of kinks and notches. 

\begin{lemma}
\label{lemma:sol}
Define $\m{N} = \cup_{j=0}^J \left( K_j (1-t_j)^{-\eps}; K_{j+1} (1-t_j)^{-\eps} \right]$
as the set of $N^*$ values for which the indifference curves are tangent to the budget frontier.
The function 
$Y^*:\m{N} \to \mmr$,
$Y^*(N^{\ast}) = \sum_{j=0}^J \mmi\{ K_j (1-t_j)^{-\eps}< N^{\ast} \leq K_{j+1} (1-t_j)^{-\eps}\} N^{\ast} (1-t_j)^{\eps}$,
maps $N^*$ values to the $Y$ values corresponding to such tangency points.
Similarly, $C^*(N^{\ast})$ is consumption on the budget frontier (Equation $\ref{eq:bf_gen}$)
when $Y=Y^*(N^{\ast})$.
Let $C_j$ be the value of $C^*(N^{\ast})$ whenever $Y^*(N^{\ast})=K_j$, $j=1,\ldots, J$.
For a notch-point $K_j$, define the value of $N_j^I$ to be that of the first indifference curve tangent to the budget frontier on the right of $Y=K_j$, such that the utility level is equal to the utility of the notch-point $K_j$,
\begin{gather}
N_j^I = \min \bigg\{ N^{\ast} \in \m{N} ~:~ U(C_j,K_j) = U(C^*(N^{\ast}),Y^*(N^{\ast})) \bigg\}.
\label{eq:indiff}
\end{gather}

In the case of a kink, the bunching interval is defined as $[\underline{N}_{j}, \overline{N}_{j}]$,
where $\underline{N}_{j} = K_{j}(1-t_{j-1})^{-\eps}$,
and $\overline{N}_{j} = K_{j}(1-t_{j})^{-\eps}$.
In the case of a notch, the expression for $\underline{N}_{j}$ equals that of the kink case, 
but $\overline{N}_{j}$ changes to $N_j^I$.

Note that the bunching intervals of two consecutive kinks do not overlap, that is,
$K_{j}(1-t_{j})^{-\eps} < K_{j+1}(1-t_{j})^{-\eps}$.
The same is not true for a kink or a notch $K_{j+1}$ that comes right after a notch $K_{j}$,
because $N_j^I$ may be greater than $ K_{j+1}(1-t_{j})^{-\eps}$ depending on $\eps$.
In this case, $Y=K_{j+1}$ does not appear in the solution. To account for that, 
construct a subsequence $\{ j_l \}_{l=1}^L$ of $\{1,\ldots, J\}$ such that:
(i) $j_1=1$; and (ii) for $l \geq 2$, set $j_l$ to be the smallest $j$ such that 
$ \underline{N}_{j}>\overline{N}_{j_{l-1}}$.
Then, the solution to the maximization problem in (\ref{eq:util_saez_gen}) is given by
\begin{equation}
Y = 
\left\{
\begin{array}{ccl}
N^{\ast} (1-t_{j_1-1})^{\eps} & \text{, if} & 0 < N^{\ast} < \underline{N}_{j_1}
\\
K_{j_1} & \text{, if} & \underline{N}_{j_1} \leq N^{\ast} \leq \overline{N}_{j_1}
\\
N^{\ast} (1-t_{j_2 - 1})^{\eps} & \text{, if} & \overline{N}_{j_1} < N^{\ast} < \underline{N}_{j_2}
\\
\vdots & & 
\\
N^{\ast} (1-t_{j_L-1})^{\eps} & \text{, if} & \overline{N}_{j_{L-1}} < N^{\ast} < \underline{N}_{j_L}
\\
K_{j_L} & \text{, if} & \underline{N}_{j_L} \leq N^{\ast} \leq \overline{N}_{j_L}
\\
N^{\ast} (1-t_{J})^{\eps} & \text{, if} & \overline{N}_{j_{L}} < N^{\ast} < \infty.
\end{array}
\right.
\label{eq:gen-sol}
\end{equation}
\end{lemma}

\begin{proof}
For every $N^* > 0$, there exists an unique solution on the budget frontier.
If the consumer is indifferent between two solutions, we assume the consumer takes the solution with less $Y$.
The proof is by induction over $\bar{J}=0, 1, \ldots, J$.
 Denote the budget frontier $BF^{\bar{J}}$ by
\begin{equation*}
C= \sum_{j=0}^{\bar{J}} \mmi\{ \bar{K}_j < Y \leq \bar{K}_{j+1} \} \left[ I_j + \left( 1- t_{j} \right) (Y - \bar{K}_j) \right].
\end{equation*}
where $\bar{K}_j=K_j$ for $j=0,1,\ldots, \bar{J}$ and $\bar{K}_{\bar{J}+1}=\infty$.

As we change the budget frontier from $BF^{\bar{J}}$ to $BF^{\bar{J}+1}$, 
$K_{\bar{J}+1}$ takes a finite value strictly greater than $K_{\bar{J}}$, and $K_{\bar{J}+2}$ is set to $\infty$.
If the solution to Problem \ref{eq:util_saez_gen} with budget frontier $BF^{\bar{J}}$ is such that $Y < K_{\bar{J}+1} < \infty$, then this is also the solution to Problem \ref{eq:util_saez_gen} with budget frontier $BF^{\bar{J}+1}$.
In fact, points on $BF^{\bar{J}}$ dominate points on $BF^{\bar{J}+1}$, and they coincide for $Y<K_{\bar{J}+1}$.

\bigskip

\textbf{Part I: } \textit{$\bar{J}=0$, solve Problem \ref{eq:util_saez_gen} with budget $BF^0$}.

This is a standard consumer maximization problem where the optimal choice for $Y$ occurs at the point the indifference curve is tangent to $BF^{0}$. 
Therefore, for $N^* > 0$, $Y=N^*(1-t_0)^{\eps}$.

\bigskip

\textbf{Part II: } \textit{$\bar{J}=1$, solve Problem \ref{eq:util_saez_gen} with budget $BF^1$}.

The budget frontier $BF^1$ has two segments $BF^1_0$ for $0<Y\leq K_1$, and $BF^1_1$ for $K_1 < Y $.
If $N^* < K_1(1-t_0)^{-\eps}$, then the solution of Part I, $Y=N^*(1-t_0)^{\eps} < K_1$, is also the solution in Part II.
It remains to find the solution for $N^* \geq K_1(1-t_0)^{-\eps}$. These solutions must lie on $BF^1$ for $Y \geq K_1$
because they strictly dominate those that lie to the left of $K_1$.

\textit{Case I : Suppose $K_1$ is a kink.} 

Assume $N^*$ is such that $K_1(1-t_0)^{-\eps} \leq N^* \leq K_1(1-t_1)^{-\eps}$. 
If the solution is interior to $BF^1_1$, then it must be at a tangent point in which case $Y=N^*(1-t_1)^{\eps}$. 
However, $Y=N^*(1-t_1)^{\eps} \leq K_1$, a contradiction because this $Y$ falls outside of the interior of $BF^1_1$. 
Therefore, if $N^*$ is such that $\underline{N}_1=K_1(1-t_0)^{-\eps} \leq N^* \leq K_1(1-t_1)^{-\eps}=\overline{N}_1$, then the solution is $Y=K_1$.
Suppose $N^* > \overline{N}_1$. Then, the solution is in the interior of $BF^1_1$, and it is equal to 
$Y=N^*(1-t_1)^{\eps}$.

\textit{Case II : Suppose $K_1$ is a notch.} 

There is a jump-down discontinuity in $BF^1$ at $K_1$, and $BF^1$ is continuous from the left. 
Consider the point $(C,Y)=(C_1,K_1)$ on $BF^1_0$.
Define $Y^D$ to be the value of $Y$ such that the corresponding $C$ value on $BF^1_1$ is equal to $C_1$.
The jump-down discontinuity creates a strictly dominated region on $BF^1_1$ because the utility of $(C_1,K_1)$
is strictly greater than the utility of any solution with $Y \in (K_1,Y^D)$. Indifference between $K_1$ and $Y^D$ is resolved towards $K_1$ by assumption. 
Therefore, we cannot have solutions to Problem \ref{eq:util_saez_gen} with budget $BF^1$ such that $Y \in (K_1,Y^D]$.

Define the point $\ti{N}_1^I$ as being the solution of Problem \ref{eq:util_saez_gen} with budget $BF^1$ (instead of $BF$).
This is the smallest $N^*$ for which Problem \ref{eq:util_saez_gen} with budget $BF^1_1$
has solution with utility equal to $U(C_1,K_1)$.

First, a solution $\ti{N}_1^I$ exists.
To see that, note that for small $N^*$, the tangent point $Y=N^*(1-t_1)^{\eps}$ along $BF^1_1$ falls in the dominated region $Y \in (K_1,Y^D]$, and the utility is less than $U(C_1,K_1)$; on the other hand, the utility at this tangent point increases with $N^*$, and it eventually equals $U(C_1,K_1)$. The solution is such that $\ti{N}_1^I \geq Y^D (1-t_1)^{-\eps} > K_1 (1-t_1)^{-\eps}$.

Second, the solution $\ti{N}_1^I$ is unique.
 To see that, solve for $N^*$ in the equation below.
\begin{gather*}
U(C_1,K_1) = U\left(I_1 + N^* (1-t_1)^{\eps+1} - K_1(1-t_1) ~,~ N^*(1-t_1)^{\eps} \right)
\end{gather*}
 where $C=I_1 + N^*(1-t_1)^{\eps+1} - K_1(1-t_1)$ is consumption on $BF^1_1$ when $Y=N^*(1-t_1)^{\eps}$.
Evaluating and rearranging the equality gives
\begin{gather*}
N^*(1-t_1)^{1+\eps} + \eps (N^*)^{-1/ \eps} (K_1)^{\frac{1+\eps}{\eps}} = (1+\eps) \left[ C_1 - I_1 + K_1(1-t_1) \right]
\end{gather*}
The solution is unique because the derivative of the right-hand side is strictly positive given $N^*>K_1(1-t_1)^{-\eps}$.
Note that $\ti{N}_1^I$ is the unique solution to Problem \ref{eq:util_saez_gen} when the budget is $BF^1$.

Call $\ti{Y}^I_1 = \ti{N}_1^I (1-t_1)^{\eps}$.
Suppose there is a solution to Problem \ref{eq:util_saez_gen} with budget $BF^1$ such that 
$Y^D < Y \leq \ti{Y}^I_1$. This solution is interior to budget $BF^1_1$, so we must have 
$Y=N^* (1-t_1)^{\eps}$ for some $N^*$. 
But such a solution cannot be a solution to Problem \ref{eq:util_saez_gen} with budget $BF^1$
because $Y \leq \ti{Y}^I_1$ and so dominated by $(C_1,K_1)$.
Therefore, we cannot have solutions to Problem \ref{eq:util_saez_gen} with budget $BF^1$ such that $Y \in (K_1, \ti{Y}^I]$.

It remains to characterize the solution when $N^*$ is such that $K_1(1-t_0)^{-\eps} \leq N^*$.
If $N^*$ is such that $\underline{N}_1 =K_1(1-t_0)^{-\eps} \leq N^* \leq \ti{Y}^I (1-t_1)^{-\eps} = \overline{N}_1$, the solution cannot be in the interior of $BF^1_0$
since $Y=N^*(1-t_0)^{\eps} \geq K_1$; it cannot be in $ (K_1, \ti{Y}^I]$ either. Assume it is in the interior of 
$BF^1_1$ with $Y>\ti{Y}_I$. Since it is interior, it satisfies $Y=N^*(1-t_1)^{\eps}$, but $N^* \leq \ti{Y}^I (1-t_1)^{-\eps} $
which makes $Y \leq \ti{Y}^{I}$, a contradiction. Therefore, the solution to Problem \ref{eq:util_saez_gen} with budget $BF^1$
when $N^* \in [\underline{N}_1 ; \overline{N}_1]$
is $Y=K_1$. Finally, suppose $N^*>\overline{N}_1$. Then, the solution is in the interior of $BF^1_1$, and it is equal to 
$Y=N^*(1-t_1)^{\eps}$.


\bigskip

\textbf{Part III: }
\textit{
Assume the solution of Problem \ref{eq:util_saez_gen} with budget $BF^{\bar{J}}$ and $1 \leq {\bar{J}} < J$
is as in Equation \ref{eq:gen-sol} with ${\bar{J}}$.
Show that \eqref{eq:gen-sol} with ${\bar{J}} +1$ solves 
Problem \ref{eq:util_saez_gen} with budget $BF^{\bar{J}+1}$}.

Consider Problem \ref{eq:util_saez_gen} with budget $BF^{\bar{J}}$ and solution \ref{eq:gen-sol}
with $L$ being $\bar{L}$.
If $N^*$ is such that $Y<K_{\bar{J}+1}<\infty$, then $Y$ also solves Problem \ref{eq:util_saez_gen} with budget $BF^{\bar{J}+1}$.
Therefore, the solution to Problem \ref{eq:util_saez_gen} with budget $BF^{\bar{J}+1}$ or budget $BF^{\bar{J}}$ coincide for those values of $N^*$.
Note also that, if $K_j$ is a notch and $j<j_{\bar{L}}$, then the value of $\overline{N}_j$ (defined in (\ref{eq:indiff})) does not change when the budget changes from $BF^{\bar{J}}$ to $BF^{\bar{J}+1}$. 
If $K_{j_{\bar{L}}}$ is a notch, then the value $\overline{N}_{j_{\bar{L}}}$ may change (case IV below).
In what follows, consider the last two budget segments of $BF^{\bar{J}+1}$: $BF^{\bar{J}+1}_{\bar{J}}$ and $BF^{\bar{J}+1}_{\bar{J}+1}$.

\bigskip

\textit{Case I : $K_{j_{\bar{L}}}$ is a kink, $K_{\bar{J}+1}$ is a kink}

In this case, ${j_{\bar{L}+1}} = {\bar{J}+1}$ because 
$\underline{N}_{\bar{J}+1} = K_{\bar{J}+1} (1- t_{\bar{J}} )^{-\eps} >\overline{N}_{j_{\bar{L}}}$,
so that ${\bar{J}+1}$ is the smallest $j$ such that $ \underline{N}_{j}>\overline{N}_{j_{\bar{L}}}$.
It is also true that ${j_{\bar{L}}} = {\bar{J}}$. 
To see that, note that consecutive intervals $[\underline{N}_{j},\overline{N}_{j}]$ never overlap for kinks
because $\overline{N}_{j} = K_j(1-t_j)^{-\eps} < K_{j+1}(1-t_j)^{-\eps} = \underline{N}_{j+1}$. 
The upper limit of a kink interval $j$ is strictly smaller than the lower limit of a notch interval $j+1$. 
However, the upper limit of a notch interval $j$ may be bigger than the lower limit of the next interval $j+1$.
Suppose ${j_{\bar{L}}} = {\bar{J}}$ were not true, that is, ${j_{\bar{L}}} < {\bar{J}}$.
 Then, any $j$ such that ${j_{\bar{L}}} < j \leq {\bar{J}}$ is not in the subsequence $\{ j_l \}$
because $K_{j_{\bar{L}}}$ is a notch, and its interval overlaps with the $j$ interval. But this is a contradiction with $K_{j_{\bar{L}}}$ being a kink point.

If $N^* < K_{\bar{J}+1} (1-t_{\bar{J}} )^{-\eps}$, then the solution \ref{eq:gen-sol} with budget $BF^{\bar{J}}$
is $Y<K_{\bar{J}+1}$, and $Y$ also solves
Problem \ref{eq:util_saez_gen} with budget $BF^{\bar{J}+1}$ for that same value of $N^*$. 
It remains to characterize the solution when $N^* \geq K_{\bar{J}+1} (1-t_{\bar{J}} )^{-\eps}$

Assume $N^*$ is such that 
$\underline{N}_{\bar{J}+1} = K_{\bar{J}+1} (1-t_{\bar{J}} )^{-\eps} 
\leq N^* \leq 
K_{\bar{J}+1} (1-t_{\bar{J}+1} )^{-\eps} = \overline{N}_{\bar{J}+1}$. 
As seen in Part II, Case I, the solution cannot be interior to $BF^{\bar{J}+1}_{\bar{J}+1}$.
The solution must be at $K_{\bar{J}+1}$.
Assume $N^* >\overline{N}_{\bar{J}+1}$. Then, the solution is interior to $BF^{\bar{J}+1}_{\bar{J}+1}$,
and it equals to $Y = N^* (1-t_{\bar{J}+1} )^{\eps}$.

\bigskip

\textit{Case II : $K_{j_{\bar{L}}}$ is a kink, $K_{\bar{J}+1}$ is a notch}

As seen in Part III, Case I, ${j_{\bar{L}}} = {\bar{J}}$. We also have ${j_{\bar{L}+1}} = {\bar{J}+1}$
because the $j$ interval $[\underline{N}_{j},\overline{N}_{j}]$ of a kink does not overlap with
the $j+1$ interval of a notch.

If $N^* < K_{\bar{J}+1} (1-t_{\bar{J}} )^{-\eps}$, then the solution \ref{eq:gen-sol} with budget $BF^{\bar{J}}$
is $Y<K_{\bar{J}+1}$, and $Y$ also solves
Problem \ref{eq:util_saez_gen} with budget $BF^{\bar{J}+1}$ for that same value of $N^*$. 
It remains to characterize the solution when $N^* \geq K_{\bar{J}+1} (1-t_{\bar{J}} )^{-\eps}$

Assume $N^*$ is such that 
$\underline{N}_{\bar{J}+1} = K_{\bar{J}+1} (1-t_{\bar{J}} )^{-\eps} 
\leq N^* \leq 
\overline{N}_{\bar{J}+1}$,
where $\overline{N}_{\bar{J}+1}$ is the solution of Problem \ref{eq:indiff} when the budget is $BF^{\bar{J}+1}$.
As seen in Part II, Case II, the solution $Y$ cannot be in $(K_{\bar{J}+1}, \overline{N}_{\bar{J}+1} ( 1 - t_{\bar{J}+1} )^{\eps}]$ or in the interior of $BF^{\bar{J}+1}_{\bar{J}+1}$. Therefore, the solution is $Y= K_{\bar{J}+1}$.
Assume $N^* >\overline{N}_{\bar{J}+1}$. Then, the solution is interior to $BF^{\bar{J}+1}_{\bar{J}+1}$,
and it equals to $Y = N^* (1-t_{\bar{J}+1} )^{\eps}$.

\bigskip

\textit{Case III : $K_{j_{\bar{L}}}$ is a notch, $\overline{N}_{j_{\bar{L}}} < \underline{N}_{\bar{J}+1}$}

For the notch $K_{j_{\bar{L}}}$, the solution $\overline{N}_{j_{\bar{L}}}$ to Problem \ref{eq:indiff} when the budget is $BF^{\bar{J}}$ does not change when the budget becomes $BF^{\bar{J}+1}$ precisely because $\overline{N}_{j_{\bar{L}}} < \underline{N}_{\bar{J}+1}$.
In this case, $j_{\bar{L} + 1} = \bar{J}+1$.
For $N^*$ such that $\overline{N}_{j_{\bar{L}}} < N^* < K_{\bar{J}+1} ( 1- t_{\bar{J}} )^{-\eps}$,
the solution \ref{eq:gen-sol} with budget $BF^{\bar{J}}$
is $Y<K_{\bar{J}+1}$, and $Y$ also solves
Problem \ref{eq:util_saez_gen} with budget $BF^{\bar{J}+1}$ for that same value of $N^*$.
It remains to characterize the solution when $N^* \geq K_{\bar{J}+1} (1- t_{\bar{J}} )^{-\eps}$.

Assume $K_{\bar{J} + 1}$ is a kink, and that $N^*$ is such that 
$\underline{N}_{\bar{J}+1} = K_{\bar{J}+1} (1-t_{\bar{J}} )^{-\eps} 
\leq N^* \leq 
K_{\bar{J}+1} (1-t_{\bar{J}+1} )^{-\eps} = \overline{N}_{\bar{J}+1}$. 
As seen in Part II, Case I, the solution cannot be interior to $BF^{\bar{J}+1}_{\bar{J}+1}$.
The solution must be at $K_{\bar{J}+1}$.
Assume $N^* >\overline{N}_{\bar{J}+1}$. Then, the solution is interior to $BF^{\bar{J}+1}_{\bar{J}+1}$,
and it equals to $Y = N^* (1-t_{\bar{J}+1} )^{\eps}$.

Assume $K_{\bar{J} + 1}$ is a notch, and that $N^*$ is such that 
$\underline{N}_{\bar{J}+1} = K_{\bar{J}+1} (1-t_{\bar{J}} )^{-\eps} 
\leq N^* \leq 
 \overline{N}_{\bar{J}+1}$,
where $\overline{N}_{\bar{J}+1}$ is the solution of Problem \ref{eq:indiff} when the budget is $BF^{\bar{J}+1}$.
As seen in Part II, Case II, the solution $Y$ cannot be in 
$(K_{\bar{J}+1}, \overline{N}_{\bar{J}+1} ( 1 - t_{\bar{J}+1} )^{\eps}]$ 
or in the interior of $BF^{\bar{J}+1}_{\bar{J}+1}$. Therefore, the solution is $Y= K_{\bar{J}+1}$.
Assume $N^* >\overline{N}_{\bar{J}+1}$. Then, the solution is interior to $BF^{\bar{J}+1}_{\bar{J}+1}$,
and it equals to $Y = N^* (1-t_{\bar{J}+1} )^{\eps}$.

\bigskip

\textit{Case IV : $K_{j_{\bar{L}}}$ is a notch, $\overline{N}_{j_{\bar{L}}} \geq \underline{N}_{\bar{J}+1}$}

The indifference value for $Y$ at $\overline{N}_{j_{\bar{L}}}$ is $Y^I_{j_{\bar{L}}} = \overline{N}_{j_{\bar{L}}} 
(1-t_{\bar{J}})^{\eps} \geq \underline{N}_{\bar{J}+1} (1-t_{\bar{J}})^{\eps} = K_{\bar{J}+1}$.
If $\overline{N}_{j_{\bar{L}}} = \underline{N}_{\bar{J}+1}$, the solution to Problem \ref{eq:indiff} when the budget is $BF^{\bar{J}}$ remains unchanged when the budget becomes $BF^{\bar{J}+1}$.
If $\overline{N}_{j_{\bar{L}}} > \underline{N}_{\bar{J}+1}$, then $Y^I_{j_{\bar{L}}} > K_{\bar{J}+1}$,
and the solution to Problem \ref{eq:indiff} when the budget is $BF^{\bar{J}}$
changes when the budget becomes $BF^{\bar{J}+1}$.
The value of $\overline{N}_{j_{\bar{L}}}$ increases 
 such that
the new indifference point satisfies $Y^I_{j_{\bar{L}}} = \overline{N}_{j_{\bar{L}}} (1-t_{\bar{J}+1})^{\eps}$.

There does not exist a $j$ such that $ \underline{N}_{j}>\overline{N}_{j_{\bar{L}}}$ because $K_{\bar{J}+1}$
is the last tax-change point available and $\underline{N}_{\bar{J}+1} \leq \overline{N}_{j_{\bar{L}}}$.
Therefore, when constructing the solution of Problem \ref{eq:util_saez_gen} with budget $BF^{\bar{J}+1}$, 
the last term in the subsequence $\{j_l\}$ remains $j_{\bar{L}}$. 

The point $K_{j_{\bar{L}}}$ is a notch, so Part II, Case II says that for $N^*$ such that 
$\underline{N}_{j_{\bar{L}}} = K_{j_{\bar{L}}} (1-t_{j_{\bar{L}}-1} )^{-\eps} 
\leq N^* \leq 
 \overline{N}_{j_{\bar{L}}}$,
 the solution $Y$ cannot be in 
$(K_{j_{\bar{L}}}, \overline{N}_{j_{\bar{L}}} ( 1 - t_{\bar{J} + 1 } )^{\eps}]$ 
or in the interior of $BF^{\bar{J}+1}_{\bar{J}+1}$. Therefore, the solution is $Y= K_{j_{\bar{L}}}$.
Assume $N^* >\overline{N}_{j_{\bar{L}}}$. Then, the solution is interior to $BF^{\bar{J}+1}_{\bar{J}+1}$,
and it equals to $Y = N^* (1-t_{\bar{J}+1} )^{\eps}$.
\end{proof}

\subsection{Friction Errors and Failure of the ``Polynomial Strategy''} 
\label{sec:app:friction_error}
\indent

This section presents a counterexample that illustrates the failure of a common identification strategy used in applied work to estimate the elasticity using kinks.
In a recent survey article, \cite{Kleven2016} summarizes the strategy which is employed in the literature to estimate the distribution of ${y_0}$. 
The ``polynomial strategy'' was first proposed by 
\cite{chetty2011} (Equations 14-15  and Figures 3-4)
and consists of fitting a flexible polynomial to an estimate of the PDF of $y$.
The polynomial regression excludes observations that lie in a range around the kink point. The researcher chooses the range based on the support of the distribution of friction errors.
The polynomial fit is then extrapolated to this excluded region as a way of predicting $f_{y_0}$.
The procedure is widely used in the bunching literature; see, for example,  
Figure 6 by \cite{bastani2014} ,
Figure 1 by \cite{devereux2014},
and Figure 4 by \cite{BestKleven2018}. 

First, we set the parameters of the model.
The true values are: $\varepsilon=1.5$ (elasticity);
$t_0=.2$ and $t_1=0.3$ (before and after tax rates);
kink-point $k=0$.
The bunching interval is $[\underline{n},\overline{n}]=[0.335,.535]$.
The distribution of the ability variable is assumed uniform, 
$n^* \sim U[-.565;1.435]$;
that is, the support is centered at $0.435$ and has length equal to 2.
The probability of bunching, or bunching mass $B$, is equal to 10\% in this example.
The friction error $e$ is also assumed uniformly distributed
$e \sim U[-0.5;0.5]$. 
The value of labor income observed by the researcher is $\ti{y} = y +e$, 
where $y$ is a function of $n^*$, $\eps$, $t_0$, and $t_1$,
as described in Equation \ref{eq:logsol-onekink}.

In the counterfactual scenario of no tax change, 
we have $\underline{n} = \overline{n}$, and the counterfactual income with friction error is denoted $\ti{y}_0$.
The counterfactual income without friction error is $y_0$.
Figure \ref{fig:friction_error:a} depicts the PDF of $\ti{y}$ and $\ti{y}_0$.

A common identification strategy used in applied work is to fit a polynomial to the PDF of $\ti{y}$ excluding observations in the neighborhood of the kink $k=0$, that corresponds to the support of the measurement error (i.e. $[-0.5;0.5]$).
The estimated bunching mass is the area between the PDF of $\ti{y}$ and the polynomial fit extrapolated to the excluded neighborhood around the kink.
Figure \ref{fig:friction_error:b} illustrates the procedure.
The figure shows that such strategy fails to identify the true bunching mass, even when the polynomial fit of 7th order is perfect, and we assume the researcher knows the support of $e$.

The last part of the estimation strategy uses the extrapolated polynomial to predict the counterfactual PDF of $y_0$.
Following Equation \ref{eq:bunch_j_saez}, identification of $\eps$ requires the counterfactual PDF of $y_0$, without measurement error.
Figure \ref{fig:friction_error:c} shows that the polynomial strategy fails to retrieve the PDF of $y_0$.
The PDF predicted by the polynomial regression does not integrate to one, and thus it is not a PDF.
If we divide the polynomial-based PDF in 
Figures \ref{fig:friction_error:b} and \ref{fig:friction_error:c}
by its integral, the PDF shifts up in the graphs.
The re-normalized PDF still misses the true $f_{y_0}$, and the underestimation of $B$ is larger than before. 

The polynomial strategy fails for two reasons:

\bigskip

\begin{enumerate}
 \item The PDF of $\ti{y}$ is not simply the PDF of ${y}$ plus the PDF of $e$ (Figure \ref{fig:friction_error:a}), but the convolution between the two PDFs. While $y_0$ and $e$ have uniform distributions, with a flat PDF,
 their convolution does not have a flat PDF. 
 As a result, extrapolating the polynomial to find the bunching mass and to predict the PDF of ${y}_0$ is misleading;

 \item The counterfactual distribution required for identification of the elasticity is the PDF of ${y}_0$, and not the PDF of $\ti{y}_0$ 
 (Equation \ref{eq:bunch_j_saez}).
Moreover, even if friction errors were not a problem, it is not possible to use the distribution of $y$ to back out the distribution of $y_0$ for values of $y_0$ inside $[k,k + (s_0-s_1) \eps]$. 
The shape of the distribution of $y_0$ is unidentified when $n^*$ falls in the bunching interval (Figure \ref{fig:imposs}).

\end{enumerate}

\newpage


\setcounter{figure}{0}
\renewcommand{\thefigure}{B.\arabic{figure}}

\begin{figure}[h!]
\caption{Counterexample where ``Polynomial Strategy'' Fails}
\label{fig:friction_error}
\centering

\begin{subfigure}[b]{0.49\textwidth}
\caption{\centering Distribution of Income with Friction Error}
\includegraphics[width=1\linewidth]{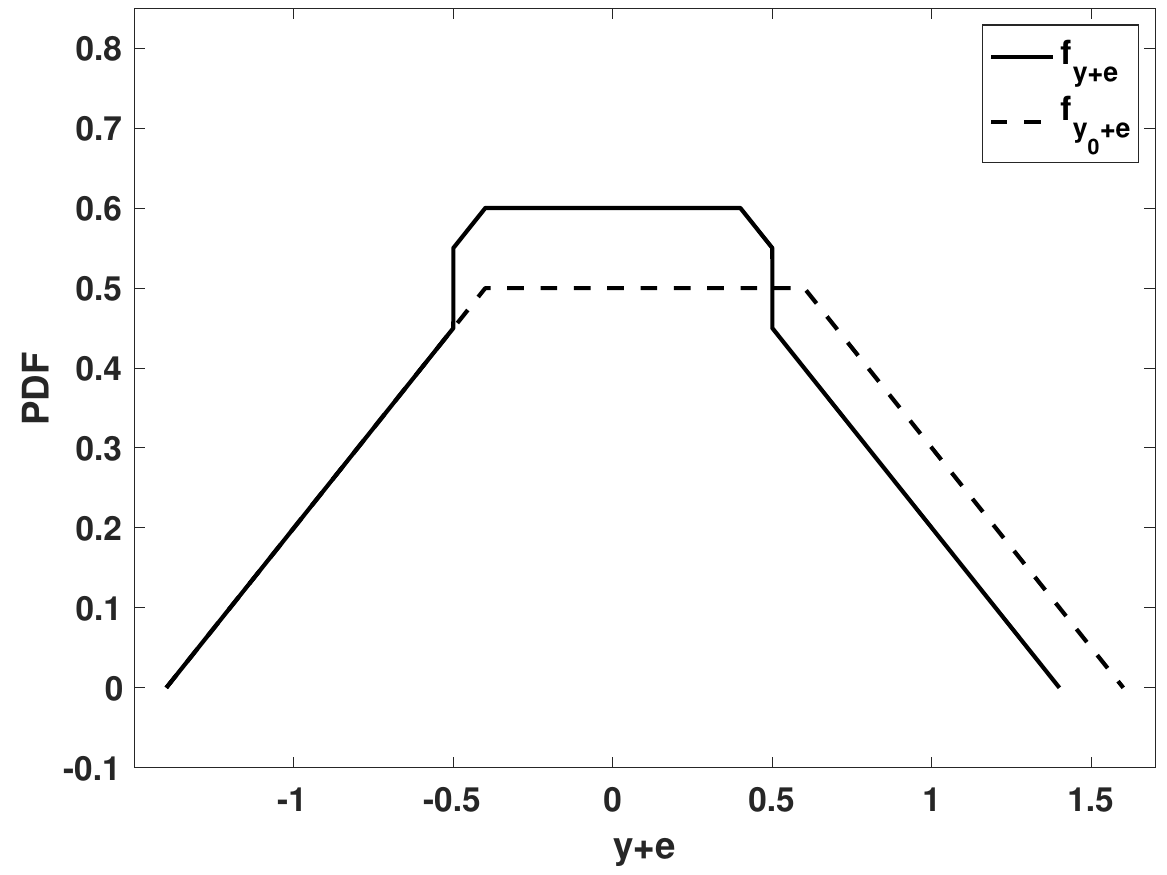}
\label{fig:friction_error:a}
\end{subfigure}
\hfill
\begin{subfigure}[b]{0.49\textwidth}
\centering
\caption{\centering Estimation of Bunching Mass}
\includegraphics[width=1\linewidth]{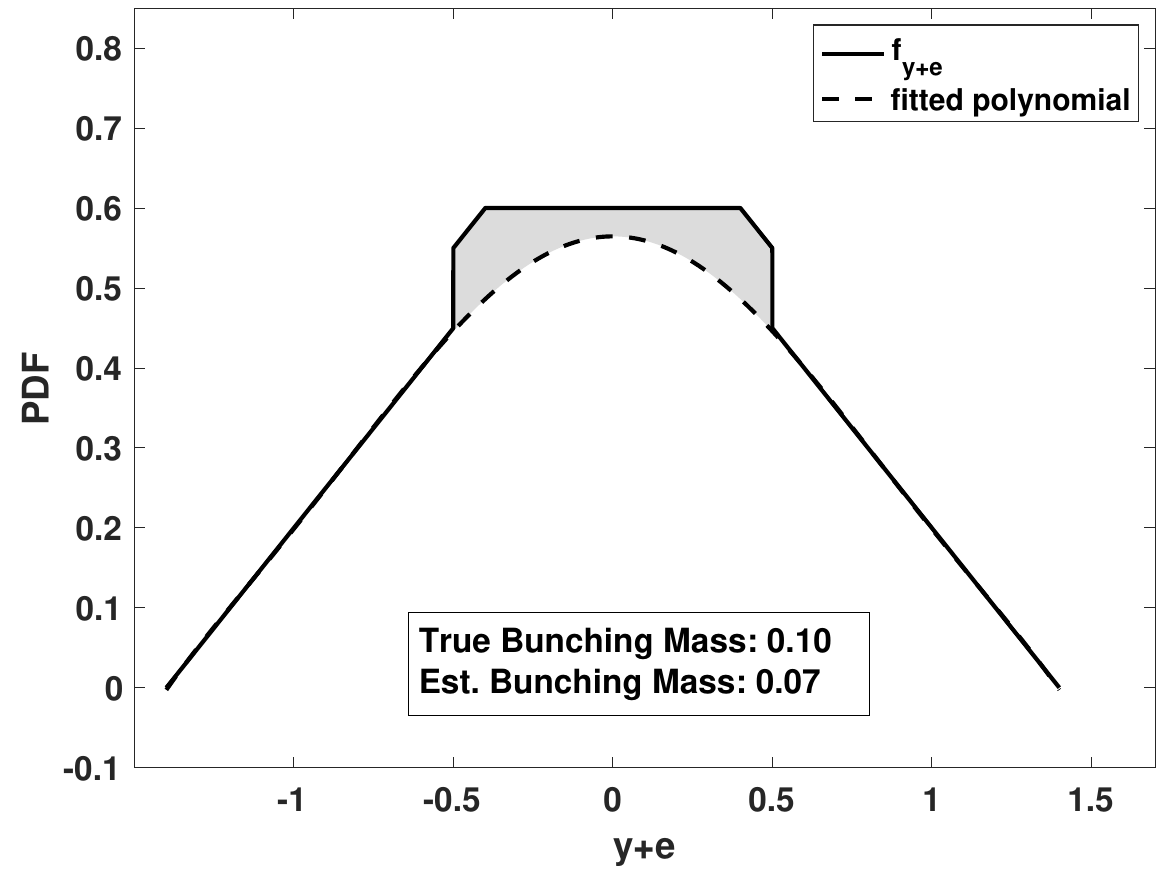}
\label{fig:friction_error:b}
\end{subfigure}

\medskip

\begin{subfigure}[b]{0.49\textwidth}
\centering
\caption{\centering Counterfactual Distribution of Income without Friction Error}
\includegraphics[width=1\linewidth]{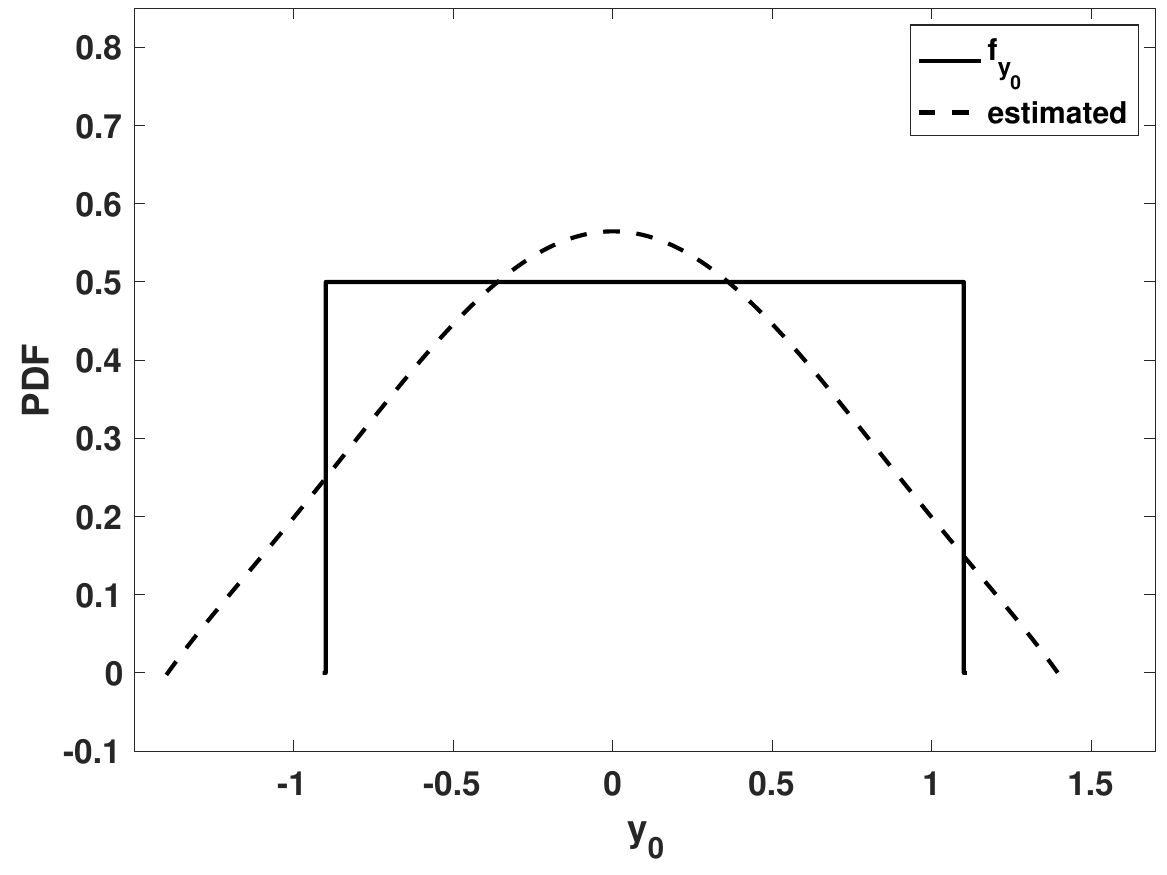}
\label{fig:friction_error:c}
\end{subfigure}
\hfill
\caption*{\footnotesize \textit{Notes:} 
The population model of this example has $\varepsilon=1.5$, 
$t_0=.2$, and $t_1=0.3$ at kink $k=0$.
The distribution of ability is assumed uniform, 
$n^* \sim U[-.565;1.435]$. 
The probability of bunching is equal to 10\%,
and the distribution of the friction error is
$e \sim U[-0.5;0.5]$. 
The researcher observes $\ti{y} = y +e$, 
where $y$ is a function of $n^*$, $\eps$, $t_0$, and $t_1$,
as described in Equation \ref{eq:logsol-onekink}.
Figure \ref{fig:friction_error:a} displays the PDF of $\ti{y}$ and 
$\ti{y}_0$.
Figure \ref{fig:friction_error:b} displays the fitted 7th-order polynomial to the PDF of $\ti{y}$ using observations in $(-\infty, -0.5) \cup (0.5, \infty)$.
The bunching mass is estimated by the integral of the difference between $f_{\ti{y}}$ and the fitted polynomial, inside the excluded region.
The polynomial strategy understimates the true bunching mass, and does not retrieve the PDF of $y_0$ (Figure \ref{fig:friction_error:c}).

\vspace{-29pt}} 
\end{figure}

\subsection{Examples of Identifying Restrictions on the Family of Unobserved Distributions} 
\label{sec:app:restriction_fn}

\indent 

This section brings three examples of restrictions imposed on $\m{F}_{n^*}$ that allows for point identification of the elasticity in the case of kinks. 

\begin{example} 
\label{example1}
\cite{Saez2010} implicitly restricts $\m{F}_{n^*}$ when using a trapezoidal approximation to solve the integral in Equation \ref{eq:bunch_j_saez}, in levels rather than in logs
(Saez's Equation 4 on page 186). That is,
\begin{gather}
B =\int_{K}^{K + \Delta Y} f_{Y_0} \left(u \right) ~ du 
\cong
\left(\frac{ f_{Y_0}( K + \Delta Y ) + f_{Y_0}( K ) }{2}\right)
\Delta Y,
\label{eq:bunch_j_saez_trap}
\end{gather}
where $\Delta Y = K \left[  ((1-t_0)/(1-t_1))^{\eps} - 1 \right]$.
A sufficient condition for the approximation to be true is to assume  $f_{Y_0}(u)$
is an affine function of $u$ for values of $u \in \left[ K, K + \Delta Y \right]$.
Given that $Y_0 = N^* (1-t_0)^{\eps}$, the PDF $f_{N^*} (u) = f_{Y_0}(u(1-t_0)^{\eps})(1-t_0)^{\eps}$ is restricted to be an affine function of 
$u$ inside the interval 
$\left[K (1-t_0)^{-\varepsilon}, K (1-t_1)^{-\varepsilon} \right]$.
This is equivalent to restricting $f_{n^*}$
to have an exponential shape within $[k - \eps s_0, k - \eps s_1]$.

The rest of Saez's identification strategy uses the fact that $f_{Y_0}(K)=f_Y(K^-)$, 
and  $f_{Y_0}(K + \Delta Y) = f_Y(K^+)((1-t_1)/(1-t_0))^{\eps}$,
where $f_Y$ is the PDF of the continuous portion of the distribution of $Y$, $f_Y(K^+)$  denotes the side limit
$\lim_{u \downarrow K} f_Y(u)$, and $f_Y(K^-)$ denotes $\lim_{u \uparrow K} f_Y(u)$. 
Substituting these into Equation \ref{eq:bunch_j_saez_trap},
\begin{gather}
B 
\cong
\frac{1}{2}\left(  
f_Y(K^+) \left( \frac{ 1-t_1 } { 1-t_0 } \right)^{\eps}  
+ f_Y(K^-) 
\right) 
K 
\left[\left( \frac{ 1-t_0 } { 1-t_1 } \right)^{\eps} -1 \right],
\label{eq:bunch_j_saez_eps}
\end{gather}
which is Equation 5 by \cite{Saez2010}.
It is then possible to solve implicitly for $\eps$ as a function
of the side limits of $f_Y$, the tax rates, the kink point, and the bunching mass. 
\end{example}

\begin{example}
\label{example2}
The derivation by \cite{chetty2011} of Equation 6 on page 761 
assumes that the PDF 
$f_{Y_0}$ is constant inside the bunching interval $[K,K+\Delta Y]$.
This is equivalent to assuming that $N^*$ is uniformly distributed in that region and thus restricts the class $\m{F}_{n^*}$.
For some scalar $a$, assume $F_{Y_0}(u)=a+f_{Y_0}(K) u$
for $u \in [K,K+ \Delta Y]$, so that the PDF of $Y_0$ is constant and
equal to $f_{Y_0}(K)$ in the bunching interval.
Then,
\begin{align}
B = & \int^{K + \Delta Y}_{K} f_{Y_0}(u) ~ du 
= F_{Y_0}(K + \Delta Y) - F_{Y_0}(K)
\nonumber
\\
= & f_{Y_0}(K) \Delta Y 
= f_{Y_0}(K)
K 
\left[\left( \frac{ 1-t_0 } { 1-t_1 } \right)^{\eps} -1 \right]
\nonumber
\\
\cong & f_{Y_0}(K)
K
\eps \ln \left( \frac{ 1-t_0 } { 1-t_1 } \right)
\nonumber
\\
\eps \cong & \frac{B/f_{Y_0}(K)}{K \ln \left( \frac{ 1-t_0 } { 1-t_1 } \right)},
\label{eq:bunch_j_chetty}
\end{align}
where the second to last approximate equality uses  
$[ ( 1-t_0 )/ ( 1-t_1 ) ]^{\eps} -1 $
 $\cong \ln  [ ( 1-t_0 )/ ( 1-t_1 )  ]^{\eps}$
 for small tax changes;
 and the last approximate equality is 
Equation 6 by \cite{chetty2011}. 
The rest of their identification procedure relies on the polynomial strategy to obtain  $B$ and $f_{Y_0}(K)$, as described in Section  \ref{sec:app:friction_error} of the supplement.

The constant PDF assumption on $f_{Y_0}$ is more restrictive than the affine PDF assumption that justifies Saez's trapezoidal approximation.
The trapezoidal approximation allows for 
$f_{Y_0}$ to have a non-zero slope in the bunching interval, 
whereas the constant PDF assumption does not.
\end{example}

The next example verifies Assumption \ref{aspt:cond_inv} from the main text in the class of Gaussian distributions.

\begin{example}
Consider the family of normal distributions with unknown mean and variance, that is,
$\mathcal{F}_{n^*} = \left\{ G_{n^*}(n;\mu,\sigma^2) = \Phi\left( \frac{n - \mu}{\sigma} \right),~ \mu \in \mmr,~ \sigma^2 \in \mmr_+ \right\}$,
where $\Phi$ denotes the standard normal CDF.
Pick an arbitrary ${F}_{n^*}(n) = \Phi\left( \frac{n - \mu^*}{\sigma^*} \right)$.
Equation \ref{eq:bunch_saez_param1} says that 
$\Phi\left( \frac{u - e s_0 - \mu}{\sigma} \right)
=
\Phi\left( \frac{u - \eps s_0 - \mu^*}{ \sigma^*} \right)$
 for  $\forall u <k$, where the right-hand side (RHS) is a known function of $u$ for $u<k$ (because that equals $F_y(u)$).
 Taking the inverse $\Phi^{-1}$ from both sides of the equation gives 
$  (u - e s_0 - \mu)/\sigma
=
  (u - \eps s_0 -   \mu^*)/\sigma^* $,
  where both sides are affine functions of $u<k$ and the RHS is a known function of $u$.
  They must have the same slope,  $1/\sigma = 1/\sigma^*$, and the same intercept, $-(e s_0 + \mu)/\sigma = -(\eps s_0 + \mu^*)/\sigma^*$.
  Apply the same steps to Equation \ref{eq:bunch_saez_param2} to obtain $1/\sigma = 1/\sigma^*$ and  $-(e s_1 + \mu)/\sigma = -(\eps s_1 + \mu^*)/\sigma^*$.
  These equations solve uniquely for $\sigma=\sigma^*$, $\mu=\mu^*$, and $e=\eps$.
 Therefore, the family of normal distributions with unknown mean and variance satisfies Assumption \ref{aspt:cond_inv}.
\end{example}

\subsection{Implementation of Censored Quantile Regressions}
\label{sec:supp:clad:implement}

\indent 

The optimization problem in Equation \ref{eq:opt_quan_min_max} is computationally difficult.
For the left (or right) censored case, \cite{chernozhukov2002} proposed a fast and practical estimator that consists of three steps.
First, you fit a flexible Probit model that explains the probability of no censoring; 
then, you select observations whose values of $X$ lead to a predicted probability of no censoring that is greater than $1- \tau$.
Second, you fit a quantile regression of $y$ on $X$ using the selected observations in the first step;
then, you select observations whose values of $X$ lead to a predicted quantile that is greater than $k$.
Third, repeat the second step using the observations selected at the end of the second step.
\cite{chernozhukov2002} demonstrate consistency and asymptotic normality of their three-step estimator.
Moreover, they show that the standard errors computed by the quantile regression in the third step are valid.

Our case of middle censoring requires a straightforward modification of the method proposed by \cite{chernozhukov2002}. Inspired by their algorithm, we propose the following implementation steps.
\begin{enumerate}
 \item Create dummies $\delta_i^- = \mmi\{y_i < k \}$ (not censored, left of $k$)
 and $\delta_i^+ = \mmi\{y_i > k \}$ (not censored, right of $k$).
 Fit two Probit models to estimate $\mmp[\delta_i^+ | X_i] = \Phi(X_i g^+)$
 and $\mmp[\delta_i^- | X_i] = \Phi(X_i g^-)$, where
 $\Phi$ denotes the cdf of a standard normal distribution,
 and $g^\pm$ are vectors of parameters.
 You may use powers and interactions of $X_i$ to make this stage as flexible as possible.
 Select two subsamples as follows.
 Compute the 10th quantile of the empirical distribution of 
 $\Phi(X_i \hat{g}^+) - (1-\tau)$
 conditional on 
 $\Phi(X_i \hat{g}^+) > 1-\tau$.
 Let $\kappa_0^+(\tau)$ be the 10th quantile of that distribution.
 The first subsample is
 $J_0^+(\tau) = \{i: \Phi(X_i \hat{g}^+)>1-\tau+\kappa_0^+(\tau) \}$.
 The second subsample is 
 $J_0^-(\tau) = \{i: \Phi(X_i \hat{g}^-)>\tau+\kappa_0^-(\tau) \}$,
 where $\kappa_0^-(\tau)$ is the 10th quantile of 
 the empirical distribution of $\Phi(X_i \hat{g}^-) - \tau$
 conditional on 
 $\Phi(X_i \hat{g}^-)>\tau$.
 Create a dummy $W_i^0 = \mmi\{i \in J_0^+(\tau) \}$.
 
 \item Fit the quantile regression model $Q_{\tau}(y_i|X_i,W_i^0) = X_i b(\tau) + W_i^0 \delta(\tau)$ using observations in $J_0^-(\tau) \cup J_0^+(\tau)$.
 Use the estimates of this quantile regression, that is $\hat{b}^0(\tau)$
 and $\hat{\delta}^0(\tau)$, to create two subsamples as follows.
 The first subsample is
 $J_1^+(\tau) = \{i: X_i \hat{b}^0(\tau) + \hat{\delta}^0(\tau) > k +\kappa_1^+(\tau) \}$,
 where $\kappa_1^+(\tau)$ is the 3rd quantile of the empirical distribution
 of $X_i \hat{b}^0(\tau) + \hat{\delta}^0(\tau) - k$ conditional on $X_i \hat{b}^0(\tau) + \hat{\delta}^0(\tau) > k$.
 The second subsample is 
 $J_1^-(\tau) = \{i: X_i \hat{b}^0(\tau) < k + \kappa_1^-(\tau) \}$,
 where $\kappa_1^-(\tau)$ is the 97th quantile of the empirical distribution
 of $X_i \hat{b}^0(\tau) - k$ conditional on $X_i \hat{b}^0(\tau) < k$.
 Create a dummy $W_i^1 = \mmi\{i \in J_1^+(\tau) \}$.
 
 \item Fit the quantile regression model $Q_{\tau}(y_i|X_i,W_i^1) = X_i b(\tau) + W_i^1 \delta(\tau)$ using observations in $J_1^-(\tau) \cup J_1^+(\tau)$ to obtain estimates $\hat{b}^1(\tau)$
 and $\hat{\delta}^1(\tau)$.
 The elasticity estimator is $\hat\eps = \hat\delta^1(\tau)/(s_1-s_0)$.

 \end{enumerate}
 
\clearpage
\subsection{Estimates with the Filtering Method of Saez (2010)} 
\label{sec:supp:saez_filter}

\indent 

In this section, we recompute the estimates of Table \ref{tbl:tobit} using a different filtering method.
Specifically, we employ the procedure used by \cite{Saez2010} to
obtain the bunching mass and the side limits of the distribution of income without friction error $Y$.
The procedure implicitly defines a way to estimate the unobserved distribution of $Y$ given the observed distribution of income with friction error $\tilde{Y}$.
We refer the reader to Figure 2 by \cite{Saez2010}.

The first step is to construct a histogram-based estimate 
of the PDF $f_{\tilde Y}$, and then average $f_{\tilde Y}$ for 
$\tilde Y \in [K-2\delta , K-\delta] \cup [K+\delta, K+2\delta]$,
where $K=8,580$ is the kink point, and $\delta = 1,500$ defines the excluded region.
Call that average $\bar{f}$.
The bunching mass is estimated by the area between two curves, 
$f_{\tilde Y}$ and $\bar{f}$.
The continuous portion of $f_Y$ equals $f_{\tilde Y}$, except for the excluded region $[K-\delta , K+\delta]$, where
$f_Y$ equals $\bar f$. 
We obtain the CDFs $F_{Y}$ and $F_{\tilde Y}$ from their PDF estimates.
Finally, we rely on $Y=F_{Y}\left(F_{\tilde Y}^{-1} (\tilde Y) \right)$ to transform $\tilde Y$ into $Y$.
Estimates are reported in Table \ref{tbl:apptobit} below.\footnote{
 Saez's filtering procedure is a particular case of the so-called ``polynomial strategy'' used by \cite{chetty2011}:
 Saez fits a flat line but the ``polynomial strategy'' allows for higher-order polynomials. 
 We discuss the problems of the ``polynomial strategy'' in Section \ref{sec:app:friction_error} of this supplement. 
 Although our proposed filtering strategy is not a general solution to the filtering problem, it does require less stringent conditions than Saez's filter
 (see Section \ref{sec:application:data} for details).
 Therefore, differences in estimates between Tables \ref{tbl:tobit} and \ref{tbl:apptobit} are less concerning than differences that could arise by using a more general filtering procedure than ours. 
} 

\setcounter{table}{0}
\renewcommand{\thetable}{B.\arabic{table}}
\clearpage
\begin{landscape}

\begin{table}
	\caption{Estimates Using U.S. Tax Returns 1995--2004}
\label{tbl:apptobit}
	\centering
	\scalebox{0.95}{
\begin{tabular}{lccccccc|lc}
\hline\hline
 & (1) & (2) & (3) & (4) & (5) & (6) & (7) & \multicolumn{2}{c}{(8)} \\ 
Statistical Model & Trapezoidal & Theorem \ref{theo_partial} & Theorem \ref{theo_partial} & Tobit & Tobit & Tobit & Tobit & \\ 
 & Approximation & Bounds & Bounds & Full Sample & Trunc. 75\% & Trunc. 50\% & Trunc. 25\% & \multicolumn{2}{c}{Sample} \\ 
 & & M = 0.5 & M = 1 & & & & & \multicolumn{2}{c}{details} \\ \hline 
\textit{All} & & & & & & & & Obs. & 188.3m \\ 
\hspace{2mm} Elasticity $\left(\varepsilon\right)$ & 0.226 & $\left[ 0.225, 0.250 \right]$ & $\left[0.211, 0.283 \right]$ & 0.120 & 0.177 & 0.182 & 0.200 & Avg. & \$53.5k \\ 
 & & & & (0.0001) & (0.0001) & (0.0001) & (0.0002) & Std. & \$64.6k \\ 
 & & & & & & & & \\ 
 \textit{Self-employed} & & & & & & & & Obs. & 33.4m \\ 
\hspace{2mm} Elasticity $\left(\varepsilon\right)$ & 0.934 & $\left[0.686,1.183\right]$ & $\left[0.612,\infty \right]$ & 0.610 & 0.809 & 0.805 & 0.825 & Avg. & \$60.7k \\ 
 &  & & & (0.0005) & (0.0007) & (0.0008) & (0.0008) & Std. & \$77.2k \\ 
 \textit{Self-employed,} & & & & & & & & & \\ 
 \textit{married} & & & & & & & & Obs. & 23.9m \\ 
\hspace{2mm} Elasticity $\left(\varepsilon\right)$ & 0.389 & $\left[0.338, 0.530 \right]$ & $\left[0.328, 0.730\right]$ & 0.191 & 0.287 & 0.330 & 0.331 & Avg. & \$73.6k \\ 
 &  & & & (0.0004) & (0.0007) & (0.0008) & (0.0008) & Std. & \$84.4k \\ 
 \textit{Self-employed,} & & & & & & & & \\ 
\textit{not married} & & & & & & & & Obs. & 9.5m \\ 
\hspace{2mm} Elasticity $\left(\varepsilon\right)$ & 1.350 & $\left[1.170, 1.787\right]$ & $\left[1.047, \infty \right]$ & 1.274 & 1.258 & 1.171 & 1.254 & Avg. & \$28.3k \\ 
 &  & & & (0.0013) & (0.0014) & (0.0015) & (0.0017) & Std. & \$39.6k \\ 
 & & & & & & & & \\ \hline
\end{tabular}
}
\exhibitnote{\textit{Notes:}
The table shows estimates of the elasticity for four different subsamples of the IRS data that was separately filtered according the procedure used by \cite{Saez2010}.  
We use three estimation approaches. 
The first approach (column 1) uses the trapezoidal approximation to point-identify the elasticity (Example \ref{example1}).
The second approach (columns 2 and 3) 
computes partially identified sets for the elasticity (Theorem \ref{theo_partial}),
using non-parametric estimates of the side limits of $f_y$ at the kink, and the bunching mass.
Side limits were estimated using the method of \cite{cattaneo2019}.
The estimate for the bunching mass equals the sample proportion of $y$ observations that equals the kink point (see discussion in Section \ref{sec:supp:saez_filter} on friction errors).
Upper and lower bounds are calculated for two choices of M, 
that is, the maximum slope of the PDF of the unobserved heterogeneity $n^*$.
Column 4 has Tobit MLE estimates of the elasticity that utilizes the full sample of data, along with robust standard errors in parentheses. 
Columns 5 through 7 report truncated Tobit MLE estimates.
As we move from column 5 to column 7, we restrict the estimation sample to shrinking symmetric windows around the kink that utilizes 75\% to 25\% of the data.
The set of covariates that enters the Tobit estimation is kept constant across different truncation windows and are listed in Section \ref{sec:estimates_across_methods}.}
\end{table}

\end{landscape}

\newpage

\subsection{Robustness of Tobit Estimates to Lack of Normality}
\label{sec:supp:tobit:robust}

\indent 

This section presents an additional experiment that showcases the robustness result of Lemma \ref{lemma:tobit_robust}.
Following the two experiments in Section \ref{sec:solutions:cov:tobit}, we call this Experiment 3.
Again, the goal is to illustrate that 
Lemma \ref{lemma:tobit_robust} does not require normality of $F_{n^*}$ or $F_{n^*|X}$ even without truncation.
For this experiment, the key parameters of Equation \ref{eq:logsol-onekink} are set at $\eps=4$, $k=2.0794$, $s_0=0.2624$, and $s_1=-0.1054$.
Figure \ref{fig:simulation3:a} plots the PDF of $n^*$, which is approximately uniformly distributed over $[0,8]$.
The distribution of scalar $X$ is discrete with 20 mass points (Figure \ref{fig:simulation3:b}) and is chosen such that
$F_{n^*}(n) = \mme\left[\Phi\left((n-X)/0.3251\right) \right]$
approximates the CDF of the uniform.
We solved numerically for non-normal conditional distributions of $n^*$ given $X$ 
that satisfy Equation \ref{aspt:tobit2}.
Figure \ref{fig:simulation3:d} displays the true PDFs $f_{n^*|X=x}$ in black and 
the normal PDFs $g_{n^*|X=x}$ assumed by the Tobit in gray, for all values of $x$.
We clearly see that $f_{n^*|X}$ is not normal.
We then generate 50,000 observations of $(y,X)$ and fit our Tobit model with the covariate $X$ to the entire sample.
The Tobit model fits the distribution of $y$ (Figure \ref{fig:simulation3:c}) and estimates the elasticity very closely to the truth 
($\ha\eps = 4.0008$, S.E. 0.0158).

\newgeometry{headheight=1mm,tmargin=10mm,bmargin=15mm, lmargin=10mm,rmargin=10mm}
\begin{landscape}

\begin{figure}[tbp]
\caption{Robustness of Tobit Estimates to Lack of Normality\textemdash Experiment 3}
\label{fig:simulation3}
\centering

\begin{minipage}{0.2\linewidth}
\vspace{6mm}
\centering

\begin{subfigure}[b]{\linewidth}
\centering
\caption{PDF of $n^*$}
\includegraphics[width=\linewidth]{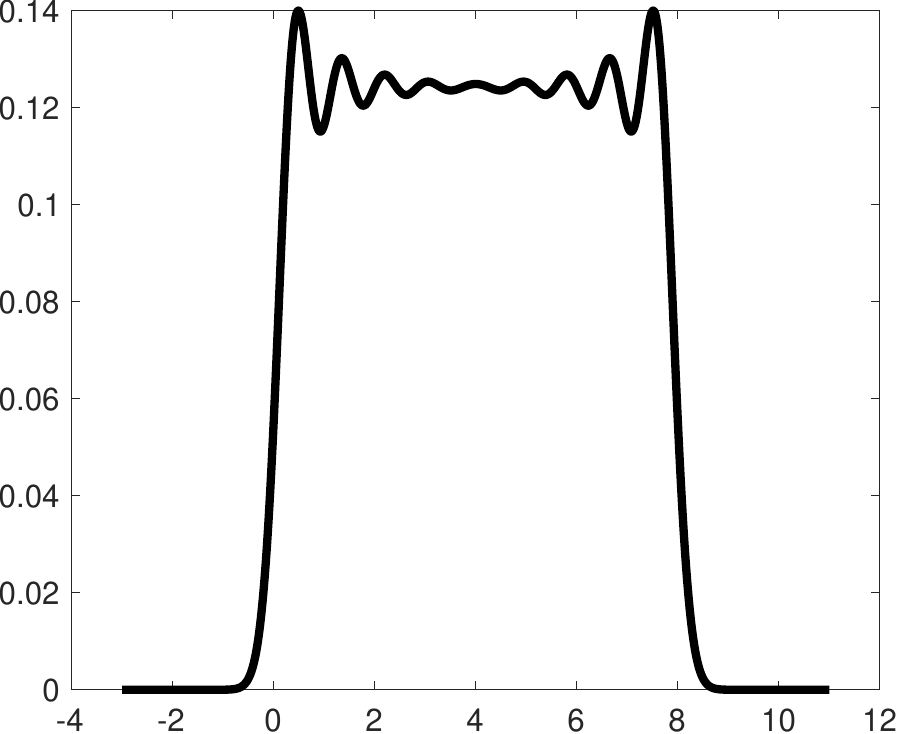}
\label{fig:simulation3:a}
\end{subfigure}

\begin{subfigure}[b]{\linewidth}
\centering
\caption{PMF of $X$}
\includegraphics[width=\linewidth]{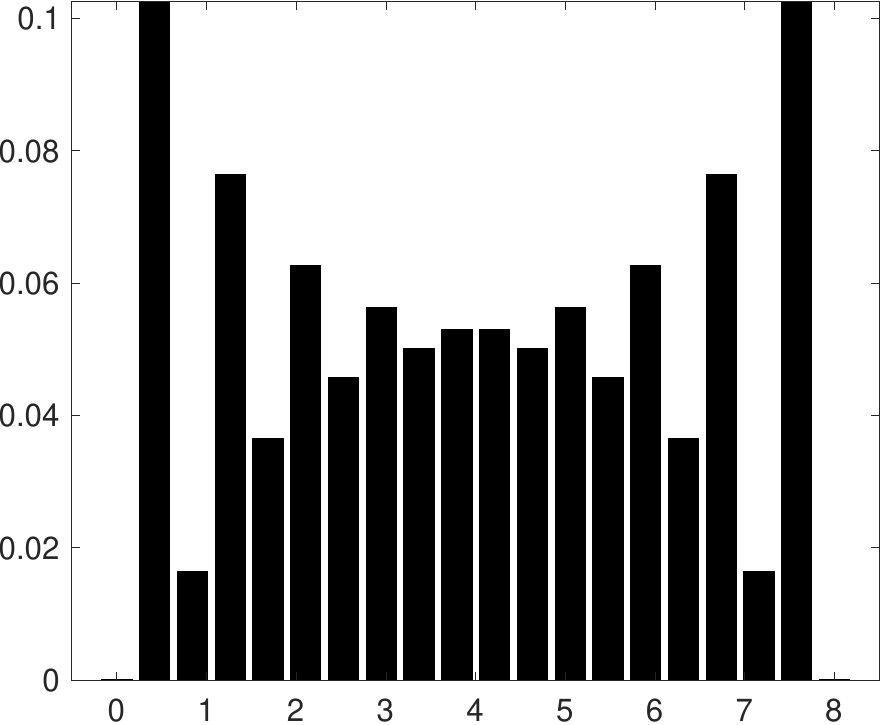}
\label{fig:simulation3:b}
\end{subfigure}

\begin{subfigure}[b]{\linewidth}
\centering
\caption{100\% of the data used}
\includegraphics[width=\linewidth]{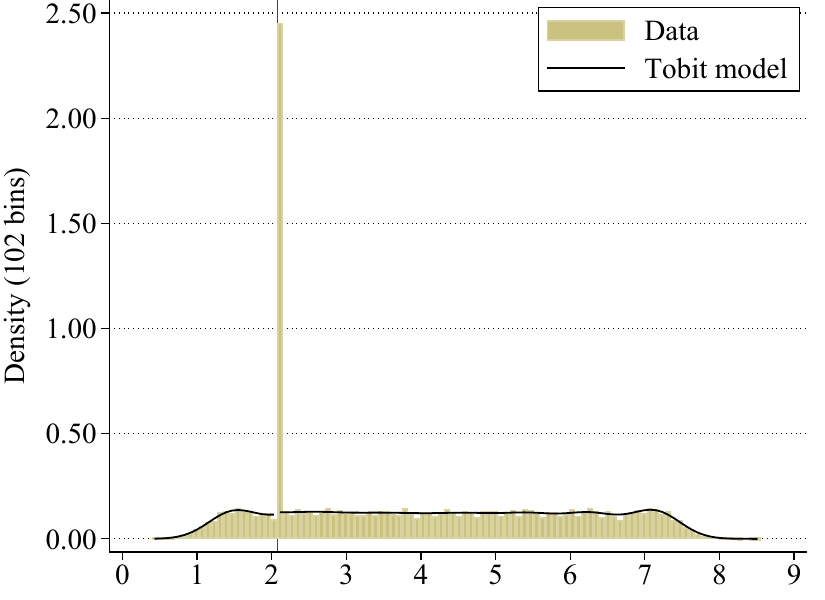}
\label{fig:simulation3:c}
\end{subfigure}

\end{minipage}
\hspace{5mm}
\begin{minipage}{0.67\linewidth}
\centering

\begin{subfigure}[c]{\linewidth}
\centering
\caption{Conditional Probability Density Functions $n^* | X=x$}
\includegraphics[width=\linewidth]{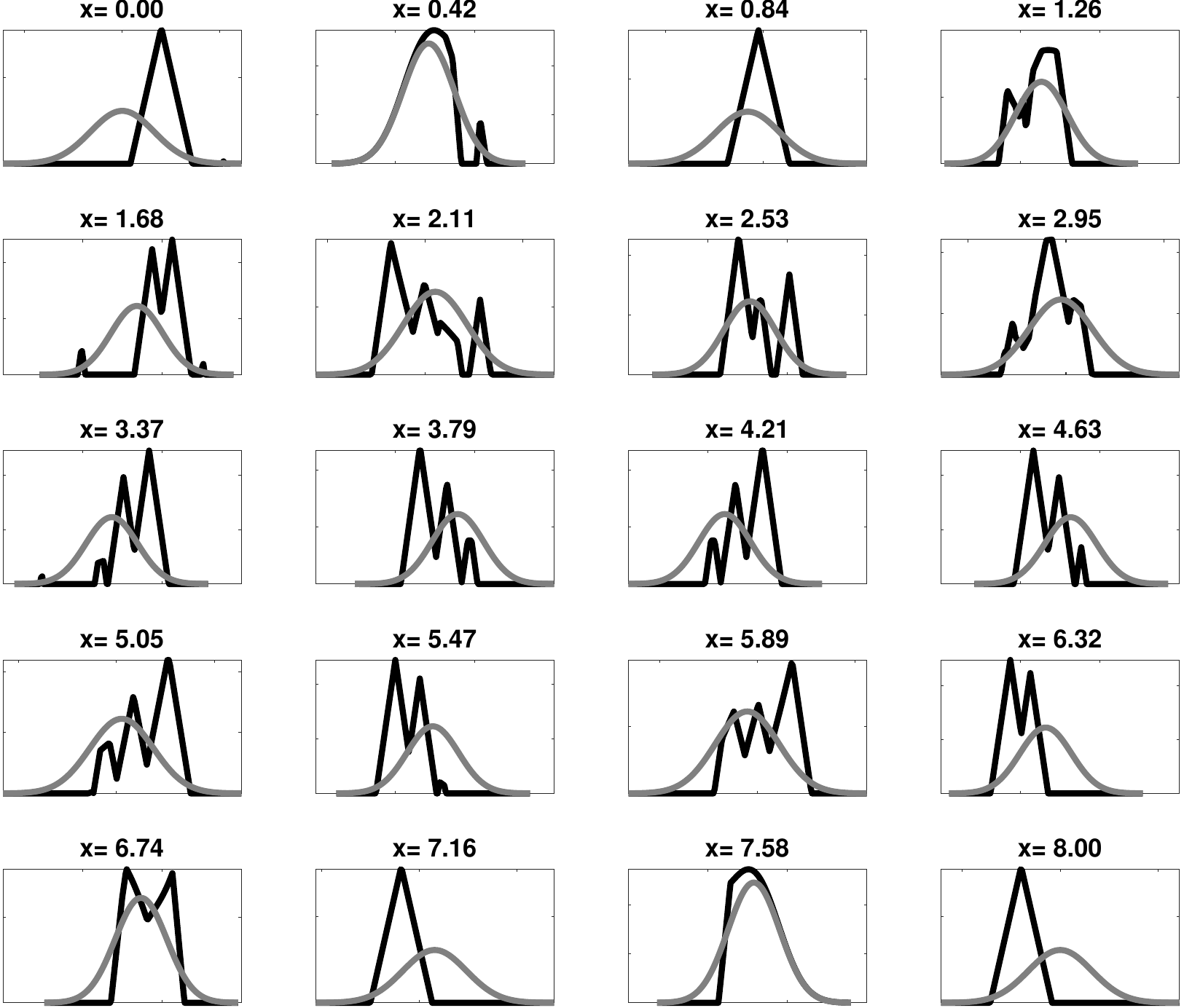}
\label{fig:simulation3:d}
\end{subfigure}

\end{minipage}

\caption*{
\footnotesize 
\textit{Notes:} 
This simulation experiment illustrates that the mid-censored Tobit model is able to fit non-normal distributions of $n^*$ and retrieve the right elasticity even when the conditional distribution $n^*|X$ is not Gaussian and there is no truncation.
We generate 50,000 observations of $y$ and a scalar $X$ following Experiment 3 described in Section \ref{sec:supp:tobit:robust}. 
The variable $n^*$ is approximately a uniform distribution over $[0,8]$ (Panel a), the distribution of $X$ is discrete (Panel b), 
the kink point is at $k=2.0794$, and $\eps=4$.
Panel c shows the histogram of simulated data for $y$ and the best-fit Tobit distribution using covariate $X$ and no truncation.
Panel d displays the true conditional PDFs of $n^*|X$ in black along with the Gaussian PDFs in gray that are assumed by the Tobit model.
The elasticity is estimated at $\ha\eps = 4.0008$ (S.E. 0.0158).
}
\end{figure}
\end{landscape}
\restoregeometry

\subsection{Graphical Analysis} 
\label{sec:app:figures}

\indent 

This section contains figures illustrating the solution of the utility maximization problem of Section \ref{sec:model:sol} and the observed distribution of income
in the case of kinks and notches.

\begin{figure}[ht!]
\caption{Budget Constraints and Distributions}
\label{fig:graphanalysis}
\centering

\begin{subfigure}[b]{0.49\textwidth}
\caption{\centering Budget Constraint\textemdash Kink}
\includegraphics[width=\linewidth,trim={3.1cm 1.2cm 7.5cm 1.7cm},clip]{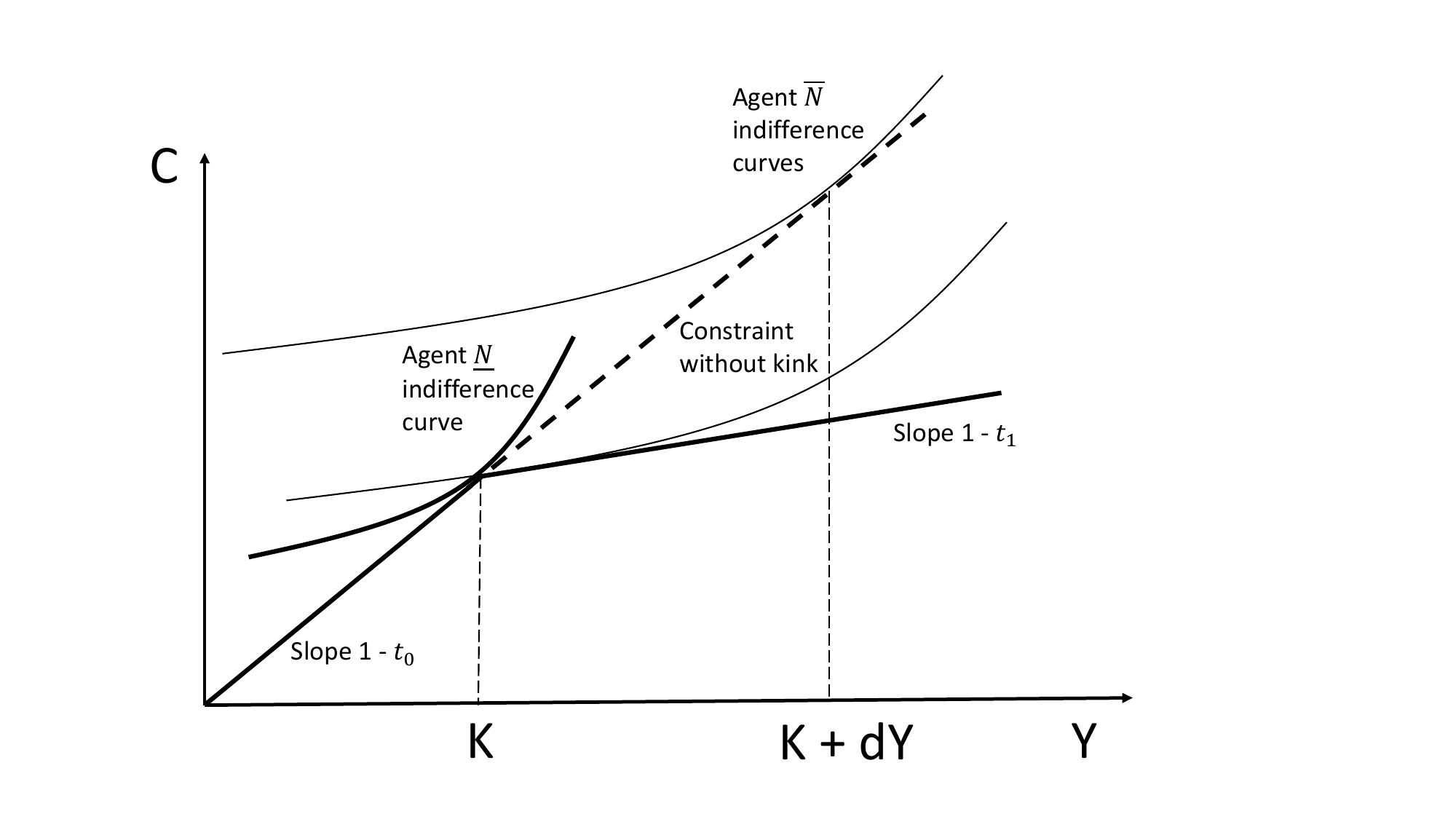}
\label{fig:graphanalysis:a}
\end{subfigure}
\hfill
\begin{subfigure}[b]{0.49\textwidth}
\centering
\caption{\centering Distribution\textemdash Kink}
\includegraphics[width=\linewidth,trim={3.1cm 1.2cm 7.5cm 1.7cm},clip]{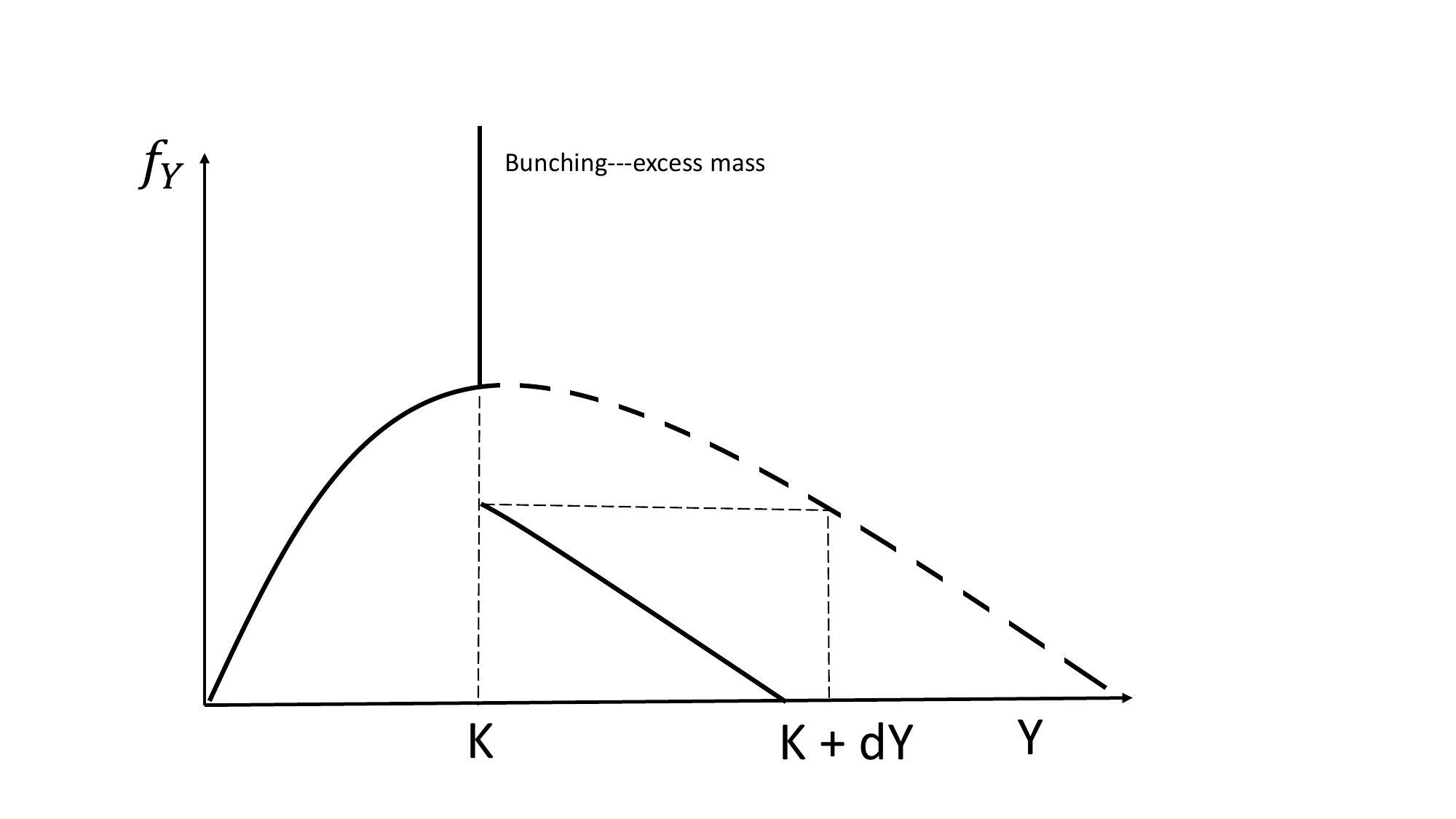}
\label{fig:graphanalysis:b}
\end{subfigure}

\medskip

\begin{subfigure}[b]{0.49\textwidth}
\centering
\caption{\centering Budget Constraint\textemdash Notch}
\includegraphics[width=\linewidth,trim={3.1cm 1.2cm 7.5cm 1.7cm},clip]{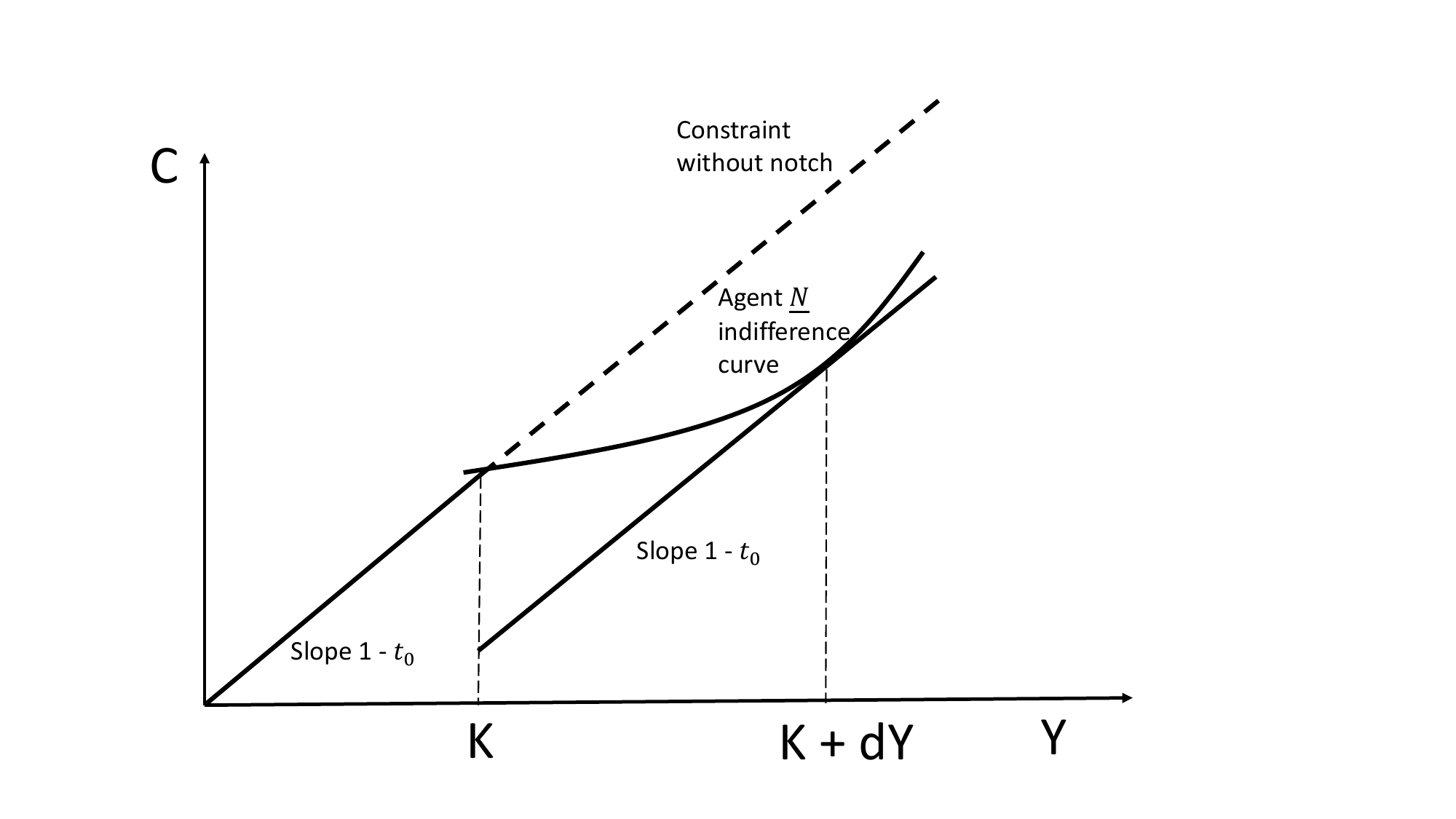}
\label{fig:graphanalysis:c}
\end{subfigure}
\hfill
\begin{subfigure}[b]{0.49\textwidth}
\centering
\caption{\centering Distribution\textemdash Notch}
\includegraphics[width=\linewidth,trim={3.1cm 1.2cm 7.5cm 1.7cm},clip]{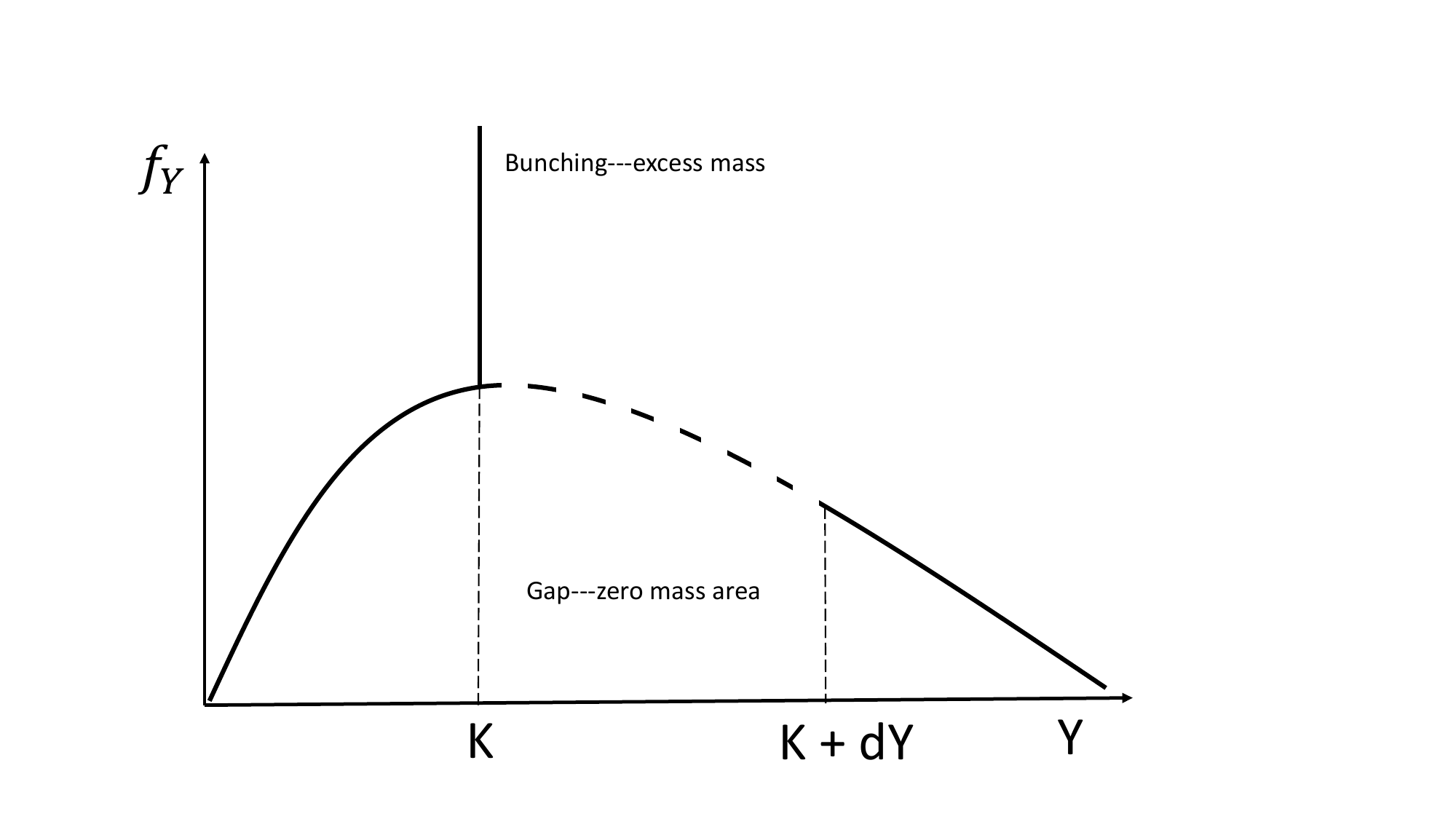}
\label{fig:graphanalysis:d}
\end{subfigure}

\caption*{\footnotesize \textit{Notes:} 
Figure \ref{fig:graphanalysis:a} shows the budget for constraint consumption $C$ and after-tax income $Y$ in the case of a kink. 
The budget line contains a kink at the point $Y = K,$ where the slope is $(1 - t_0)$ for $Y< K$ and $(1 - t_1)$ for $Y > K,$ because the marginal tax rate increases at the kink.
The dashed line is the budget line that would prevail if there were no tax change at the kink.
Agents of type $N^*$ who have indifference curve that is tangent to the budget line without the kink between $K$ and $K+ dY$ optimally choose $Y=K$ in the presence of the kink. 
In Figure \ref{fig:graphanalysis:b}, the solid line for $Y<K$ and the dashed line for $Y\geq K$ together depict the PDF of $Y$ in the absence of the kink.
In contrast, the solid line for $Y\leq K$ and the solid line for $Y > K$ depict the PDF in the presence of the kink, where the vertical line at $K$ denotes the bunching mass.
Figure \ref{fig:graphanalysis:c} displays the budget constraint for the case of a notch.
The budget line contains a notch at the point $Y= K,$ where there is a jump-down discontinuity
and the slope is equal to $1 - t_0$ on both sides of $K.$
The dashed line represents the budget line that would prevail in the absence of the notch. 
Agents of type $N^*$ who have indifference curve that is tangent to the budget line between $K$ and $K + dY$ optimally choose $Y=K$ in the presence of the notch.
In Figure \ref{fig:graphanalysis:d}, the solid line for $Y<K$, the dashed line for $Y \in [K, K+dY]$, and the solid line for $Y>K+dY$ together depict the PDF of $Y$ in the absence of the notch.
In contrast, the solid line for $Y\leq K$ and the solid line for $Y>K+dY$ represent the PDF in the presence of the notch, where the vertical line at $K$ denotes the bunching mass.

\vspace{-29pt}} 
\end{figure}
\clearpage

\subsection{Filtering Procedure} 
\label{sec:supp:filter}

\indent 

The filtering procedure described in Section \ref{sec:application:data}
recovers the CDF of income without optimizing frictions under the following conditions:
 \begin{enumerate}
 \item optimizing frictions only affect bunching individuals additively, that is, $\ti{y}=y+\mathbb{I}\{y=k\} e$,
 where $e$ is the optimizing friction random variable;
 
 \item $e$ is independent of $n^*$ and the support of $e$ is a closed interval $[-\delta_-,\delta_+]$, 
 where $\delta_->0$ and $\delta_+>0$ are known by the researcher;
 
 \item $F_y$ is a polynomial of order $p$ in an interval around the kink, $[k-l,k+u]$, with a change in intercept at $y=k$;
 that interval is bigger than the support of frictions, i.e., 
 $ [k-\delta_-,k+\delta_+] \subseteq [k-l,k+u] $ and 
 $ [k-\delta_-,k+\delta_+] \neq  [k-l,k+u] $; 
 the constants $l$, $u$, and $p$ are known by the researcher.
 
 \end{enumerate}
 
 The CDF of $\ti{y}$ can be written in terms of the CDF of $y$ and $e$: 
 \begin{align*}
 F_{\ti{y}}(v) = & \mmp[\ti{y}\leq v|y=k] \mmp[y=k] + \mmp[\ti{y}\leq v|y \neq k] \mmp[y \neq k]
 \\
 = & \mmp[k+e\leq v|n^* \in [\underline{n},\overline{n}] ]B + \mmp[y \leq v| y \neq k] (1-B)
 \\
 = & F_e(v-k) B + F_y(v) - \mmi\{v \geq k \}B. 
 \end{align*}
 
 If $v<k-\delta_-$, then $F_e(v-k)=0$ and $F_{\ti{y}}(v)=F_{y}(v)$.
 If $v>k+\delta_+$, then $F_e(v-k)=1$ and $F_{\ti{y}}(v)=F_{y}(v)$.
 Our filtering procedure takes the empirical CDF of $\ti{y}$ for values of income inside $[k-l, k-\delta_-)$
 and
 $(k+\delta_+, k+u]$ and fits a polynomial of order $p$ with a change in intercept at $y=k$.

 As the number of observations grows large, the empirical CDF of $\ti{y}$ converges to the true CDF of $y$ 
 for values of income inside $[k-l, k-\delta_-) \cup (k+\delta_+, k+u]$, 
 and the fitted polynomial converges to the true polynomial that characterizes $F_y$ inside $[k-l,k+u]$.
 That polynomial is then extrapolated to retrieve $F_y(v)$ for every $v \in [k-l,k+u]$.
 For $v \not\in [k-l,k+u]$, the empirical CDF of $\ti{y}$ converges to the true CDF of $y$, and no polynomial fit is needed.

\end{document}